\pdfoutput=1
\documentclass[reqno,oneside,letterpaper, 12pt]{amsart}
\usepackage{mathrsfs}
\usepackage[foot]{amsaddr}
\usepackage{amssymb}
\usepackage{wasysym}
\usepackage{url}
\usepackage[utf8]{inputenc}
\usepackage{amsmath, amssymb,bm, cases, mathtools, thmtools}
\usepackage{verbatim}
\usepackage{graphicx}\graphicspath{{figures/}}
\usepackage{multicol}
\usepackage{tabularx}
\usepackage[usenames,dvipsnames]{xcolor}
\usepackage{mathrsfs}
\usepackage{url}
\usepackage[normalem]{ulem}
\usepackage{xstring}
\usepackage[shortlabels]{enumitem}

\usepackage[minnames=1,maxnames=99,maxcitenames=3,
    style=authoryear,
    doi=false,url=false,
    firstinits=false,hyperref,natbib,backend=bibtex]{biblatex}
\renewbibmacro{in:}{%
  \ifentrytype{article}{}{\printtext{\bibstring{in}\intitlepunct}}}
\bibliography{biblio}

\usepackage[colorlinks,citecolor=blue,urlcolor=blue,linkcolor=RawSienna]{hyperref}
\usepackage{hypernat}
\usepackage{datetime}

\DeclareMathAlphabet\EuRoman{U}{eur}{m}{n}
\SetMathAlphabet\EuRoman{bold}{U}{eur}{b}{n}

\usepackage{euscript,microtype}
\usepackage[capitalize]{cleveref}

\crefname{assumption}{Assumption}{Assumptions}
\crefname{claim}{Claim}{Claims}

\makeatletter
\let\reftagform@=\tagform@
\def\tagform@#1{\maketag@@@{\ignorespaces\textcolor{gray}{(#1)}\unskip\@@italiccorr}}
\renewcommand{\eqref}[1]{\textup{\reftagform@{\ref{#1}}}}
\makeatother

\definecolor{WowColor}{rgb}{.75,0,.75}
\definecolor{SubtleColor}{rgb}{0,0,.50}

\newcounter{margincounter}

\declaretheorem[style=plain,numberwithin=section,name=Theorem]{theorem}
\declaretheorem[style=plain,sibling=theorem,name=Lemma]{lemma}

\declaretheorem[style=plain,sibling=theorem,name=Claim]{claim}

\declaretheorem[style=definition,sibling=theorem,name=Definition]{definition}
\declaretheorem[style=definition,name=Assumption]{assumption}

\declaretheorem[style=remark,sibling=theorem,name=Remark]{remark}
\crefformat{conditionINNER}{#2(#1)#3}
\crefmultiformat{conditionINNER}
  {(#2#1#3)}
  { and~(#2#1#3)}
  {, (#2#1#3)}
  { and~(#2#1#3)}
\Crefformat{conditionINNER}{Condition~#2(#1)#3}
\Crefmultiformat{conditionINNER}
  {Conditions~(#2#1#3)}
  { and~(#2#1#3)}
  {, (#2#1#3)}
  { and~(#2#1#3)}

\declaretheoremstyle[
    spaceabove=-6pt,
    spacebelow=6pt,
    headfont=\normalfont\bfseries,
    bodyfont = \normalfont,
    postheadspace=1em,
    qed=$\square$,
    headpunct={{}}]{myproofstyle}

\numberwithin{equation}{section}
\numberwithin{theorem}{section}

\usepackage{amssymb}
\usepackage{mathrsfs}
\usepackage{leftidx}

\usepackage{setspace}
\usepackage{thmtools}
\usepackage{geometry}
\usepackage{enumitem}
\usepackage{caption}
\usepackage{float}
\usepackage[colorinlistoftodos]{todonotes}
\usepackage[en-US]{datetime2}
\def\[#1\]{\begin{align}#1\end{align}}
\def\*[#1\]{\begin{align*}#1\end{align*}}

\newcommand{\Reals}{\mathbb{R}}
\newcommand{\Nats}{\mathbb{N}}

\newcommand{\NNReals}{\Reals_{\ge 0}}
\newcommand{\PosReals}{\Reals_{> 0}}

\newcommand{\conv}{\textrm{conv}}

\newcommand{\dee}{\mathrm{d}}

\DeclareMathOperator*{\newlim}{\mathrm{lim}\vphantom{\mathrm{infsup}}}
\DeclareMathOperator*{\newmin}{\mathrm{min}\vphantom{\mathrm{infsup}}}
\DeclareMathOperator*{\newmax}{\mathrm{max}\vphantom{\mathrm{infsup}}}
\DeclareMathOperator*{\newinf}{\mathrm{inf}\vphantom{\mathrm{infsup}}}
\DeclareMathOperator*{\newsup}{\mathrm{sup}\vphantom{\mathrm{infsup}}}
\renewcommand{\lim}{\newlim}
\renewcommand{\min}{\newmin}
\renewcommand{\max}{\newmax}
\renewcommand{\inf}{\newinf}
\renewcommand{\sup}{\newsup}

\newcommand{\cF}{\mathcal F}

\newcommand{\proj}[1]{\mathrm{proj}_{#1}}

\newcommand{\BorelSets}[1]{\mathcal{B}[#1]}
\newcommand{\NSE}[1]{{^{*}#1}}
\newcommand{\ST}{\mathsf{st}}

\newcommand{\PowerSet}{\mathscr{P}}
\newcommand{\HReals}{\NSE{\Reals}}

\newcommand{\NS}[1]{\mathrm{NS}(#1)}

\newcommand{\cA}{\mathcal{A}}

\newcommand{\cC}{\mathcal{C}}

\newtheorem{open problem}{Open Problem}

\newcommand{\Loeb}[1]{\overline{#1}}

\newcommand{\interior}[1]{%
  {\kern0pt#1}^{\mathrm{o}}%
}

\newcommand{\refproof}[1]{See \cref{#1} for \IfSubStr{#1}{,}{proofs}{a proof}. }

\newif\iflongform
\longformtrue
\iflongform

\else

\fi

\renewcommand\thmcontinues[1]{Continued}

\makeatletter
\providecommand*{\toclevel@definition}{0}
\providecommand*{\toclevel@theorem}{0}
\providecommand*{\toclevel@lemma}{0}
\makeatother
\usepackage{leftidx}

\makeatletter
\patchcmd{\@maketitle}
  {\ifx\@empty\@dedicatory}
  {\ifx\@empty\@date \else {\vskip3ex \centering\footnotesize\@date\par\vskip1ex}\fi
   \ifx\@empty\@dedicatory}
  {}{}
\patchcmd{\@maketitle}   % <----- This line changed
  {\ifx\@empty\@date\else \@footnotetext{\@setdate}\fi}
  {}{}{}
\makeatother

\newtheorem{innercustomthm}{Theorem}
\newenvironment{customthm}[1]
{\renewcommand\theinnercustomthm{#1}\innercustomthm}
{\endinnercustomthm}

\newtheorem{innercustomexp}{Example}
\newenvironment{customexp}[1]
{\renewcommand\theinnercustomexp{#1}\innercustomexp}
{\endinnercustomexp}

\newtheorem{innercustomprop}{Proposition}
\newenvironment{customprop}[1]
{\renewcommand\theinnercustomprop{#1}\innercustomprop}
{\endinnercustomprop}

\title[]{ \vspace{-2cm} Existence of Equilibria in  Large Competitive Markets \\ \vspace{0.5 em} with Bads, Production and Comprehensive Externalities}\thanks{This work is dedicated to the memory of Peter A. Loeb (July 3, 1937 -- November 20, 2024):  deeply-valued friend and colleague of the two older authors, and an important intellectual influence for all four. The authors thank Bernard Cornet, Elena del Mercato, Aniruddha Ghosh, Carlos Herves-Beloso, Peter Loeb, Rich McLean, Emma Garcia Moreno, Arthur Paul Pederson, John Quah, Kevin Reffett, Eddie Schlee, Chris Shannon, Max Stinchcombe, Yeneng Sun, Xinyang Wang, Nicholas Yannelis, and Bill Zame for helpful discussions. The authors appreciate financial support from the following sources: Anderson from Swiss Re through the Consortium for Data Analytics in Risk, Duanmu from National Science Foundation of China No.2460100122,
Khan from 2022 JHU Discovery Award ‘‘Deception and Bad-Faith Communication," and under a contract with {\it Blue Green Future $\&$ Rebalance Earth}.} 

\author[R.~M.~Anderson]{Robert M.~Anderson $^{1}$}
\address{$^{1}$ Institute for Advanced Study in Mathematics, Harbin Institute of Technology\\
and Department of Economics, University of California, Berkeley}
\email{robert.anderson@berkeley.edu}
 
\author[H.~Duanmu]{Haosui Duanmu $^{2}$}
\address{$^{2}$ Institute for Advanced Study in Mathematics, Harbin Institute of Technology}
\email{duanmuhaosui@hotmail.com}

\author[M.~A.~Khan]{M.~Ali Khan $^{3}$}
\address{$^{3}$ Department of Economics, The Johns Hopkins University}
\email{akhan@jhu.edu}

\author[M.~Uyanik]{Metin Uyanik $^{4}$}
\address{$^{4}$ School of Economics, University of Queensland}
\email{m.uyanik@uq.edu.au}

\date{\today}

\newcommand{\cB}{\mathscr{B}}

\newcommand{\FM}[1]{\mathcal{M}(#1)}

\newcommand{\cK}{\mathcal{K}}

\newcommand{\topology}{\mathcal{T}}

\newcommand{\cto}{\twoheadrightarrow}

\newcommand{\cO}{\mathcal{O}}
\newcommand{\cT}{\mathcal{T}}
\newcommand{\NegReals}{\mathbb{R}_{\leq 0}}
\newcommand{\rc}{\mathbb{C}}

\newcommand{\RNum}[1]{\uppercase\expandafter{\romannumeral #1\relax}}

\newcommand*{\argmax}{\operatornamewithlimits{arg max}\limits}

\begin{document}

%\author[M.~A.~Khan]{M.~Ali Khan}
%\address{M.~Ali Khan, The Johns Hopkins University, Department of Economics}

%\author[W.~Weiss]{William Weiss}
%\address{William Weiss, University of Toronto, Department of Mathematics}

%\author[A.~Smith]{Aaron Smith}
%\address{University of Ottawa, Department of Mathematics and Statistics}

\maketitle

\begin{abstract}

\vspace{-1cm}

%This paper presents a general equilibrium model with infinitely many agents, individualistic consumption externalities and bads. We identify and resolve critical flaws of general equilibrium theory (GE) that are specific to the presence of bads such as free-disposal in production and failure of uniform integrability, and show that the model is coherent (i.e equilibrium exists), hence addressing the question raised in \citet{ha05}. 
%The main technical drive of the paper is nonstandard analysis, which allows us to (1) construct the so-called Loeb space which has a richer measure theoretic structure compare with standard measure spaces; (2) import results that are true in finite GE models to results in GE models with infinitely many agents. We give examples that illustrate generality and applicability of the result. 

This paper establishes existence of  equilibrium in a measure-theoretic general equilibrium (MGE) model with production, bads, and comprehensive externalities. These features are jointly essential for modeling perfect competition in which emissions of production byproducts impose harm on agents. We show that, when bads and externalities are modeled in an economically natural way, equilibrium exists. This yields the first existence theorem with bads for MGE models, the benchmark for perfect competition, overcoming \citet{ha05}'s nonexistence example. The proof uses nonstandard analysis, which provides a systematic technique to extend results for finite to infinite models.
\hfill{(92~words)}

\end{abstract}

{\it Key Words:}  {bads, production, comprehensive externalities,  measure-theoretic economy, existence of equilibrium, %incomplete and  price-dependent preferences, 
nonstandard analysis}%, Loeb space 

%

%\begin{quote}

%Unlike the case of economies consisting of finitely many consumers, no equilibrium existence theorem without monotone preference relations and free disposability has been provided for continuum economies. [An] example  of the nonexistence of a competitive equilibrium shows that even the simplest model of a continuum economy with {\it bads} cannot pass the most basic internal consistency test for economic models.\footnote{See pages 648-649 in \citet{ha05}.}%Also see \citet{mu74} as a direct precursor of the use of the word \lq\lq bads." 
%\hfill{\citet{ha05}}
%\end{quote}

\section{Introduction}
\setcounter{footnote}{4}
The Measure-theoretic General Equilibrium 
(MGE)\footnote{In this paper, we use the acronym MGE to refer to General Equilibrium models with a measure space of consumers, including exchange and production economies, with or without bads and externalities.} model introduced by \citet{au64} and \citet{vind64} is the canonical benchmark for large competitive markets.\footnote{Nonstandard analysis provides an alternative formulation of large competitive economies.  In particular, hyperfinite economies offer an important advantage, allowing consideration of allocations that are {\em not} integrable with respect to the underlying population measure, as in \cref{hypereqexist}. 
See \citet{br72-ns}, \citet{ali74}, \citet{nsexchange} and the survey article \citet{an91}.}  Its central economic appeal is that each consumer is negligible and therefore takes prices as given, providing the standard justification for consumer price-taking behavior. 
A natural question is whether this benchmark remains internally consistent once one incorporates economically important features associated with bads, i.e. undesirable commodities.  \citet{ha05} shows that the standard measure-theoretic exchange model fails this test: when bads are present, non-free-disposal equilibrium may fail to exist because candidate equilibrium allocations involve nonintegrable consumption of bads.\footnote{Moreover, in finite approximations the corresponding equilibrium allocations need not be uniformly integrable, so that an asymptotically negligible group of agents absorbs a nonnegligible share of the bads.  Hara explicitly singles out \citet{ma91ecma} on the failure of uniform integrability in sequences of core allocations with nonmonotonic preferences.\label{fn_manelli}} 
Thus, the standard continuum model cannot, as it stands, accommodate bads.

This paper establishes the first existence theorem for an MGE model that simultaneously incorporates bads, production, and comprehensive externalities. These three features belong together economically. Bads that matter in applications are typically generated as byproducts of production rather than from endowments, and the relevant social harms arise through external effects on agents. A competitive model that takes the control of bads seriously must therefore allow for production-generated bads, externalities, and non-free-disposal equilibrium. We show that equilibrium existence can be restored once the MGE model is formulated in a way that respects these economic considerations.

   We argue below that the primary economic motivation for controlling bads depends on the presence of externalities, that economically important bads arise as byproducts of production, and that the disposal of bads generates externalities that can usefully be controlled through pricing.  In other words, we need an MGE model in which bads generated as byproducts of the production process can be consumed by consumers or sequestered by firms, but {\em not}  disposed.  
   This leads us to consider an MGE model that differs 
   from \citet{ha05}'s in several economic respects.   We give economically natural conditions under which equilibrium exists in our model. In particular, we argue in this setting that an integrable bound on individuals' consumption of bads is a natural consequence of their physical limitations.  
Our result also requires rethinking several standard assumptions of the MGE literature. Monotonic preferences and free disposal in production are technically convenient, but they are economically problematic when bads are to be controlled.\footnote{The assumptions of monotonic preferences, free disposal in production, and free disposal in equilibrium, have been relaxed in finite economies;  see \citet{Mck59, McKenzie81},  \citet{Hart75}, \citet{sh76}, \citet{Berg76}, \citet{ga79}, \citet{pole93}, and others.} Monotonicity forces nonnegative prices, while free disposal in production makes it possible for byproducts such as emissions to disappear from the accounting, thereby ruling out the negative prices that provide incentives for control. More generally, if externalities are allowed to depend only on others’ consumption, the model misses the economically central case in which harms arise from what firms produce but society does {\em not} consume. For this reason, the relevant continuum benchmark must allow preferences to depend not only on other consumers’ allocations but also on firms’ production plans\footnote{See
the discussion in \cref{section:Problematic}.}

The theorem has important economic content: it establishes existence of competitive equilibrium in the large-market benchmark needed to analyze production-generated bads and their external effects. The proof is technically demanding, but the statement is standard in economic form. What is new is that bads, production, and comprehensive externalities can be incorporated simultaneously into an MGE economy without sacrificing existence. To make these ideas precise, our model has eight key economic ingredients:
\begin{enumerate}[leftmargin=*]
\item {\em Categorization of commodities into goods and 
bads:} %\footnote{\label{fn:manna1}The {\it goods-bads} distinction goes back to the very foundation of the subject: see the first three paragraphs of the introduction in \citet{ya74pc} and his references to the pioneering papers in Volumes 2 and 3 of \citet{sa-c}, to  \citet{mu59bk}, and to \citet{bu68}.  An alternative property-rights viewpoint is offered in \citet{ba23bk}. Many important fair division problems involve {\it bads} -- nondisposable commodities  generating disutility; see \citet{bmsy17ecma} for many examples in a variety of settings beyond Walrasian equilibrium theory.}   
 We categorize the commodities into {\it goods} and {\it bads}.
Preferences need not be monotonic over  bads.\footnote{This classification into goods and bads need not be universally agreed upon by consumers.  A particular individual may consider a commodity to be a bad when another individual may consider it to be a good, or may be indifferent to it. In other words, commodities are ``mixed manna" as defined by \citet{bmsy17ecma}.
This partitioning of the commodity space is due originally to  \citet{fo70}  in the context of public and private commodities, and was followed in subsequent work as in \citet{kv85} and their followers. %In distinction from  the partial equilibrium references in Footnote \ref{fn:manna1}, this latter work is in general equilibrium.
} 
\item \label{item-Simultaneous} {\em Simultaneous consideration of externalities and bads:}
In the absence of externalities, every free-disposal equilibrium with nonnegative prices is Pareto optimal;\footnote{See \cref{theorem-First_No_Externality} in \cref{section-First_No_Externality}.} bads are freely disposed at zero price. In the absence of externalities, non-free-disposal equilibrium prices of bads are typically negative and non-free-disposal equilibria are often Pareto dominated by free-disposal equilibria.\footnote{There is no externality in \citet{ha05}'s example. If free disposal of the bad is allowed in that example, the bad can be disposed at zero price and the free-disposal equilibrium avoids the nonintegrability problem. However, when disposal generates external effects, free disposal is no longer innocuous.} %Consider an exchange economy with one agent. The agent's endowment is $(1, 1)$ and her/his utility function is $u(x_1,x_2)=x_1-x_2$. Then the allocation $(1,1)$ with the price $(0.5,-0.5)$ is the only non-free-disposal equilibrium. This equilibrium is Pareto dominated by the free-disposal equilibrium $(1, 0)$ with price $(1,0)$.} 
Thus, in the absence of externalities, \emph{free-disposal equilibrium} rather than non-free-disposal equilibrium is the \emph{correct} equilibrium notion. 
In the presence of externalities, negative prices provide incentives to control bads.  
For this reason, an MGE model that is compatible with externalities and the control of bads must consider non-free-disposal equilibrium.
\enlargethispage{1\baselineskip}
 \item {\em Comprehensive externalities:}\footnote{
%What, if anything, should we do with this?  For the literature on externalities, see  \citet{ha99} and   \citet[Section 16.2]{ha11}, and the references in Footnotes \ref{fn:fd} and \ref{fn:lit} below. In the concept of an economy with a  finite set of agents, the notion of this dependence is straightforward, but it is in the context of economies with a continuum of agents that distinctions as to these dependence have been proposed; externalities in the form of average or distribution of the others' actions or {\it comprehensive} externalities.  Hammond lists \citet{Berg76} and \citet{ga79} on the one hand, and \citet{ha68}, \citet{bo75}, and \citet{sc78} among others; see especially the discussion on pp. 9-10. Again, in a further elaboration, he writes,  \lq\lq Obvious examples include familiar externalities such as the emission of methane, carbon dioxide, or other greenhouse gases into the atmosphere, as well as gases causing ozone depletion, and farming practices which eventually add to pesticide residues in aquifers or in the oceans." He gives a host of other examples, and also discusses {\it narrow} externalities whose generalization in the form of {\it comprehensive} externalities  we take up in the sequel. 
For recent work on general equilibrium theory with externalities, see \citet{mercato17}, \citet{sico23}, and \citet{mercato23}.}
In the equilibrium concept that we present (\cref{defmeasureeco}), consumers’ preferences may depend not only on their own consumption bundles but also on the allocation, firms’ production plans, and equilibrium prices. This allows the model to capture both global externalities (e.g. ${\rm CO}_2$ affecting climate) and local ones (e.g. wastewater affecting nearby households), whereas earlier MGE models restricted externalities to others’ consumption.\footnote{See for example \citet{hammond89}; \citet{cornet05}; \citet{Noguchi05}; \citet{bal08}; \citet{nieto21}; and \citet{yanni24}, which treat only exchange economies. The model in \citet{Noguchi2006} has a production sector, but consumers' preferences do not depend on the production.} Our framework instead acknowledges that most bads originate as byproducts of production, and thus extends the literature to cover the externalities that matter most.  Crucially, consumers' preferences may depend on emissions, commodity bundles that are produced but {\em not} consumed.\footnote{Our main result, \cref{standardmain}, establishes the existence of non-free-disposal equilibrium: all byproducts of production must be either eliminated by further production activity (e.g. sequestration of ${\rm CO}_2$)  or consumed; emissions are zero at equilibrium.  Equilibrium prices give firms incentives to eliminate bads and consumers the incentive to absorb them. \cref{standardmainqt} in the Supplementary Material extends our existence theorem to quota equilibrium, which allows for emissions equal to a positive regulator-specified cap.}  Since there are multiple technologies for generating electricity, the $\mathrm{CO}_{2}$ emissions cannot be recovered from the total electricity consumption alone.

\item {\em No free disposal:} 
Free disposal can occur at three stages in the Arrow-Debreu model,\footnote{In free-disposal equilibrium, the market clearing condition is that demand is no greater than supply; the excess is disposed freely.  Free disposal in production is a standard assumption in production economies; the difference between the amount produced by a firm and the amount sold to consumers or other firms is disposed freely, and disappears from the accounting. Finally, consumers are not required to use up all the commodities they purchase; the excess is disposed freely, and disappears from the accounting.} so we exclude free disposal at each of the three stages.\footnote{We rule out free disposal in consumption for bads only.  See \cref{cointro}, especially \cref{footnote:badsonly}.} This paper focuses on non-free-disposal equilibrium.  No previous paper has established the existence of non-free-disposal equilibrium in a measure-theoretic economy with both bads and externalities.\footnote{\label{fn:nz}\citet{cornet03} consider production economies with possibly satiated consumers, but no externalities; they establish the existence of individualistic free-disposal equilibrium. \citet{cornet05} consider exchange economies with externalities.  Their Theorem 4 establishes the existence of free-disposal individualistic equilibrium with nonnegative prices, so negative prices cannot provide incentives to control the externalities arising from bads.
Their other results require monotonic preferences. \citet{Noguchi2006} consider production economies with externalities. They prove the existence of distributional equilibrium, but they require monotonic preferences and free disposal in production.
\citet{ma06wp} and \citet{in22} consider exchange economies with no externalities.  They impose  conditions to imply that the candidate non-free-disposal equilibrium allocation is integrable.  \citet{yanni24} consider exchange economies with externalities and locally nonsatiated satiated consumers. They impose an integrable bound on consumption of all commodities.}

\item \em{Consumption sets integrably bounded:}  
\citet{ha05} considers an exchange economy without externalities; the nonexistence is driven by the fact that  the marginal rate of substitution from the good to the bad goes to infinity.\footnote{\citet{ma06wp} and \citet{in22} explore conditions on marginal utility that resolve Hara's example by making the candidate equilibrium integrable.} 
However, as we argued in \cref{item-Simultaneous} above, in the absence of externalities, free-disposal equilibrium is the correct equilibrium notion. As we demonstrate in \cref{frdiseqremark}, Hara's sequential example has a sequence of free-disposal equilibria with transfers that Pareto dominate the non-free-disposal equilibria.

 In the presence of externalities, the focus should be on non-free-disposal equilibrium. Typically, it is the {\em emission} of the bad, such as $\mathrm{CO}_{2}$, rather than its consumption, that creates a negative externality. While the ownership of a good conveys the right, but not the obligation, to use up the good, we require non-free disposal in consumption of bads: a consumer who purchases a unit of the bad (typically at a negative price) is required to absorb or eliminate it.
But the capacity of a given consumer to absorb the bad is clearly limited.
For example, a consumer endowed with land may agree to accept garbage in return for payment, but the amount of garbage consumers may accept is limited by their integrable endowment of land.\footnote{Alternatively, ${\rm CO}_2$ produced as a byproduct of electricity generation can potentially be captured and sequestered.  If so, this is a production process carried out by firms, not by consumers.  While a consumer may ingest ${\rm CO}_2$ by breathing in air or drinking a carbonated beverage, it is impossible to convert the carbon into body mass; the ${\rm CO}_2$ is promptly re-emitted.  Thus, the capacity of a {\em consumer} to absorb ${\rm CO}_2$ is {\em zero}.}  This is a constraint on the physical abilities of the consumer, as no consumer can absorb an unbounded amount of bads. {To control bads, we must rule out free disposal in consumption, and our integrable bound on consumption of bads is a natural assumption reflecting the physical {\em capabilities} of consumers, i.e. it should be imposed on their consumption sets, rather than their preferences.\footnote{The absorption of ${\rm CO}_2$ and some other bads is uniformly bounded.  The assumption of an integrable bound is weaker than the assumption of a uniform bound.}

\item {\em Individualistic approach:} Our model uses the individualistic approach  based on a measure space of consumers, rather than the distributional approach in which the distributions of preferences, agents and endowments are modeled rather than the 
space of consumers. 
\item {\em Convex preferences:}  \citet{au66} showed that in an exchange economy with an atomless measure space of consumers, convex preferences are not required for existence of equlibrium. %\footnote{In a pure exchange economy with an atomless continuum of agents and no extermalities, as in \citet{au66} and \citet{schmeidler69}, atomlessness allows nonconvexities to be integrated away.}    %Even with externalities, one can conceive of specialized aggregative forms -- as in \citet{sc73} and the theory of large anonymous and nonanonymous games built on the basis of his work — where convexity can be dispensed with. 
However, it is well understood that this is no longer possible with comprehensive externalities.\footnote{See \citet{getal79jme}; \citet{cornet03}; \citet{cornet05}; \citet{bal08}; and \citet{nieto21}.
%Moreover, the fact that convexity cannot be dispensed with even in the special case of externalities stemming from public goods, and even with a continuum of agents, has been understood since \citet{ro73jet}; \citet{emmons84}; and \citet{kv85}.  
\citet{Noguchi2006} shows that individualistic equilibrium may not exist when preferences are not convex. Thus, one can handle externalities or nonconvex preferences, but not both. We believe that handling externalities is more important.}     For this reason, we assume preferences are convex.%Moreover, convexity is essential for a central piece of our nonstandard analysis argument, the ``pushing down'' from an equilibrium in the Loeb space to an equilibrium in the original measure space of consumers.   %As noted above, we do not assume atomlessness from the outset—though our framework allows it if needed -- since imposing it would not help in addressing the difficulties created by comprehensive externalities.
\item {\em Quota equilibrium:} \citet{ad25} define the notion of quota equilibrium, a generalization of cap-and-trade, in which disposal is allowed up to a positive cap.   They establish the existence of quota equilibrium in finite production economies. \cref{standardmainqt} in the Supplementary Material extends the MGE results of this paper to quota equilibrium.}
\end{enumerate}

{ Nonstandard analysis, originated in \citet{AR65}, is a powerful technique from mathematical logic.\footnote{Mathematical logic has recently found other applications in economics. For example, \citet{Hellman19}, \citet{hellman22} and \citet{yehuda25} applied infinitary logic and descriptive set theory to establish the existence of equilibrium in broad classes of games.} As opposed to earlier applications of nonstandard analysis in economics, which used results from infinite settings to derive asymptotic results in large finite settings,\footnote{Previous applications of nonstandard analysis to economics include \citet{nsexchange}, \citet{oligopoly}, \citet{strongcore}, \citet{Simon-Stinchcombe95}, \citet{ks01}, \citet{DS07}, \citet{ar08}, and \citet{DLS18}.} our use  provides a unified framework for extending  finite economy theorems to continuum economies.\footnote{\citet{kominer25} applied the compactness theorem from propositional logic to remove finiteness assumptions in various economic models.} In particular, nonstandard analysis allows for the construction of a \emph{hyperfinite probability space}, which satisfies all the first order logical properties of finite probability spaces, but can be simultaneously viewed as a measure-theoretic probability space via the Loeb measure. Thus, nonstandard analysis is particularly powerful in situations in which the 
desired result is known for finite objects, its {\it proof} depends heavily on finiteness, but its {\it statement} makes sense for infinite objects.\footnote{This particular technique in nonstandard analysis was developed by the second author in his thesis \citet{duanmuthesis}, and subsequently applied to Markov processes (\citet{drw21}), statistical decision theory (\citet{nsbayes} and \citet{nscredible}) and economics (\citet{adku2}, \citet{adku1} and \citet{berknash23}).}
Our argument involves the following steps:\footnote{A more detailed sketch of the proof is given in \cref{secmethod}.}}
\begin{enumerate}[leftmargin=*]
\item Adapt results of \citet{fl03bk}\footnote{Since Florenzano (2003), existence theory with a finite set of agents has advanced along several dimensions: continuity, order, price-dependence, the survival assumption, and free-disposal assumptions have been weakened; see for example \citet{HY2016,hy17} and \citet{py21}. 
%{Inoue (2025), Anderson, Duanmu, Khan, and Uyanik (2025), and Khan, McLean, and Uyanik (2025).} 
\citet{yanni24} considers MGE exchange economies with discontinuous preferences. Our argument requires a metric structure on the space of preferences, so we restrict attention to the compact metric space of continuous preferences, and use the results in \citet{fl03bk}.} to prove the existence of equilibrium in weighted finite production economies with externalities (\cref{finitemainexpd});
\item Take any standard measure-theoretic production economy with externalities, and construct a \textit{lifting}, %\footnote{\citet{emmons84} does a simpler form of lifting. As a result, he obtains existence of Lindahl equilibrium only in measure-theoretic economies with a hyperfinite Loeb space of consumers, while our result is valid for standard measure spaces of consumers.} 
embedding our standard economy in a hyperfinite economy (\cref{secconshp});
\item Use the \textit{transfer principle} of nonstandard analysis to obtain the theorem for the hyperfinite economy, essentially for free (\cref{hypereqexist}); 
\item\textit{Push down} the theorem for the hyperfinite economy to obtain the existence of equilibrium in the corresponding Loeb measure economy (\cref{alocboundlemma}).\footnote{Many authors have shown the richness of the Loeb $\sigma$-algebra allows the existence of solutions when, due to a lack of measurable sets, no solution exists in the original measure.  See e.g., \citet{Keisler87}, \citet{emmons84}, \citet{KS96pnas}, %\citet{KS99}, 
\citet{DS07}.  %, and \citet{DLS18}. 
\citet{keisler09} %. \citet{He17} 
and \citet{He18} established the necessity of using spaces with rich measure-theoretic structure to model economies with many agents.
See also \citet{CP09} in this connection.}  \cref{alocboundlemma} is of economic interest in its own right, because it allows a broader class of externalities. However, understanding it requires a detailed knowledge of nonstandard analysis.\footnote{When consumers' preferences maps are only continuous with respect to the $\mathcal{L}^{1}$ norm topology on $\mathcal{L}^{1}(\Omega, \NNReals^{\ell})$, equilibrium need not exist in $\mathcal{E}$. This failure is due to the lack of sufficiently many measurable sets in $\Omega$: the candidate equilibrium allocation may not be measurable. The Loeb measure space has a much richer collection of measurable sets, and this allows us, in \cref{alocboundlemma}, to show the existence of equilibrium of the Loeb production economy $\Loeb{\mathscr{E}}$ associated with $\mathcal{E}$. 
The allocation of the equilibrium of $\Loeb{\mathscr{E}}$ lies outside the domain of the preferences in the original economy $\mathcal{E}$, but the Loeb production economy $\Loeb{\mathscr{E}}$ extends the preferences from $\mathcal{E}$ in a way that preserves all of their essential properties. Thus, when consumers' preferences are only continuous with respect to the $\mathcal{L}^{1}$ norm on $\mathcal{L}^{1}(\Omega, \NNReals^{\ell})$, we view the Loeb production economy as a better modeling alternative, and its equilibrium as the most natural solution.}
%\footnote{When consumers' preferences maps are continuous with respect to the weak topology on $\mathcal{L}^{1}(\Omega, \NNReals^{\ell})$, we can push down the equilibrium of the Loeb production economy to obtain an equilibrium in the original measure-theoretic production economy $\mathcal{E}$, as in \cref{standardmain}.} 

\item If consumers' preferences are weakly continuous,\footnote{For previous applications of the weak topology, see  \citet{kv84} and \citet{cornet05}.} we {\em push down} the equilibrium of the Loeb economy to an equilibrium of the original measure space economy (\cref{standardmain}).\footnote{The push down makes use of a nonstandard characterization of the standard part with respect to the weak topology that, to our knowledge, has not previously appeared in the literature.  It involves taking a conditional expectation in the Loeb measure space with respect to the $\sigma$-algebra of inverse images of sets that are measurable in the original measure space.}
\end{enumerate} 
From this perspective, nonstandard analysis is not simply a technical device but a conceptual tool: it provides a bridge between finite general equilibrium models—where existence is well understood—and measure-theoretic economies, where direct arguments are often more involved. While alternative proofs based on standard methods may exist in principle, our approach offers a systematic route to the result that may be useful for studying other problems in large economies with similar structural features.

\section{Problematic Assumptions in the Existing MGE Literature}\label{section:Problematic}

Assumptions relating to bads and externalities have been relaxed in finite economies, but are still imposed in state-of-the-art papers on MGE economies.  In this section, we explain why these familiar assumptions are problematic in the presence of bads.  

\subsection{Bads without Externalities}
\label{section-First_No_Externality}
  
In the presence of bads, non-free-disposal equilibrium is commonly regarded as the appropriate equilibrium notion, rather than free-disposal equilibrium. In this section, however, we demonstrate that when externalities are absent, free-disposal equilibria  associated with nonnegative equilibrium prices are Pareto optimal even with the presence of bads.\footnote{For simplicity, we prove the result for economies with finitely many consumers. Note that it is straightforward to generalize the following result to economies with a measure-theoretic space of consumers.} It follows that to rationalize the adoption of non-free-disposal equilibrium, an MGE model must jointly feature externalities and bads.

\begin{theorem} \label{theorem-First_No_Externality}
Let $\mathcal{E}$ be a finite production economy such that each consumer's preference exhibits no externality, is negatively transitive and locally non-satiated. Let $(\bar{x},\bar{y},\bar{p})$ be a free-disposal equilibrium such that $\bar{p}\geq 0$. Then $\bar{x}$ is Pareto optimal.   
\end{theorem}

\begin{proof}
Since consumers' preferences exhibit no externality, we use $\succ_{\omega}$ to denote consumer $\omega$'s preference. 
Suppose that there is an attainable allocation $\hat{x}$ that Pareto dominates $\bar{x}$. Since $\hat{x}$ is attainable, we can choose $\hat{y}\in Y$ such that $\sum_{\omega\in \Omega}\hat{x}(\omega)-\sum_{\omega\in \Omega}e(\omega)-\sum_{j\in J} \hat{y}(j)\leq 0$. That is, $(\hat{x}, \hat{y})$ is an attainable consumption-production pair. 
Then $\bar x (\omega) \not \succ_\omega \hat x (\omega)$ for every $\omega$, and there is some $\omega_0\in \Omega$ such that $\hat{x}(\omega_0)\succ_{\omega_0} \bar{x}(\omega_0)$. As $(\bar{x}, \bar{y}, \bar{p})$ is a free-disposal equilibrium, we have $\bar{p}\cdot \hat{x}(\omega_0)>\bar{p}\cdot e(\omega_0)+\sum_{j\in J}\theta_{\omega_{0} j}\bar{p}\cdot \hat{y}(j)$. 

\begin{claim}\label{nolesshatf}
For every $\omega\in \Omega$, we have
$\bar{p}\cdot \hat{x}(\omega)\geq \bar{p}\cdot e(\omega)+\sum_{j\in J}\theta_{\omega j}\bar{p}\cdot \bar{y}(j)$. 
\end{claim}
\begin{proof}
Suppose there is a $\omega_1\in \Omega$ so that
$\bar{p}\cdot \hat{x}(\omega_{1})<\bar{p}\cdot e(\omega_{1})+\sum_{j\in J}\theta_{\omega_{1} j}\bar{p}\cdot \bar{y}(j)$.
As $\succ_{\omega_{1}}$ is locally non-satiated, there is a $u\in X_{\omega_{1}}$ such that $u\succ_{\omega_1}\hat{x}(\omega_{1})$ and
$\bar{p}\cdot u<\bar{p}\cdot e(\omega_{1})+\sum_{j\in J}\theta_{\omega_{1} j}\bar{p}\cdot \bar{y}(j)$.
Note that we have $\bar{x}(\omega_1)\not\succ_{\omega_1} \hat{x}(\omega_1)$. 
If $u\not\succ_{\omega_1} \bar{x}(\omega_1)$, by negative transitivity of $\succ_{\omega_1}$, we have $u\not\succ_{\omega_1} \hat{x}(\omega_1)$, a contradiction. Hence, we must have $u\succ_{\omega_1} \bar{x}(\omega_1)$.
This leads to a contradiction since $(\bar{x}, \bar{y}, \bar{p})$ is a free-disposal equilibrium.
\end{proof}

By \cref{nolesshatf},
$\bar{p}\cdot \hat{x}(\omega)\geq \bar{p}\cdot e(\omega)+\sum_{j\in J}\theta_{\omega j}\bar{p}\cdot \hat{y}(j)$
for all $\omega\in \Omega$. So, we have $\bar{p}\big(\sum_{\omega\in \Omega}\hat{x}(\omega)-\sum_{\omega\in \Omega}e(\omega)-\sum_{j\in J} \hat{y}(j)\big)>0$. But this is impossible since $\bar{p}\geq 0$ and $\sum_{\omega\in \Omega}\hat{x}(\omega)-\sum_{\omega\in \Omega}e(\omega)-\sum_{j\in J} \hat{y}(j)\leq 0$, contradiction. Hence, $\bar{x}$ is Pareto optimal. 
\end{proof}

\subsection{Monotonic Preferences}\label{secmonnfdiscuss}
All but one of the papers establishing existence of equilibrium in MGE  production economies assumes strong monotonicity of consumers' preferences.\footnote{The one exception, \citet{hilden70}, allows for nonmonotonic preferences but assumes free disposal in production. As in \cref{fdintro}, free disposal in production implies that the equilibrium prices are nonnegative.} Strong monotonicity implies that the candidate equilibrium prices are positive, and hence that the candidate equilibrium allocation is integrable. This defeats the goal of obtaining possibly negative equilibrium prices, providing incentives for controlling bads.

%\subsection{Free Disposal in Production} {As \cref{nnpricelemma} shows, free disposal in production implies that equilibrium prices are nonnegative, precluding the use of negative prices to provide incentives to control bads. \cref{freenofree}  further demonstrates that free disposal in production allows firms to “erase” pollution from their equilibrium  production vectors, even though emissions persist in the environment. By dispensing with free disposal in production, we obtain a model in which negative prices, which emerge naturally from market clearing, provide incentives to control bads.}

\subsection{Free Disposal in Consumption}\label{cointro}
In the presence of externalities, the imposition of non-free-disposal in equilibrium\footnote{The existence of non-free-disposal equilibrium in economies with a finite number of agents has a long history dating to \citet{mc55bkch} %, \citet{Mck59, McKenzie81}, as well as Debreu's 1962 paper referred to  in \citet{de82, de83}; 
and culminating in the papers of \citet{Hart75}, \citet{Berg76}, \citet{ga79} and \citet{sh76}. \citet{fl03bk} offers a comprehensive treatment.   \citet{pole93} deals with bads, while the works  of Shafer and Florenzano cited above 
deal with bads and externalities.}
is necessary but not sufficient to control bads.  We need, in addition, to rule out free disposal in consumption. 
The classical general equilibrium model tacitly assumes free disposal in consumption: commodity ownership conveys the right, but not the obligation, to use up the commodity: a consumer might purchase a commodity, then leave it unconsumed.\footnote{No one forces you to eat out-of-date food rotting in the refrigerator.}   However, if the price of a bad is negative, and a consumer is allowed free disposal in consumption, the consumer has a strong incentive to buy (but not consume) an unbounded amount of the bad to generate income to purchase other commodities.  As a result, the consumer's demand set is empty. Thus, free disposal in consumption is inconsistent with negative equilibrium prices.

In order to ensure that negative prices can exist in equilibrium and provide incentives for the proper stewardship of bads, we rule out free disposal in consumption for bads.\footnote{\label{footnote:badsonly} We do not rule out free disposal in consumption of goods, since this is not required for the validity of \cref{standardmain}.} 
This has an important implication for consumers' consumption sets:  the capacity of a consumer to render a bad harmless to others is limited, and the consumption set must reflect that limit.  

Note finally the distinction between the externalities generated by consumption and those generated by production. The climate change externality generated by ${\rm CO}_2$ emissions does not arise from the {\em consumption} of ${\rm CO}_2$.  Instead, it arises from the ${\rm CO}_2$ which is {\em produced but  not consumed}, and is emitted into the atmosphere as a result.\footnote{Carbon sequestration may emerge as a practical technology for eliminating ${\rm CO}_2 $ emissions.  If so, it will be an industrial production process, not a consumption activity.}  %If our model allowed free disposal in consumption, we would have no mechanism to prevent individual consumers from buying ${\rm CO}_2$ at a negative price and then simply releasing it into the atmosphere.

\subsection{Free Disposal in Production}\label{fdintro}
Free disposal in production asserts that, if $y$ is a feasible production vector and $z<y$, then $z$ is also a feasible production vector. %Free disposal in production is, however, incompatible with the control of bads for the following reasons:
\begin{enumerate}[leftmargin=*]
    \item Free disposal in production is unrealistic. Suppose we have three commodities: $\mathrm{CO}_2$, coal and electricity. Suppose that $(1,-1,1)$ is a feasible production vector: burning one unit of coal generates one unit of electricity and one unit of $\mathrm{CO}_2$, as a byproduct. Under free disposal in production, $(0,-1,1)$ must also be a feasible production vector: one can burn one unit of coal to produce one unit of electricity and {\em zero} $\mathrm{CO}_2$. This is physically impossible; the unit of $\mathrm{CO}_{2}$ has to go somewhere, most likely the atmosphere. Under free disposal in production, it simply disappears from the accounting. %In the Arrow-Debreu model, the firm's profit is determined from the output vector, and thus is unaffected by its $\mathrm{CO}_{2}$ emissions;   
    \item Free disposal in production implies that the equilibrium price is nonnegative, precluding taxes on bads to provide incentives for controlling emissions; see \cref{nnpricelemma};
    \item %In the previous literature, non-free-disposal equilibrium has been viewed as the correct equilibrium notion in the presence of bads, because it seems to indicate that markets clear for all bads. 
    \cref{freenofree} shows that, with free disposal in production, a free-disposal equilibrium can be disguised as a non-free-disposal equilibrium.
\end{enumerate}

% \textcolor{blue}{Finally we turn to the plan of the paper.} {\color{red}*****We recommend dropping this due to the length of the paper.*****}

\section{An Example of Equilibrium Non-existence}\label{secnonest}
In this section, we study an example by \citet{ha05} of the non-existence of  non-free-disposal equilibrium in MGE exchange economies with bads.%, and that a very small group of consumers may end up consuming almost all the bads in finite economies with bads. %As \citet{ha05} pointed out, such an equilibrium allocation casts serious doubts on the plausibility of price-taking behavior in finite economics. 

\begin{customexp}{1}\normalfont{\citep[][Example 1]{ha05}}\label{haexample}
%This example presents an exchange economy with 
%bads and a measure space of consumers, in which equilibrium does not exist. In particular, 
Let $\mathcal{E}
=\{\big(X_{\omega},u_{\omega},e(\omega)\big)_{\omega\in \Omega}, (\Omega, \cB, \mu)\}$ be an exchange economy such that:
\begin{itemize}[leftmargin=*]
    \item there are two commodities, a good and a bad, and negative prices are allowed;
    \item the consumer space $(\Omega, \cB, \mu)$ is the Lebesgue measure space on $(0, 1)$;
    \item consumer $\omega\in \Omega$ has consumption set $X_{\omega}=\NNReals^{2}$, endowment $e(\omega)=(2,1)$, and utility function $u_{\omega}(x_1, x_2)=x_{1}-\omega\times  (x_2)^{2}$.
    %\item the set of allocations is $\mathcal{L}^{1}(\Omega,\NNReals^{2})$;
    %\item for $\omega\in \Omega$, the consumer's utility function is defined as $u_{\omega}(x_1, x_2)=x_{1}-\omega(x_2)^{2}$, where $x_1, x_2$ are the agent's consumption of the first and second commodities, respectively. 
\end{itemize}
%Since the second commodity is a bad, we allow for possibly negative prices. 
%Finally, we require non-free-disposal at equilibrium.
 By the first-order conditions for utility maximization,  any non-free-disposal equilibrium allocation $f$ must satisfy $f_{2}(\omega)=\frac{|p_2|}{2\omega}$ for almost all $\omega\in \Omega$. As the function $f_{2}$ is not integrable, the exchange economy $\mathcal{E}$ has no non-free-disposal equilibrium.\footnote{If $f$ is a non-free-disposal equilibrium allocation, then  $\int_{\Omega}\big(f(\omega)-e(\omega)\big)\mu(\dee \omega)=0$, so $f$ must be integrable.}   
\end{customexp}

%Measure-theoretic economies are often regarded as a useful approximation for large but finite economies, and can be viewed as the limit of a sequence of finite economies with respect to the weak topology on probability measures. 
In the next example, we consider a sequence of finite economies that converges to the MGE economy in \cref{haexample}. 

\begin{customexp}{2}\normalfont{\citep[][Example 2]{ha05}}\label{haexample2}
%This example presents a sequence of finite exchange economies in which the natural limit of the sequence of equilibria is not integrable. 
Consider a sequence $\{\mathcal{E}_{n}\}_{n\in \Nats}$ of finite economies with $\mathcal{E}_{n}=\{\big(X_{\omega}, u^{n}_{\omega}, e(\omega)\big)_{\omega\in \Omega_{n}}, (\Omega_{n}, \cB_{n}, \mu_{n})\}$:
\begin{itemize}[leftmargin=*]
    \item there are two commodities, a good and a bad, and negative prices are allowed;
    \item the set of consumers is $\Omega_{n}= \{\frac{1}{n},\frac{2}{n},\dotsc,1\}$;
    \item consumer $\omega\in \Omega_{n}$ has consumption set $X_{\omega}=\NNReals^{2}$, endowment $e(\omega)=(2,1)$, and utility function $u^{n}_{\omega}(x_1, x_2)=x_1-\omega\times (x_2)^{2}$;
    %\item for $n\in \Nats$, let the consumer space  $\cB_{n}$ be the power set of $\Omega_{n}$, and $\mu_{n}$ be the uniform probability measure on $\Omega_{n}$;
    %\item for $\omega\in \Omega_{n}$, the consumer's , where $x_1, x_2$ are the agent's consumption of the first and second commodities, respectively. 
\end{itemize}
The sequence $\{\mathcal{E}_{n}\}_{n\in \Nats}$ of economies converges to the economy $\mathcal{E}$ in \cref{haexample}, in the sense of \citet{hilden74}.\footnote{Because the sequence of distributions of the economies converge weakly, and the sequence of endowments is uniformly integrable, the sequence is ``purely competitive'' and has limit $\mathcal{E}$.} However, the sequence of non-free-disposal equilibria of $\{\mathcal{E}_{n}\}_{n\in \Nats}$ does not converge to an allocation, much less an equilibrium, of $\mathcal{E}$. 
Let $S^{n}=\sum_{s=1}^{n}\frac{1}{s}$. For $n\in \Nats$,  $\mathcal{E}_{n}$ has a unique non-free-disposal equilibrium $(p^{n}, f^{n})$, where $p^{n}=\left(1, -\frac{2}{S^{n}}\right)$ and $f^{n}(\omega)=\left(2+\frac{2}{S^{n}}(\frac{1}{S^{n}\omega}-1), \frac{1}{S^{n}\omega}\right)$. The sequence $\{f^n\}$ is not uniformly integrable.\footnote{There exists a sequence 
$\{a^n\}_{n\in \Nats}$ of positive integers such that
$a^n\leq n$ for all $n\in \Nats$, $\lim_{n\to \infty}\frac{a^n}{n}=0$, and $\lim_{n\to \infty}\frac{1}{n}\sum_{s=1}^{a^n}f^{n}_{2}(\frac{s}{n})=1$.} 
Hence, an asymptotically negligible portion of the population consumes almost all of the bad. Economically, it is physically impossible for the group to absorb the bad. Mathematically, the sequence $\{f^n\}$ has no limit in the limit economy ${\mathcal{E}}$.
If we impose a uniform integrability assumption on the consumption of bads, we get a different, economically sensible, sequence $\{\bar f^n\}$ of equilibria.  So the resolution of existence of equilibrium in the continuum model tells us how to resolve the economically nonsensical equilibria in the sequence of finite economies.%\footnote{The natural candidate for a limit equilibrium assigns the entire consumption of the bad to the point $0$, which is not integrable with respect to Lebesgue measure.}
\end{customexp}

\begin{remark}\label{frdiseqremark}
Hara's examples have no externalities.  We argued above that, in the absence of externalities, the right notion is free-disposal equilibrium, rather than non-free-disposal equilibrium.  Indeed, we now show that there is a free-disposal equilibrium with transfers that Pareto dominates $(p^n, f^n)$. 
%The exchange economy $\mathcal{E}^{n}$ has a free-disposal equilibrium $\left(\overline{p^n}, \overline{f^n}\right)$, where  $\overline {p^n}=\bar p = (1,0)$ and $\overline{f^n}(\omega)=
%\bar f(\omega) = (2, 0)$ for all $n$. %So we have a constant sequence of free-disposal equilibria for the sequence $\{\mathcal{E}_{n}\}_{n\in \Nats}$ of exchange economies. 
%$(\bar{p}, \bar{f})$ is a free-disposal equilibrium for the exchange economy $\mathcal{E}$ in \cref{haexample}. %We now construct a free-disposal equilibrium with transfers $(\bar{p}, \tilde{f}^{n})$ of $\mathcal{E}_n$ that Pareto dominates the non-free-disposal equilibrium $(p^n, f^n)$. 
Let $\bar p = (1,0)$, $T^{n}(\omega)=\frac{2}{S^{n}}(\frac{1}{S^{n}\omega}-1)$ and $\tilde{f}^{n}(\omega)=\left(2+\frac{2}{S^{n}}(\frac{1}{S^{n}\omega}-1),0\right)$.
Since $\sum_{\omega\in \Omega_{n}}T^{n}(\omega)=0$, $(\bar{p}, \tilde{f}^{n})$ is a free-disposal equilibrium with transfers for $\mathcal{E}_n$. Since the first components of $\tilde{f}^{n}(\omega)$ and $f^{n}(\omega)$ are the same, $\tilde{f}^{n}$ Pareto dominates $f^n$. From the welfare perspective, free-disposal equilibrium is the right notion of equilibrium in this example.\footnote{Note that the value $\bar{p}\cdot \sum_{\omega\in \Omega}\bar{f}(\omega)$ of the free-disposal equilibrium allocation is the same as the value $p^{n}\cdot \sum_{\omega\in \Omega}f^{n}(\omega)$ of the non-free-disposal equilibrium. From a utilitarian perspective, the aggregate utility of $\bar{f}$ is strictly larger than the aggregate utility of $f^{n}$, even without transfers.}
\end{remark}

%\citet{ha05} highlighted the importance of giving economically justifiable sufficient conditions for the existence of non-free-disposal equilibrium in measure-theoretic economies with bads.
%\cref{standardmain} and \cref{standardmainqt} provide such sufficient conditions.

\section{The MGE Production Model}\label{secmodel}
In this section, we present an MGE model with production, bads and general externalities.
%We provide a rigorous description of the model as well as an extensive discussion of several key features. 
To characterize consumers' preferences, we introduce the following definition given in \citet{hilden74}. 

\begin{definition}\label{defpref} (\citet{hilden74})
The set $\mathcal{P}$ of strict preferences on $\Reals^{\ell}$ is the set of pairs $(X, \succ)$, where
the \emph{consumption set} $X\subset \NNReals^{\ell}$ is closed and convex; and $\succ$ is a continuous, irreflexive and asymmetry\footnote{A preference $\succ$ on $X$ is \emph{continuous} if $\{(x,y) \in X\times X: x \succ y\}$ is relatively open in $X\times X$. Irreflexive means that $a\not\succ a$. Asymmetry means that if $a\succ b$, then $b\not\succ a$. } binary relation defined on $X$. 
\end{definition}

%For every $(X, \succ)\in \mathcal{P}$ and $a, b\in X$, $a\succ b$ means $a$ is strictly preferred to $b$. 
Note that we require neither completeness nor transitivity of $\succ$ in \cref{defpref}. 
 $\mathcal{P}$ is a compact metric space in the closed convergence topology (\citet{hilden74}). For two elements $x_1, x_2\in \Reals^{\ell}$, we abuse notation and write $(x_1, x_2)\in (X, \succ)$ if $x_1,x_2\in X$ and $x_1\succ x_2$. A preference $P=(X, \succ)$ is \emph{convex} if $\{y\in X: y\succ x\}$ is convex for every $x\in X$. 
Let $\mathcal{P}_{H}\subset \mathcal{P}$ denote the set of convex preferences from $\mathcal{P}$.
Then $\mathcal{P}_{H}$ is a closed subset of $\mathcal{P}$ with respect to the closed convergence topology.  
Let $\Delta=\{p\in \Reals^{\ell}: \|p\|=\sum_{k=1}^{l}|p_k|=1\}$ be the set of all prices. Note that we allow negative prices. 
 $\cK(\NNReals^{\ell})$  denotes the set of all closed and convex subsets of $\NNReals^{\ell}$, which is a compact metric space under the closed convergence topology.

\begin{definition}\label{defmeasureeco}
A measure-theoretic production economy is a list
\[
\mathcal E\equiv \{(X, \succ_{\omega}, P_{\omega}, e_{\omega}, \theta)_{\omega\in \Omega}, (Y_j)_{j\in J}, (\Omega, \cB, \mu)\}\nonumber \mbox{ such that}
\] 
\begin{enumerate}[(i),leftmargin=*]
\item $(\Omega,\cB,\mu)$ is a probability space of consumers;

\item $J$ is a finite set of producers;

\item $X: \Omega\to \cK(\NNReals^{\ell})$ is a measurable function such that $X(\omega)\neq \emptyset$ for all $\omega\in \Omega$. $X(\omega)$ is the consumption set for consumer $\omega$. We sometimes write $X_{\omega}$ for $X(\omega)$;  \label{item-consumption_set}

\item A producer $j\in J$ has a non-empty production set $Y_j\subset \Reals^{\ell}$. Let $Y=\prod_{j\in J}Y_j$;

\item the set of allocations is $\mathcal{A}=\{x\in \mathcal{L}^{1}(\Omega, \NNReals^{\ell}): x(\omega)\in X_{\omega} \mbox{ almost surely}\}$, which is equipped with the $\mathcal{L}^{1}$ norm topology;

\item \label{chapreferglobal} Let $M_{\omega}=\mathcal{L}^{1}(\Omega, \NNReals^{\ell})\times Y\times \Delta\times X_{\omega}$. The \emph{global preference relation} of agent $\omega$ is $\succ_{\omega}\subset M_{\omega}\times M_{\omega}$. For $m,m'\in M_{\omega}$, $m\succ_{\omega} m'$ means that the agent $\omega$ strictly prefers $m$ over $m'$. $\succ_{\omega}$ represents the agent's preference on the other agents' consumption, production, prices, and own-consumption. $\succ_{\omega}$ is essential for studying welfare properties and Pareto rankings among equilibria; 

\item\label{chaprefer} The preference map $P_{\omega}: \mathcal{L}^{1}(\Omega, \NNReals^{\ell})\times Y\times \Delta\to \{X_{\omega}\}\times
\PowerSet(X_{\omega}\times X_{\omega})$\footnote{$\PowerSet(X_{\omega}\times X_{\omega})$ is the power set of $X_{\omega}\times X_{\omega}$.} is derived from the global preference relation $\succ_{\omega}$:
\[
P_{\omega}(x,y,p)=(X_{\omega},\{(a, b)\in X_{\omega}\times X_{\omega} | (x,y,p,a)\succ_{\omega}(x,y,p,b)\}).\nonumber
\]
For every $\omega\in \Omega$, $P_{\omega}$ satisfies:
\begin{itemize}
    \item The range of $P_{\omega}$ is contained as a subset in $\mathcal{P}$. By \cref{defpref}, we can write $P_{\omega}(x,y,p) = (X_{\omega},\succ_{x,y,\omega,p})$;
    \item $P_{\omega}$ is continuous in the norm topology on $\mathcal{L}^{1}(\Omega, \NNReals^{\ell})\times Y\times \Delta$;
    %\footnote{\label{preferftnote}Let $d$ be the metric induced by the norm topology on $\mathcal{L}^{1}(T\setminus \{t\}, \NNReals^{\ell})\times Y\times \Delta$ and $d_{\mathrm{cv}}$ denote the metric on $\mathcal{P}$ that generates the closed convergence topology. For all $\epsilon>0$ and all $(f, y, p)\in \mathcal{L}^{1}(T\setminus \{t\}, \NNReals^{\ell})\times Y\times \Delta$, there exists $\delta>0$ such that for all $(f', y', p')\in \mathcal{L}^{1}(T\setminus \{t\}, \NNReals^{\ell})\times Y\times \Delta$ with $d\big((f,y,p), (f',y',p')\big)<\delta$, we have $d_{\mathrm{cv}}\big(P_{t}(x,y,p), P_{t}(x',y',p')\big)<\epsilon$, where $x,x'$ are points in $\mathcal{L}^{1}(T, \NNReals^{\ell})$ such that $x(i)=f(i)$ and $x'(i)=f'(i)$ for all $i\neq t$. Note that, if $(T,\cB,\mu)$ is atomless, then this condition is equivalent to $P_{t}$ being continuous in the norm topology $\mathcal{L}^{1}(T, \NNReals^{\ell})\times Y\times \Delta$.};
\end{itemize}

%\item For each $t\in T$, $P_{t} : \mathcal{L}^{1}(T, \NNReals^{l})\times Y\times \Delta \rightarrow \mathcal{P}$ denotes the preference relation of consumer $t$ where $P_{t}(x,y,p) = (X_t,\succ_{x,y,t,p})$. $P$ is continuous in the norm topology on $\mathcal{L}^{1}(T, \NNReals^{l})\times Y\times \Delta$ and measurable in $T$ \footnote{Let $f, f'\in \mathcal{L}^{1}(T, \NNReals^{l})$ be such that $f(i)=f'(i)$ for all $i\neq t$. As $(T,\cB,\mu)$ is an atomless probability space and $P$ is continuous in the norm topology on $\mathcal{L}^{1}(T, \NNReals^{l})\times Y\times \Delta$, we have $P_{t}(f, y, p)=P_{t}(f', y, p)$ for all $(y, p)\in Y\times \Delta$. }; 

\item $\theta\in  \mathcal{L}^{1}(\Omega, \NNReals^{|J|})$ is the density of shareholdings of firms by consumers such that $\int_{\Omega}\theta(\omega)(j)\mu(\dee \omega)=1$ for all $j\in J$, where $\theta(\omega)(j)$ is the $j$-th coordinate of $\theta(\omega)$. We sometimes write $\theta_{\omega j}$ for $\theta(\omega)(j)$; 

\item $e\in \mathcal{L}^{1}(\Omega, \NNReals^{\ell})$ is the initial endowment map. Hence, $e(\omega)$ is the density of the initial endowment of the consumer $\omega$. 
\end{enumerate}
\end{definition}

\vspace{-9pt}

\begin{remark}\label{modelremark}
Our model, defined in \cref{defmeasureeco}, has the following features:
\begin{enumerate}[leftmargin=*]
     \item %Externalities are the essential motivation for regulating pollution emissions. 
     Items \ref{chapreferglobal} and \ref{chaprefer} characterize each consumer's preference through the global preference relation $\succ_{\omega}$ and the preference map $P_{\omega}$. The global preference relation $\succ_{\omega}$ represents the consumer's preference on the choices of all consumers, production, prices and her own consumption bundles. The consumer, however, has no control over other consumers' choices, production and prices. Hence, given all other consumers' choices, production and prices, the consumer chooses her bundle according to the preference map $P_{\omega}$. For the existence of equilibrium, one only needs the preference map $P_{\omega}$. Hence, we impose regularity conditions directly on $P_{\omega}$. However, the preference relation $\succ_{\omega}$ is essential for studying welfare properties and potential Pareto improvements of the equilibrium;\footnote{For welfare analysis on economies with externalities, see \citet{ad25} and \citet{anderson25pareto}.} 
    \item We do not require $(\Omega, \cB, \mu)$ to be atomless, so finite weighted production economies are special cases of measure-theoretic production economies defined in  \cref{defmeasureeco}.   
\end{enumerate}
\end{remark}

For each $\omega\in \Omega$ and $(x,y,p)\in \mathcal{L}^{1}(\Omega, \NNReals^{\ell})\times Y\times \Delta$, the budget set $B_\omega(y,p)$, demand set $D_{\omega}(x,y,p)$ and quasi-demand set $\bar{D}_{\omega}(x,y,p)$ are defined as: 
\begin{eqnarray}
 B_{\omega}(y,p)
&=&\left\{z\in X_{\omega}: p\cdot z \leq p\cdot e(\omega) + \sum_{j\in J}\theta_{\omega j}p\cdot y(j)\right\},\nonumber\\
D_{\omega}(x,y,p)&=&\{z \in B_{\omega}(y,p): w \succ_{x,y,\omega,p} z\implies p\cdot w>p\cdot e(\omega) + \sum_{j\in J}\theta_{\omega j}p\cdot y(j)\},\nonumber\\
\bar{D}_{\omega}(x,y,p)&=&\{z \in B_{\omega}(y,p): w \succ_{x,y,\omega,p} z\implies p\cdot w\geq p\cdot e(\omega)+\sum_{j\in J}\theta_{\omega j}p\cdot y(j)\}.\nonumber
\end{eqnarray}
For each $j\in J$, let
$S_j(p)=\argmax_{z \in Y_j} p\cdot z $ denote the (possibly empty) supply set at $p\in \Delta$.
Note that producers are profit maximizers and their profits depend only on prices and their own production.\footnote{We assume producers are profit maximizers. \citet{makarov81} established a general equilibrium existence theorem which allows for objectives other than profit maximization.} We now give the definition of (quasi)-equilibrium.

\begin{definition}\label{def_mameqprod}
Let $\mathcal E=\{(X, \succ_\omega, P_\omega, e_\omega, \theta)_{\omega\in \Omega}, (Y_j)_{j\in J}, (\Omega, \cB, \mu)\}$ be a measure theoretic production economy.  
A quasi-equilibrium is $(\bar{x}, \bar{y}, \bar{p})\in \mathcal{A}\times Y\times \Delta$ such that:
\begin{enumerate}[(i),leftmargin=*]
\item $\bar{x}(\omega)\in \bar{D}_{\omega}(\bar{x},\bar{y}, \bar{p})$ for $\mu$-almost all $\omega\in \Omega$;
\item $\bar{y}(j)\in S_{j}(\bar{p})$ for all $j\in J$;
\item \label{nfreefeasible}$\int_{\Omega}\bar{x}(\omega)\mu(\dee \omega)-\int_{\Omega}e(\omega)\mu(\dee \omega)-\sum_{j\in J} \bar{y}(j)=0$.
\end{enumerate}
An equilibrium $(\bar{x}, \bar{y}, \bar{p})\in \mathcal{A}\times Y\times \Delta$ is a quasi-equilibrium with $\bar{x}(\omega)\in D_{\omega}(\bar{x},\bar{y}, \bar{p})$ for $\mu$-almost all $\omega\in \Omega$.
\end{definition}

\subsection{Free Disposal in Production}\label{secproblem}
%Free disposal in production conflicts with proper accounting of bads:
%\begin{itemize}[leftmargin=*]
%    \item The (quasi)-equilibrium price is non-negative if any firm has free disposal in production;
%    \item If the preference map $P_{\omega}$ is independent of production for almost all $\omega\in \Omega$, and any firm has free disposal in production, every free-disposal equilibrium\footnote{A free-disposal equilibrium is an equilibrium in \cref{def_mameqprod}, but with \cref{nfreefeasible} replaced by the feasibility constriant $\int_{\Omega}\bar{x}(\omega)\mu(\dee \omega)-\int_{\Omega}e(\omega)\mu(\dee \omega)-\sum_{j\in J} \bar{y}(j)
    %\leq 0$} can be disguised as a non-free-disposal equilibrium.
%\end{itemize}
In this section, we provide a rigorous treatment of the problematic aspects of the free disposal in production assumption, discussed in \cref{fdintro}.
%\begin{definition}\label{mpfpdef}
Recall a firm $j\in J$ has free disposal in production if, given $y\in Y_j$ and $z<y$, then $z\in Y_j$.
%\end{definition}

\begin{customprop}{1}\label{nnpricelemma}
Let $\mathcal{E}=\{(X, \succ_{\omega}, P_{\omega}, e_{\omega}, \theta_{\omega})_{\omega\in \Omega}, (Y_j)_{j\in J}, (\Omega, \cB, \mu)\}$ be a measure-theoretic production economy. If there is a firm $j\in J$ that has free disposal in production, then the equilibrium price is non-negative.
\end{customprop}

\begin{proof}
Let $\bar{p}$ denote the equilibrium price. 
Suppose $\bar{p}$ has a negative coordinate. 
Without loss of generality, we assume that $\bar{p}_{1}<0$. 
Note that each producer is profit maximizing. 
As firm $j$ has free disposal in production, firm $j$'s profit is unbounded, which is a contradiction. 
\end{proof}

%\cref{nnpricelemma} implies that, in the presence of free disposal in production, negative prices cannot be used to control bads. 

%Free-disposal equilibrium implies that resources can be discarded without any cost. 
The next theorem shows that, with free disposal in production, free-disposal equilibria can often be disguised as non-free-disposal equilibria.

\begin{customprop}{2}\label{freenofree}
Let $\mathcal{E}=\{(X, \succ_{\omega}, P_{\omega}, e_{\omega}, \theta_{\omega})_{\omega\in \Omega}, (Y_j)_{j\in J}, (\Omega, \cB, \mu)\}$ be a measure-theoretic economy with production. Suppose 
\begin{enumerate}[(i),leftmargin=*]
    \item For all $\omega\in \Omega$, the preference map $P_{\omega}$ is independent of the production;
    \item There is a firm $j_0\in J$ that has free disposal in production. 
\end{enumerate}
If $(\bar{x},\bar{y},\bar{p})$ is a free-disposal (quasi)-equilibrium and the value of the excess supply is $0$,\footnote{The value of the excess supply is $0$ if and only if the value of almost all consumers' consumption bundle is on the budget line, which is implied by assuming locally non-satiated preferences.} then there exists $\bar{y}'$ such that $(\bar{x}, \bar{y}', \bar{p})$ is a non-free-disposal (quasi)-equilibrium.
\end{customprop}
\begin{proof}
Suppose $(\bar{x}, \bar{y}, \bar{p})$ is a free-disposal equilibrium.\footnote{We only prove the case where $(\bar{x}, \bar{y}, \bar{p})$ is a free-disposal equilibrium. The proof of the quasi-equilibrium case is essentially the same.} Let 
\[
w=\int_{\omega\in \Omega}e(\omega)\mu(\dee \omega)+\sum_{j\in J}\bar{y}(j)-\int_{\omega\in \Omega}\bar{x}(\omega)\mu(\dee \omega)\geq 0.\nonumber
\]
Without loss of generality, assume that firm $1$ has free disposal in production and let $\bar{y}'(1)=\bar{y}(1)-w$. 
Since firm $1$ has free disposal in production, we have $\bar{y}'(1)\in Y_1$. 
Form $\bar{y}'$ from $\bar{y}$ by substituting $\bar{y}'(1)$ for $\bar{y}(1)$. 
We show that $(\bar{x}, \bar{y}', \bar{p})$ is a non-free-disposal equilibrium. 
\begin{enumerate}[leftmargin=*]
    \item Clearly, we have $\int_{\omega\in \Omega}\bar{x}(\omega)\mu(\dee \omega)-\sum_{j\in J}\bar{y}'(j)-\int_{\omega\in \Omega}e(\omega)\mu(\dee \omega)=0$;
    \item As the value of the excess supply is $0$, we have $\bar{p}\cdot w=0$. So $B_{\omega}(\bar{y}, \bar{p})=B_{\omega}(\bar{y}', \bar{p})$ for $\omega\in \Omega$. As the preference map is independent of the production, we conclude that $\bar{x}(\omega)\in D_{\omega}(\bar{x}, \bar{y}, \bar{p})$ for almost all $\omega\in \Omega$;
    \item As $\bar{p}\cdot w=0$, we have $\bar{p}\cdot \bar{y}_{1}=\bar{p}\cdot \bar{y}'_{1}$. Hence, all firms are profit maximizing.
\end{enumerate}
Hence, $(\bar{x}, \bar{y}', \bar{p})$ is a non-free-disposal equilibrium. 
\end{proof}

%As a result of \cref{nnpricelemma} and \cref{freenofree}, it is inappropriate to assume free disposal in production in any measure-theoretic general equilibrium model with bads.  

\section{Main Results and Examples}\label{secmainresult}

In this section, we show that the production economy in \cref{defmeasureeco} has an equilibrium. %In the special case that the consumer space is a finite number of atoms, 
%we call the measure-theoretic production economy a \emph{weighted production economy}. 
We first establish the existence of equilibrium for 
finite weighted production economies in \cref{finitemainexpd}, which is closely related to Proposition 3.2.3 in \citet{fl03bk}. Furthermore, in \cref{standardmain}, we prove the existence of equilibrium in measure-theoretic economies with bads and preference externalities  under moderate regularity conditions.

\begin{assumption}\label{assumptionsurvival}
Let $\mathcal{E}$ be a measure-theoretic production economy as in \cref{defmeasureeco}: 
\begin{enumerate}[(i),leftmargin=*]
    \item\label{esto0} there exists a set $\Omega_0\subset \Omega$ of positive measure such that, for every $\omega\in \Omega_0$, the set $X_{\omega} - \sum_{j\in J}\theta_{\omega j}Y_j$ has non-empty interior $U_{\omega}\subset \Reals^{\ell}$ and $e(\omega)\in U_{\omega}$;
\item \label{cdexplain} there exists 
a commodity $s\in \{1,2,\dotsc,\ell\}$ such that:
\begin{itemize}[leftmargin=*]
    \item for every $\omega\in \Omega_0$, the projection $\pi_{s}(X_{\omega})$ of $X_{\omega}$ to the $s$-th coordinate is unbounded, and the consumer $\omega$ has a strongly monotone preference on the commodity $s$;\footnote{Given any $(x,y,p)\in \mathcal{L}^{1}(\Omega,\NNReals^{\ell})\times Y\times \Delta$, for every $a,a'\in X_{\omega}$ such that $a_s>a'_s$ and $a_t=a'_t$ for all $t\neq s$, we have $(a,a')\in P_{\omega}(x,y,p)$.}
    \item for almost all $\omega\in \Omega$, there is an open set $V_{\omega}$ containing the $s$-th coordinate $e(\omega)_{s}$ of $e(\omega)$ such that $(e(\omega)_{-s},v)\in X_{\omega} - \sum_{j\in J}\theta_{\omega j}Y_j$\footnote{$(e(\omega)_{-s},v)$ is the vector such that its $s$-th coordinate is $v$, and its $t$-th coordinate is the same as the the $t$-th coordinate of $e(\omega)$ for all $t\neq s$.} for all $v\in V_{\omega}$.
\end{itemize}
\end{enumerate}
\end{assumption}

%The version of \cref{assumptionsurvival} for finite production economies was introduced in \citet{ad25}. 
\cref{assumptionsurvival} is closely related to the classical survival assumption $e_{\omega}\in X_\omega-\sum_{j\in J}\theta_{\omega j}Y_j$.\footnote{The survival assumption implies that a consumer can survive without participating in any exchanges using her initial endowment and shares in production. In particular, a consumer who supplies labor in an equilibrium is able to survive, and hence supply that labor.} 
Following the literature, we could have assumed $e_{\omega}\in \mathrm{int}\big(X_\omega-\sum_{j\in J}\theta_{\omega j}Y_j\big)$,\footnote{As in the previous literature, the interior is taken with respect to the topology of $\Reals^{\ell}$, not with respect to the subspace topology. Hence, the strengthened survival assumption implies the set $X_\omega-\sum_{j\in J}\theta_{\omega j}Y_j$ has non-empty interior in $\Reals^{\ell}$.  There exist, however, a few papers relaxing the assumption $e_{\omega}\in \mathrm{int}\big(X_\omega-\sum_{j\in J}\theta_{\omega j}Y_j\big)$. See %e.g. \citet{de62}, \citet{ah71}, \citet{Berg76} and \citet{ga79}, \citet{McKenzie81},  and \citet{Florig99}. Also see 
\citet{fl03bk} for a detailed discussion.} 
but this stronger assumption is economically restrictive; it fails if there is a consumer with no shareholding of any firm and the projection of the consumer's consumption set to some coordinate is $\{0\}$.\footnote{Poor people generally do not have any shareholdings of private firms. Moreover, individuals may be incapable of consuming certain bads. Thus, it is reasonable to assume that the projections of consumers' consumption sets to some commodities are $\{0\}$. }
\cref{assumptionsurvival}  allows for consumers who are not endowed with certain commodities, have no shareholdings of private firms and are unable to consume certain bads.
\cref{esto0} in \cref{assumptionsurvival} requires there be a positive measure set of consumers $\Omega_0$ whose endowments are in the interior of $X_{\omega}-\sum_{j\in J}\theta_{\omega j}Y_j$,\footnote{\cref{esto0} in \cref{assumptionsurvival} is generally satisfied if there is a group of rich consumers such that, for each commodity, every consumer in the group is either endowed with a positive amount of the commodity or has a positive shareholding of a firm that is capable of producing the commodity.} which implies that every consumer from the group $\Omega_0$ has a strictly positive budget at every quasi-equilibrium. \cref{esto0} in \cref{assumptionsurvival} and the first bullet of \cref{cdexplain} in \cref{assumptionsurvival} imply that the quota quasi-equilibrium price of the commodity $s$ is strictly positive. Hence, by the second bullet of \cref{cdexplain} in \cref{assumptionsurvival}, every consumer has a positive budget at every quasi-equilibrium, implying that every quasi-equilibrium is an equilibrium;

\begin{definition}\label{def_attainable}
The set of attainable consumption-production pairs is
\[
\cO=\left\{(x,y)\in \mathcal{L}^{1}(\Omega, \NNReals^{\ell})\times Y: \int_{\Omega}x(\omega) \mu(\dee \omega)-\int_{\Omega}e(\omega)\mu(\dee \omega)-  \sum_{j\in J}y(j)=0\right\}.\nonumber
\]
\end{definition}
%We start with the existence of equilibrium for weighted production economies:

\begin{customthm}{1}\label{finitemainexpd}
Let $\mathcal{E}=\{(X, \succ_\omega, P_\omega, e_\omega, \theta_{\omega})_{\omega\in \Omega}, (Y_j)_{j\in J}, (\Omega, \PowerSet(\Omega), \mu)\}$ be a weighted production economy as in \cref{defmeasureeco}. Suppose $\mathcal{E}$ satisfies \cref{assumptionsurvival}, and the following conditions: 
\begin{enumerate}[(i),leftmargin=*]
    \item for almost all $\omega\in \Omega$, $P_{\omega}$ takes value in $\mathcal{P}_{H}$;\footnote{As noted in \citet{fl03bk}, this condition can be weakened to the following condition: for each $(x,y,p)\in \cO\times \Delta$ and all $\omega\in \Omega$, $(x(\omega),x(\omega))\not\in \conv(P_{\omega}(x,y,p))$, where $\conv(P_{\omega}(x, y,p))$ denotes the convex hull of $P_{\omega}(x,y,p)$.} 
    \item \label{satiate}for almost all $\omega\in \Omega$, for each $(x,y)\in \cO$ with $x_\omega\in X_\omega$, there exists $u\in X_\omega$ such that $(u, x_\omega)\in \bigcap_{p\in \Delta}P_{\omega}(x,y,p)$;
    \item $\bar{Y}$ is closed, convex, and  $\bar Y \cap (-\bar Y)=\bar Y \cap \NNReals^{\ell}=\{0\}$, where $\bar Y =\left\{\sum_{j\in J}{y}(j): y\in Y\right\}$ is the aggregate production set.
\end{enumerate}
Then, $\mathcal{E}$ has an equilibrium.
\end{customthm}

\begin{remark}
\cref{finitemainexpd} is the weighted version of Proposition 3.2.3 in \citet{fl03bk} with \emph{Florenzano's disposal cone} being the singleton $\{0\}$, and it plays a key role in the proof of existence of equilibrium for measure-theoretic production economies.  
%The conditions of \cref{finitemainexpd} are similar to classical assumptions in the GE literature: 
%\begin{enumerate}[leftmargin=*]
%\item Closedness and convexity are classical assumptions on consumption sets and on the aggregate production set. These conditions, in particular, imply that commodities are perfectly divisible. Irreversibility and no free production\footnote{Irreversibility in production means $\bar{Y}\cap (-\bar{Y})=\{0\}$. No free production means $\bar{Y}\cap \NNReals^{\ell}=\{0\}$.} in the aggregate production set imply the compactness of attainable consumption sets and relative compactness of attainable production sets, which play essential roles in establishing the existence of equilibrium, as in \citet{de59};
Our formulation allows for quite general externalities in consumers' preferences.
The consumers' preference maps $P_{\omega}$ are assumed to be continuous with respect to the closed convergence topology, hence lower hemicontinuous if viewed as correspondences. If the preferences are price-independent, \cref{satiate} of \cref{finitemainexpd} is non-satiation at every attainable consumption-production pair. 
\end{remark}

%The unweighted version of \cref{finitemainexpd} is very similar to Proposition 3.2.3 in \citet{fl03bk}. Compared with the existing literature on general equilibrium theory with finitely many consumers and  bads\footnote{\citet{McKenzie81}used the non-free-disposal equilibrium as the equilibrium concept. See also \citet{Berg71}, \citet{Hart75} and \citet{pole93}.}, \cref{finitemainexpd} is more general since (1): it allows for externality on each consumer's preference, (2): it assumes neither  free-disposal in production nor monotone preferences and (3): the disposal cone allows the society to choose which commodities are disposed, and in what quantity. As we will see, \cref{finitemainexpd} also plays a crucial in the proof of existence of (quasi)-equilibrium for measure-theoretic production economies. 

\subsection{Equilibrium in Measure-theoretic Production Economies}
Fix a measure-theoretic production economy $\mathcal E=\{(X, \succ_\omega, P_\omega, e_\omega, \theta)_{\omega\in \Omega}, (Y_j)_{j\in J}, (\Omega, \cB, \mu)\}$ as in \cref{defmeasureeco}. 

\begin{assumption}\label{assumptionmseco}
The consumer space $\Omega$ of $\mathcal{E}$
is a Polish space\footnote{That is, $\Omega$ is a complete separable metric space.} endowed with the Borel $\sigma$-algebra $\BorelSets \Omega$ and $\mu$ is a Borel probability measure on $\Omega$. %As a result, $(\Omega, \BorelSets \Omega, \mu)$ is a  Radon probability space. Moreover, the map $(\omega,x,y,p)\to P_{\omega}(x,y,p)$ is measurable.
\end{assumption}

For $\epsilon>0$, the set of \emph{$\epsilon$-attainable consumption-production pairs} is
\[
\cO_{\epsilon}=\left\{(x',y')\in \mathcal{L}^{1}(\Omega, \NNReals^{\ell})\times Y: \left|\int_{\Omega} x'(\omega)\mu(\dee \omega)-\int_{\Omega}e(\omega)\mu(\dee \omega)-\sum_{j\in J}{y'}(j)\right|<\epsilon\right\}.\nonumber
\]

%We now provide sufficient conditions to address the failure of uniform integrability that arises in \cref{haexample}. 

\begin{assumption}\label{assumptionalocbound}
There is some $k\leq \ell$ such that:
\begin{enumerate}[(i),leftmargin=*]
    \item\label{projform} Let $\proj{k}$ denote the projection onto the first $k$-th coordinates. For every $\omega\in \Omega$, we have $X(\omega)=\proj{k}\big(X(\omega)\big)\times \NNReals^{\ell-k}$;
    \item\label{badintegrable}  The mapping $\proj{k}\circ X: \Omega\to \cK(\NNReals^{k})$ is integrably bounded, i.e., there exists an integrable function $\psi: \Omega\to \NNReals^{k}$ such that %for all $\omega\in \Omega$, 
    for every $x\in \proj{k}\big(X(\omega)\big)$, $x\leq \psi(\omega)$;
    \item\label{goodmonotone} For every $s\in \{k+1,\dotsc,\ell\}$, there exists some $\epsilon_{s}>0$ such that the set
    \[
    \Omega_0^{s}=\{\omega\in \Omega_0: (\forall (x,y,p)\in \cO_{\epsilon_{s}}\times \Delta)(P_{\omega}(x,y,p)\in M_{s})\}\nonumber
    \]
has positive measure, where $\Omega_0\subset \Omega$ is the set in \cref{esto0} of \cref{assumptionsurvival} and $M_{s}$ is the set of preferences that are strongly monotonic in commodity $s$.\footnote{A consumer has a strongly monotonic preference in commodity $s$ if, holding the consumption of all other commodities fixed, the consumer strictly prefers having more of commodity $s$.} 
\end{enumerate}
\end{assumption}

\begin{remark}\label{integralassumptionremark}
Recall that there are, in total, $\ell$ commodities. We divide commodities into two categories: bads and goods. The economic interpretation of \cref{assumptionalocbound} is: 
\begin{enumerate}[leftmargin=*]
    \item The first $0\leq k\leq \ell$ commodities are bads. As discussed in the Introduction, the integrable bound reflects consumers' limited capacities to absorb bads. Note that we do not impose a uniform bound on consumers' consumption of bads. %This reflects that consumers may have different capacities to absorb bads. 
    We also do not require consumers to be unanimous in the designation of commodities as goods or bads.
    %For example, around $1.5\%-6\%$ of the population suffers from hoarding disorder.\footnote{See e.g., \citet{postle19}.} Hence, a set of consumers of positive measure may have strongly monotone preferences over some bads. If these consumers' consumption sets were unbounded with respect to bads, then the equilibrium prices for these bads would necessarily be positive. However, the market prices for bads are usually either $0$ or negative. Moreover, negative prices are critical to provide incentives for controlling bads. Since hoarders have strongly monotone preferences over some bads, they often consume as much bads as they can. Thus, the limited capacity for consuming bads\footnote{Once a hoarder's house is completely full, there is no room for additional hoarding.} for agents who have hoarding disorder is reflected by the integrable bound on these agents' consumption sets rather than their budget constraints;       
    \item The commodities $k+1,\dotsc, \ell$ are goods. %Ownership of goods conveys the right to consume the commodity, but does not entail the obligation to dispose of them. 
    We allow for arbitrarily large consumption of goods. Furthermore, for every good, there is a set of consumers with positive measure whose preferences for that good are strongly monotonic. We do not require any consumer to have a preference that is monotonic over multiple goods. Our formulation allows, for example, individuals who derive no utility from a subset of the goods, and thus whose demands are zero for that subset of goods, regardless of the prices of those goods.  
\end{enumerate}
\end{remark}

%The convexity of the quasi-demand set plays an important role in establishing the existence of equilibrium in measure-theoretic production economies. 
%{\color{blue}To ensure convexity of the quasi-demand set, we restrict ourselves to transitive, negatively transitive and convex preferences from $\mathcal{P}$. 
%{\color{blue}Let $\mathcal{P}^{-}\subset \mathcal{P}$ be the set of all transitive and negatively transitive preferences of $\mathcal{P}$. That is: 
%\[
%\mathcal P^{-}=\{(C,\succ)\in \mathcal P: (\forall x,y,z\in C) [(x\succ y\succ z \implies x\succ z )\wedge (x\nsucc y\nsucc z \implies x\nsucc z)]\}.\nonumber
%\]
%Let $\mathcal P_H^{-}$ be the set of transitive, negatively transitive and convex preferences from $\mathcal{P}$. %In \cref{cqdappendix}, we show that if a consumer's preference map takes value in $\mathcal P_H^{-}$, then the consumer's quasi-demand set is convex.}

%\cref{assumptionalocbound} is compatible with the presence of bads as commodities. It, however, implies that each consumer's consumption set is bounded, although not uniformly bounded. \citet{Noguchi2006} established the existence of equilibrium in continuum economies with externalities and unbounded consumption sets, but under the additional assumption of strongly monotone preferences. Our existence of equilibrium result remains if we replace \cref{assumptionalocbound} by assuming strongly monotone preferences. 

As in \citet{cornet05}, we assume that consumers' preferences are continuous with respect to the weak topology on $\mathcal{L}^{1}(\Omega, \NNReals^{\ell})$. 

\begin{assumption}\label{assumptionwkcts}
For $\omega\in \Omega$, the preference map $P_{\omega}: \mathcal{L}^{1}(\Omega,\NNReals^{\ell})\times Y\times \Delta\to \mathcal{P}$ is continuous with respect to the product of weak topology on $\mathcal{L}^{1}(\Omega,\NNReals^{\ell})$ and the norm topology on $Y\times \Delta$.
\end{assumption}

The main result of this section is: 

\begin{customthm}{2}\label{standardmain}
Let $\mathcal{E}$ be a measure-theoretic production economy as in \cref{defmeasureeco}. Suppose $\mathcal{E}$
satisfies Assumptions \ref{assumptionsurvival},  \ref{assumptionmseco}, \ref{assumptionalocbound}, \ref{assumptionwkcts}, and the following conditions:
\begin{enumerate}[(i),leftmargin=*]
\item \label{tran}for almost all $\omega\in \Omega$, $P_\omega$ takes value in  $\mathcal{P}_{H}^{-}$, the set of transitive, negatively transitive and convex preferences from $\mathcal{P}$;
\item\label{mersatiate} for some $\epsilon>0$, for almost all $\omega\in \Omega$ and all $(x,y)\in \cO_{\epsilon}$ such that $x(\omega)\in X_{\omega}$, there exists $u\in X_{\omega}$ such that $(u, x(\omega))\in \bigcap_{p\in \Delta}P_{\omega}(x,y,p)$;

\item The aggregate production set $\bar{Y}$ is closed and convex, $\bar Y \cap (-\bar Y)=\bar Y \cap \NNReals^{\ell}=\{0\}$, and $Y_j$ is closed for all $j\in J$.\footnote{We do not need $Y_j$ to be closed if consumers' preferences only depend on allocations, aggregate production and prices.}
\end{enumerate}
Then, $\mathcal{E}$ has an equilibrium. 
\end{customthm}

\begin{remark}\label{mestremark}
\cref{standardmain} is the first equilibrium existence theorem for measure-theoretic GE models with bads and externalities. We briefly discuss Assumptions \ref{assumptionalocbound}, \ref{assumptionwkcts} and \cref{mersatiate}:
\begin{enumerate}[leftmargin=*]
    \item \cref{assumptionalocbound} reflects that consumers typically have limited capacity to absorb bads. We do not rule out the possibility that firms have the capacity to eliminate bads as part of the production process.\footnote{In \cref{garbageexample}, consumers have limited capability to consume garbage and there is a firm with the technology to eliminate garbage. At the equilibrium, all consumers consume a small quantity of garbage in aggregate while a firm eliminates a large quantity of garbage.} 
    \cref{badintegrable} of \cref{assumptionalocbound} ensures the integrability of the consumption of bads at the candidate equilibrium, which then implies integrability of consumers' income. 
    \cref{goodmonotone} of \cref{assumptionalocbound} implies the equilibrium price on goods are positive, so the integrabilty of income implies the integrability of the equilibrium allocation of goods. Thus, \cref{assumptionalocbound} allows us to overcome the failure of uniform integrability in \citet{ha05};    
    \item \cref{assumptionwkcts} is stronger than assuming the preference map is continuous with respect to the $\mathcal{L}^{1}$ norm topology on $\mathcal{L}^{1}(\Omega, \NNReals^{\ell})$. \cref{assumptionwkcts} allows us to push down an S-integrable function to construct an allocation in the original standard economy;   
    \item To ensure convexity of the quasi-demand set, we restrict ourselves to transitive, negatively transitive and convex preferences from $\mathcal{P}$ in \cref{tran}.
    \item \cref{mersatiate} of \cref{standardmain} is stronger than \cref{satiate} of \cref{finitemainexpd}, since $(x,y)\in \cO_{\epsilon}$ may not be an attainable consumption-production pair. The proof of the equilibrium existence result for measure-theoretic production economies requires this strengthening.\footnote{This stronger condition is needed since a hyperfinite attainable consumption-production pair may not be an exact $\NSE{}$attainable consumption-production pair.} \cref{mersatiate} of \cref{standardmain} is implied by non-satiation at every consumption-production pair. 
\end{enumerate}   
\end{remark}

We conclude this subsection with the following example in which consumers have limited capability to consume bads. Note that, although consumers disagree on which commodities are bads, the example satisfies the assumptions of \cref{standardmain}.

\begin{customexp}{3}\normalfont\label{garbageexample}
Let $\mathcal{F}=\{(X, \succ_{\omega}, P_{\omega}, e_{\omega}, \theta)_{\omega\in \Omega}, (Y_j)_{j\in J}, (\Omega, \cB,\mu)\}$ be: 
\begin{enumerate}[leftmargin=*]
    \item The economy $\mathcal{E}$ has three commodities: garbage, human capital and consumption good, which we denote by $c_1$, $c_2$ and $c_3$;
    \item The consumer space is the Lebesgue measure space on $[0, 1]$. For each $\omega\in \Omega$, consumer $\omega$'s consumption set is $[0,\omega]\times \NNReals^{2}$, the endowment is $e(\omega)=(0,2\omega,0)$. For $\omega\in [0, 0.5]\cup [0.6, 1]$, the utility function is $u_{\omega}(c_1,c_2,c_3)=\ln{(c_3)}-c_1$.\footnote{The agent's utility is $-\infty$ if $c_3=0$. This means that the agent will never choose to consume $0$ units of the third commodity unless that is the only choice.} For $\omega\in (0.5, 0.6)$, the utility function is $u_{\omega}(c_1,c_2,c_3)=\ln{(c_3)}+c_1$; these consumers have hoarding disorder;\footnote{$1.5\%-6\%$ of the population has hoarding disorder; see \citet{postle19}.}    
    \item There are two producers with production sets $Y_1=\{(r,-r,r): r\in \NNReals\}$ and $Y_2=\{(-r,-r,0): r\in \NNReals\}$;
    \item The shareholding $\theta\in \mathcal{L}^{1}(\Omega, \NNReals^{2})$ is $\theta(\omega)=(1,1)$ for all $\omega\in \Omega$. 
\end{enumerate}
In this example, consumers have limited capacity to absorb garbage, some consumers have hoarding disorder, and the second firm has a technology to use human capital to eliminate garbage. Hence, consumers disagree on the designation of commodities as goods or bads. 
We show that $\cF$ has a unique equilibrium, in which the price of garbage is negative even though some consumers have strongly monotone preferences over garbage; these consumers are constrained by the integrable bound on consumption of garbage. 
\begin{claim}\label{equiprc}
    If equilibrium exists, then the equilibrium price must be $(-\frac{1}{4},\frac{1}{4},\frac{1}{2})$.
\end{claim}
\begin{proof}
Let $\bar{p}=(\bar{p}_{1},\bar{p}_{2},\bar{p}_{3})\in \Delta$ be an equilibrium price. 
As both firms have linear technology, both firms' profits must be $0$ at equilibrium. 

Suppose the equilibrium price $\bar{p}_{2}$ of human capital is non-positive. The second firm's production set implies that the equilibrium price $\bar{p}_{1}$ of garbage must be non-negative. The consumers' utility functions and consumption sets imply that the equilibrium price $\bar{p}_{3}$ of consumption good is non-negative. As $\bar{p}\in \Delta$, the first firm's profit at equilibrium is infinite, a contradiction. Hence, the equilibrium price $\bar{p}_{2}$ must be positive. 

Since consumers do not acquire utility from human capital and the equilibrium price of human capital is positive, all human capital must be consumed by firms. The non-free disposal of garbage implies that both firm must operate at equilibrium.\footnote{If the first firm consumes all the human capital, it generates $1$ unit of garbage. However, the consumers are only capable of consuming $\frac{1}{2}$ unit of garbage in aggregate. If the second firm consumes all the human capital, then there is no garbage for the second firm to eliminate.} 
Hence we have $\bar{p}_{1}-\bar{p}_{2}+\bar{p}_{3}=0$ and $-\bar{p}_{1}-\bar{p}_{2}=0$. Since $\bar{p}\in \Delta$, $\bar{p}=(-\frac{1}{4},\frac{1}{4},\frac{1}{2})$. 
\end{proof}
Since no consumer obtains utility from human capital and the price of human capital is positive, no consumer consumes human capital at equilibrium. Suppose consumer $\omega$'s consumption is $\big(x_{\omega}(1),0,x_{\omega}(3)\big)$. The budget constraint implies:
\[
\frac{1}{2}x_{\omega}(3)-\frac{1}{4}x_{\omega}(1)\leq \frac{1}{2}\omega \iff x_{\omega}(3)\leq \omega+\frac{x_{\omega}(1)}{2}.\nonumber
\]
For $\omega\in (0.5, 0.6)$, the consumer's garbage consumption is $\omega$. 
For $\omega\in [0, 0.5]\cup [0.6, 1]$, the consumer's utility is given by $\ln(\omega+\frac{x_{\omega}(1)}{2})-x_{\omega}(1)$. By taking the derivative, we conclude that consumer $\omega$'s equilibrium consumption is:
\begin{equation}
  \big(x_{\omega}(1),0, x_{\omega}(3)\big) =
    \begin{cases}
      (\omega,0, \frac{3\omega}{2}) & \text{for $\omega\in [0,\frac{1}{3}]$}\\
      (1-2\omega,0,\frac{1}{2}) & \text{for $\omega\in (\frac{1}{3},\frac{1}{2}]$}\\
      (\omega,0,\frac{3\omega}{2}) & \text{for $\omega\in (\frac{1}{2}, 0.6]$}\\
      (0,0,\omega) & \text{for $\omega\in (0.6, 1]$}\\ 
    \end{cases}\nonumber    
\end{equation}
Consumers with low human capital are willing to consume as much garbage as their consumption sets allow in order to generate income to purchase the consumption good. Consumers with medium human capital are willing to consume some garbage, but less than their consumption sets allow. Consumers with hoarding disorder consume as much garbage as their consumption sets allow. Note that the equilibrium price for garbage is negative, even though there is a positive measure set of consumers whose preferences over garbage are strongly monotonic.
Consumers with high human capital and without hoarding disorder are not willing to consume garbage at all; even though these consumers have high capacity to consume bads, they choose not to do so.\footnote{Note that the integrable bound on consumption sets is not binding for non-hoarding-disorder consumers with medium and high human capital, and hence is only needed in this example for consumers with low human capital or hoarding disorder.} Among consumers without hoarding disorder, the income effect is the main factor driving different consumptions of bads.

The aggregate consumption of the consumption good by consumers is $\frac{683}{1200}$, and the aggregate garbage consumption by consumers is $\frac{83}{600}$. Since we require non-free disposal at equilibrium, we conclude that the first firm's equilibrium production is $(\frac{683}{1200},-\frac{683}{1200},\frac{683}{1200})$ while the second firm's equilibrium production is $(-\frac{517}{1200},-\frac{517}{1200}, 0)$. This is the unique equilibrium. At equilibrium, consumers absorb, and the second firm eliminates, garbage, 
but the second firm eliminates much more garbage than all consumers absorb collectively. 
\end{customexp}

\section{Sketch of Proofs}\label{secmethod}

To ease the burden on readers who are not familiar with nonstandard analysis, we provide a sketch of the proof of our main results, Theorems \ref{finitemainexpd}, \ref{standardmain}, and \ref{standardmainqt} in this section. \cref{secnspre} provides a short introduction to nonstandard analysis, tailored for economists.

Let $\mathcal{E}=\{(X, \succ_\omega, P_\omega, e_\omega, \theta_{\omega})_{\omega\in \Omega}, (Y_j)_{j\in J}, (\Omega, \cB, \mu)\}$ be a measure-theoretic production economy satisfying the assumptions of \cref{standardmain}. 
The proof of \cref{standardmain} is broken into the following steps:
\begin{enumerate}[leftmargin=*]
    \item In \cref{finitemainexpd}, we establish the existence of equilibrium for finite weighted production economies which exhibit general global/local externalities on consumers' preferences. %and are compatible with the elimination of bads. 
    \cref{finitemainexpd} is the finite weighted version of Proposition 3.2.3 in \citet{fl03bk};
    \item Let $\NSE{\Omega}$ be the nonstandard extension  of the consumer space $\Omega$. \cref{secconshp} presents a technical result (\cref{hyptsexist}) on the existence of a desirable hyperfinite partition of $\NSE{\Omega}$. For almost all partition sets and all consumers in the same partition set, consumers' consumption sets, preferences, endowments, and shareholdings are infinitely close;
    \item In \cref{seceqesthyper}, we construct a hyperfinite production economy $\mathscr{E}$.  We first choose a hyperfinite set $\mathscr{T}_{\Omega}$ by picking one element from each partition set.  The weight of each consumer in $\mathscr{E}$ is derived from the probability measure $\mu$ on the standard consumer space $\Omega$. Consumers' consumption sets, preferences, endowments, and shareholdings in $\mathscr{E}$ preserve all the essential properties of their standard counterparts in $\mathcal{E}$. 
    By the transfer of \cref{finitemainexpd}, there exists a \emph{hyperfinite equilibrium} $(\bar{x},\bar{y},\bar{p})$ for $\mathscr{E}$;
    \item Every hyperfinite probability space can be extended to the associated Loeb space. In \cref{exstlbeqsec}, we construct a Loeb production economy $\Loeb{\mathscr{E}}$ from the hyperfinite production economy $\mathscr{E}$ by taking the Loeb space generated by the hyperfinite probability space defined on $\mathscr{T}_{\Omega}$. As is shown in \cref{estfreenonfree}, $\big(\ST(\bar{x}), \ST(\bar{y}), \ST(\bar{p})\big)$ is a \emph{Loeb equilibrium} for $\Loeb{\mathscr{E}}$ if and only if $\bar{x}$ is S-integrable and $\bar{y}$ is near-standard.\footnote{A nonstandard element is near-standard if it is infinitely close to a standard element, which is called the standard part of the nonstandard element. The standard part map $\ST$ maps near-standard elements to their standard parts. $\bar{x}$ is S-integrable if no infinitesimal group of consumers consumes a non-infinitesimal amount of any commodity. We provide rigorous definitions of these nonstandard objects in \cref{secnspre}.} The near-standardness of $\bar{y}$ follows from Theorem 2 on page 77 of \citet{de59}. \cref{assumptionalocbound} asserts integrable bounds for consumption of bads and implies strictly positive equilibrium prices for goods, and hence guarantees the S-integrability of 
    $\bar{x}$.\footnote{If $\bar{x}$ is not S-integrable, then $\left(\ST(\bar{x}), \ST(\bar{y})\right)$ involves strictly more free disposal than $(\bar{x},\bar{y})$. The associated Loeb allocation is a free-disposal equilibrium if consumers' preferences do not depend on the allocation, but need not be an equilibrium in general.}
    \item Since $\bar{x}$ is S-integrable, $\bar{x}$ is near-standard with respect to the weak topology on $\mathcal{L}^{1}(\Omega, \NNReals^{\ell})$.\footnote{A uniformly integrable set in $\mathcal{L}^{1}(\Omega, \NNReals^{\ell})$ is weakly compact but not norm compact. Hence, $\bar{x}$ may not be near-standard with respect to the norm topology.} The standard part $\ST_{w}(\bar{x})$\footnote{$\ST_{w}(\bar{x})$ is the conditional expectation of $\bar{x}$ with respect to the $\sigma$-algebra $\{\ST^{-1}(B): B\in \cB\}$. Informally, we obtain $\ST_{w}(\bar{x})$ by taking average of $\bar{x}$ over monads.} of $\bar{x}$ with respect to the weak topology is an allocation for the original measure-theoretic production economy $\mathcal{E}$. \cref{assumptionwkcts} asserts that consumers' preference maps are continuous with respect to the weak topology on $\mathcal{L}^{1}(\Omega, \NNReals^{\ell})$. The assumption that consumers' preferences are continuous, transitive, negatively transitive, irreflexive, and convex implies that consumers' quasi-demand sets are convex. Continuity of preferences with respect to the weak topology and convexity of quasi-demand sets jointly imply that $\ST_{w}(\bar{x})(\omega)$ is in the quasi-demand set for almost all $\omega\in \Omega$, and hence $\big(\ST_{w}(\bar{x})(\omega), \ST(\bar{y}), \ST(\bar{p})\big)$ is a quasi-equilibrium for $\mathcal{E}$. \cref{assumptionsurvival} implies that $\big(\ST_{w}(\bar{x})(\omega), \ST(\bar{y}), \ST(\bar{p})\big)$ is an equilibrium for $\mathcal{E}$.  
    \item We derive \cref{standardmainqt} from \cref{standardmain} by shifting the production set by the firm's pre-assigned quota. Thus, every measure-theoretic quota economy has a quota equilibrium.
\end{enumerate}

%the preference map $P_{t}$ takes value in $\mathcal{P}$ and is continuous with respect to the $\mathcal{L}^{1}$-norm topology on $\mathcal{L}^{1}(\Omega, \NNReals^{\ell})$. While these conditions are sufficient to imply the existence of (quasi)-equilibrium in the derived Loeb production economy, we need to assume slightly stronger regularity conditions on preference maps to establish the existence of (quasi)-equilibrium for $\mathcal{E}$. 
%In particular, these conditions include:
%\begin{itemize}
    %\item A continuity condition that allows for the construction of a standard allocation from a Loeb quasi-equilibrium allocation;
    %\item A regularity condition that ensures the convexity of the quasi-demand set, which further implies that the standard allocation constructed from a Loeb quasi-equilibrium allocation is a quasi-equilibrium for $\mathcal{E}$.
%\end{itemize}
%We first state the continuity condition which allows for the construction of a standard allocation from a Loeb allocation:

\section{Concluding Remarks}\label{secconclude}
In this paper, under natural assumptions, we establish the existence of equilibrium for measure-theoretic production economies with bads and externalities. Our main result, \cref{standardmain}, addresses the open problem raised in \citet{ha05}, and is the first equilibrium existence theorem for measure-theoretic GE models with bads and externalities. \cref{standardmain} assumes consumers' preferences are weakly continuous on $\mathcal{L}^{1}(\Omega,\NNReals^{\ell})$.\footnote{When consumers' preferences are only continuous with respect to the $\mathcal{L}^{1}$ norm topology on $\mathcal{L}^{1}(\Omega, \NNReals^{\ell})$, we show, in \cref{alocboundlemma}, the existence of an equilibrium in the Loeb production economy $\Loeb{\mathscr{E}}$, but equilibrium need not exist in the original measure-theoretic production economy $\mathcal{E}$.}
Our proof relies on a novel application of nonstandard analysis. 

In the Supplementary Material, we formulate the notion of measure-theoretic quota economy by incorporating the quota regulatory scheme, developed in \citet{ad25}, into measure-theoretic production economies. We establish, in \cref{standardmainqt}, the existence of quota equilibrium for all feasible quotas. 

%We leave open the following question: suppose we start with a measure-theoretic production economy $\mathcal{F}$ where the consumer space is a hyperfinite Loeb space, and consumers' preferences are continuous with respect to the $\mathcal{L}^{1}$ norm topology. Does $\mathcal{F}$ has an equilibrium? 

\appendix

%The theorem above follows from adjusting the results of Debreu (1959, Theorem 2, p.77) and Florenzano (2003, Proposition 2.2.1, p.52) for weighted economies.  

\section{Equilibrium Existence for Weighted Production Economies}\label{secpfweight}

In this section, we provide a rigorous proof of \cref{finitemainexpd}, hence establishing the existence of equilibrium for weighted production economies.
We first recall the definitions of attainable production plans and attainable consumption sets.
\begin{definition}\label{def_attainable}
The set $\hat Y_j$ of attainable production plans for the $j$-th producer is the projection of the set $\cO$ of the attainable consumption-production pairs to $Y_j$:
\[
\hat Y_j=\left\{y_j\in Y_j : \exists (x, y')\in \prod_{\omega\in \Omega}\NNReals^{\ell}\times \prod_{i\neq j}Y_{i}, \sum_{\omega\in \Omega} x_{\omega}\mu(\{\omega\})-\sum_{\omega\in \Omega} e(\omega) \mu(\{\omega\})- y_j-\sum_{i\neq j}{y'}(i)=0\right\}.\nonumber
\] 
The set $\hat{X}_{i}$ of attainable consumption for the $i$-th consumer is the projection of the set $\cO$ to $X_i$. In particular, $\hat{X}_{i}$ is given by:
\[
\left\{x_i\in X_i : \exists (x', y)\in \prod_{\omega\neq i}\NNReals^{\ell}\times \prod_{j\in J}Y_{j}, x_{i}\mu(\{i\})+\sum_{\omega\neq i} x'_{\omega}\mu(\{\omega\})-\sum_{\omega\in \Omega} e(\omega) \mu(\{\omega\})-\sum_{j\in J}y(i)=0\right\}.\nonumber
\]
\end{definition}
\cref{finitemainexpd} is closely related to Proposition 3.2.3 in \citet{fl03bk}.
The proof of \cref{finitemainexpd} is broken into the following three steps: 
\begin{enumerate}[leftmargin=*]
    \item We first consider the unweighted production economy;
    \item We then consider weighted production economies with positive weights;
    \item We finally prove \cref{finitemainexpd} for general weighted production economy. 
\end{enumerate}

\begin{proof}[Proof of the Unweighted Case]
We first consider the unweighted production economy $\mathcal{F}=\{(X, \succ_{\omega}, P_\omega, e_\omega, \theta_\omega)_{\omega\in \Omega}, (Y_j)_{j\in J}\}$.
In this case, \cref{finitemainexpd} is similar to Proposition 3.2.3 in \citet{fl03bk}. 
For every $\omega\in \Omega$, let $P'_{\omega}: \prod_{i\in \Omega}\NNReals^{\ell}\times Y\times \Delta\cto X_\omega$ be 
\[
P'_{\omega}(x,y,p)=\{a\in X_\omega|(a, x_\omega)\in P_{\omega}(x,y,p)\}.\nonumber
\]
Note that $P'_{\omega}$ is lower hemicontinuous since $P_{\omega}$ is continuous. As $P_{\omega}$ takes value in $\mathcal{P}_{H}$, we have $x_{\omega}\not\in \conv(P'_{\omega}(x,y,p))$ for all 
$(x,y,p)\in \prod_{i\in \Omega}\NNReals^{\ell}\times Y\times \Delta$ and all $\omega\in \Omega$. By \cref{satiate}, we have $\bigcap_{p\in \Delta}P'_{\omega}(x,y,p)\neq \emptyset$ for all $(x,y)\in \cO$ with $x_\omega\in X_\omega$. By the second bullet of \cref{cdexplain} in \cref{assumptionsurvival}, $e(\omega)\in X_{\omega}-\sum_{j\in J}\theta_{\omega j}Y_j$ for all $\omega\in \Omega$. 

\begin{claim}\label{thm_attainable}
$\hat{X}_{\omega}$ is compact for all $\omega\in \Omega$ and $\hat{Y}_{j}$ is relatively compact for all $j\in J$.
\end{claim}

\begin{proof}
For any set $B\subset \Reals^{\ell}$, let $\rc(B)$ denote the recession cone of $B$. Note that $\bar{X}=\sum_{\omega\in \Omega}X_\omega$ is a convex subset of $\NNReals^{\ell}$,  hence $\rc(\bar{X})\subset \NNReals^{\ell}$. Thus, we have $\rc(\bar{X})\cap (-\rc(\bar{X}))=\{0\}$. As $\bar{Y}\cap \NNReals^{\ell}=\{0\}$, we have $\rc(\bar{X})\cap \rc(\bar{Y})=\{0\}$. By Proposition 2.2.4 in \citet{fl03bk}, $\hat{X}_{\omega}$ is compact for every $\omega\in \Omega$. 
Note that $\bar{Y}\cap (-\bar{Y})=\{0\}$ implies that $\rc(\bar{Y})\cap (-\rc(\bar{Y}))=\{0\}$. By Proposition 2.2.4 in \citet{fl03bk} again, $\hat{Y}_{j}$ is relatively compact for every $j\in J$. 
\end{proof}
By Proposition 3.2.3 in \citet{fl03bk}, we conclude that $\mathcal{F}$ has a quasi-equilibrium $(\bar{x}, \bar{y}, \bar{p})\in \cA\times Y\times \Delta$. 
\end{proof}

We consider weighted production economies such that each consumer's weight is positive:
\begin{proof}[Positive weighted production economy:]
Let $\mu_{\omega}=\mu(\{\omega\})$ for $\omega\in \Omega$. 
Note that $\mu_{\omega}>0$. We consider the unweighted production economy $\mathcal{E}'=\{(X', \succ'_{\omega}, P'_\omega, e'_\omega, \theta'_\omega)_{\omega\in \Omega}, (Y_j)_{j\in J}\}$: 
\begin{itemize}[leftmargin=*]
    \item $\Omega$ is a finite set of consumers, and $J$ is a finite set of producers;
    \item for $\omega\in \Omega$, $X'_\omega=\mu_{\omega}X_{\omega}$ is the consumption set. Let $\mathcal{A}'=\prod_{\omega\in \Omega}X'_\omega$;
    \item $Y_j$ is the production set for producer $j$;
    \item We only provide a rigorous definition of the induced preference map $P'_{\omega}$.\footnote{The consumer's global preference relation $\succ'_{\omega}$ is defined similarly. To establish the existence of an equilibrium, it is sufficient to work with the preference map $P'_{\omega}$.} $P'_{\omega}: \prod_{i\in \Omega}\NNReals^{\ell}\times Y\times \Delta\to \mathcal{P}$ is the preference map for consumer $\omega$ such that $P'_{\omega}(x',y,p)=\mu_{\omega}P_{\omega}(x,y,p)$ where $x_{i}=\frac{x'_i}{\mu_{i}}$ for all $i\in \Omega$. Then $P'_{\omega}$ is a continuous function from $\prod_{i\in \Omega}\NNReals^{\ell}\times Y\times \Delta$ to $\mathcal{P}_{H}$;
    \item $\theta'_{\omega}=\mu_{\omega}\theta_{\omega}$ is the share for consumer $\omega$. It is clear that $\theta'_{\omega}\in \NNReals^{|J|}$ and $\sum_{k\in \Omega}\theta'_{kj}=\sum_{k\in \Omega}\mu_{k}\theta_{kj}=1$ for all $j\in J$; 
    \item $e'_{\omega}=\mu_{\omega}e_{\omega}$ is the initial endowment of consumer $\omega$. In addition, we have
    \[
    e'_{\omega}=\mu_{\omega}e_{\omega}\in \mu_{\omega}X_{\omega}-\sum_{j\in J}\mu_{\omega}\theta_{\omega j}Y_j=X'_\omega-\sum_{j\in J}\theta'_{\omega j}Y_j.\nonumber
    \]
\end{itemize}
Clearly, $\bar{Y}$ is closed, convex, and $\bar Y \cap (-\bar Y)=\{0\}=\bar Y \cap \NNReals^{\ell}$. Let
\[
\cO'=\left\{(x',y')\in \prod_{\omega\in \Omega}\NNReals^{\ell}\times Y: \sum_{\omega\in \Omega} x'_{\omega}-\sum_{\omega\in \Omega}e'(\omega)-\sum_{j\in J}{y'}(j)=0\right\}.\nonumber
\]
Note that $P'_{\omega}$ takes value in $\mathcal{P}_{H}$ for all $\omega\in \Omega$. 

\begin{claim}\label{inzweightclaim}
For each $(x',y')\in \cO'$ with $x'_{\omega}\in X'_\omega$, there exists $u\in X'_{\omega}$ such that $(u, x'_\omega)\in \bigcap_{p\in \Delta}P'_{\omega}(x',y',p)$. 
\end{claim}
\begin{proof}
Pick $(x',y')\in \cO'$ with $x'_{\omega}\in X'_{\omega}$. 
Let $x_{\omega}=\frac{1}{\mu_{\omega}}x'_{\omega}$. 
Then, we have $(x,y')\in \cO$ with $x_{\omega}\in X_\omega$. 
There exists $v\in X_\omega$ such that $(v, x_\omega)\in \bigcap_{p\in \Delta}P_{\omega}(x,y',p)$. 
Let $u=\mu_{\omega}v$. 
Then $u\in X'_{\omega}$ and $(u, x'_\omega)\in \bigcap_{p\in \Delta}P'_{\omega}(x',y',p)$.
\end{proof}
Hence, there is a quasi-equilibrium $(\bar{x}',\bar{y},\bar{p})$ for the unweighted production economy $\mathcal{E}'$. Let $\bar{x}\in X$ be such that $\bar{x}_{\omega}=\frac{\bar{x}'_{\omega}}{\mu_{\omega}}$. Clearly, we have $(\bar{x},\bar{y},\bar{p})\in \mathcal{A}\times Y\times \Delta$, where $\mathcal{A}=\prod_{\omega\in \Omega}X_{\omega}$. 
\begin{claim}\label{indemandclaim}
$\bar{x}_{\omega}\in \bar{D}_{\omega}(\bar{x},\bar{y}, \bar{p})$ for all $\omega\in \Omega$
\end{claim}
\begin{proof}
Clearly, we have $\bar{x}_{\omega}\in X_\omega$ and
\[
\bar{p}\cdot \bar{x}_{\omega}&=\bar{p}\cdot \frac{\bar{x}'_{\omega}}{\mu_{\omega}}
\leq \bar{p}\cdot \frac{e'(\omega)}{\mu_{\omega}}+\sum_{j\in J}\frac{\theta'_{\omega j}}{\mu_{\omega}}\bar{p}\cdot \bar{y}(j)
=\bar{p}\cdot e(\omega)+\sum_{j\in J}\theta_{\omega j}\bar{p}\cdot \bar{y}(j).\nonumber
\]
Hence, we conclude that $\bar{x}_{\omega}\in B_{\omega}(\bar{y}, \bar{p})$. 
Suppose $(w,\bar{x}_{\omega})\in P_{\omega}(\bar{x}, \bar{y}, \bar{p})$. 
Let $w'=\mu_{\omega}w$. Then, we have $(w',\bar{x}'_{\omega})\in P'_{\omega}(\bar{x}', \bar{y}, \bar{p})$. 
Hence, we have
\[
\bar{p}\cdot w&=\bar{p}\cdot \frac{w'}{\mu_{\omega}}
\geq\bar{p}\cdot \frac{e'(\omega)}{\mu_{\omega}}+\sum_{j\in J}\frac{\theta'_{\omega j}}{\mu_{\omega}}\bar{p}\cdot \bar{y}(j)
=\bar{p}\cdot e(\omega)+\sum_{j\in J}\theta_{\omega j}\bar{p}\cdot \bar{y}(j),\nonumber
\]
completing the proof. 
\end{proof}

We now show that $(\bar{x},\bar{y},\bar{p})$ is a quasi-equilibrium for $\mathcal{E}$:
\begin{itemize}[leftmargin=*]
    \item By \cref{indemandclaim}, $\bar{x}(\omega)\in \bar{D}_{\omega}(\bar{x},\bar{y}, \bar{p})$ for all $\omega\in \Omega$;
    \item $\bar{y}(j)\in S_{j}(\bar{p})$  for all $j\in J$;
    \item $\sum_{\omega\in \Omega} \bar{x}(\omega) \mu(\{\omega\})-\sum_{\omega\in \Omega}e(\omega)\mu(\{\omega\})-\sum_{j\in J} \bar{y}(j)=0$.
\end{itemize}
Thus, $(\bar{x},\bar{y},\bar{p})$ is a $\mathcal{Z}$-disposal quasi-equilibrium for $\mathcal{E}$. 
\end{proof}

We now prove the general weighted case, hence proving \cref{finitemainexpd}. 

\begin{proof}[Proof of \cref{finitemainexpd}]
Let $\Omega'=\{\omega\in \Omega: \mu(\{\omega\})>0\}$. 
For every $\omega\in \Omega\setminus \Omega'$, pick $\epsilon_{\omega}\in X_\omega$. 
For $a\in \prod_{\omega\in \Omega'}X_\omega=\mathcal{A}'$, let $E(a)\in \mathcal{A}=\prod_{\omega\in \Omega}X_\omega$ be:
\begin{equation}
  E(a)_{\omega} =
    \begin{cases}
      a_{\omega} & \text{for all $\omega\in \Omega'$}\\
      \epsilon_{\omega} & \text{for all $\omega\not\in \Omega'$}\\  
    \end{cases}\nonumber    
\end{equation}
Consider the weighted production economy $\mathcal{E}'=\{(X, \succ^{\Omega'}_{\omega}, P^{\Omega'}_\omega, e_\omega, \theta_\omega)_{\omega\in \Omega'}, (Y_j)_{j\in J}, \mu\}$ where $P^{\Omega'}_{\omega}(x,y,p)=P_{\omega}(E(x),y,p)$. It is easy to verify that $\mathcal{E}'$ satisfies all the conditions of \cref{finitemainexpd} and every consumer in $\mathcal{E}'$ has positive weight. Hence, there is a quasi-equilibrium 
$(\bar{x},\bar{y}, \bar{p})\in \mathcal{A}'\times Y\times \Delta$ for $\mathcal{E}'$. 
Then, $(E(\bar{x}), \bar{y}, \bar{p})\in \mathcal{A}\times Y\times \Delta$ is a   quasi-equilibrium for $\mathcal{E}$. By \cref{interioritylm},  $(E(\bar{x}), \bar{y}, \bar{p})$ is an equilibrium.
\end{proof}

\section{Equilibrium Existence for Measure-theoretic Production Economies}\label{secexisteq}
The primary goal of this section is to give a rigorous proof to our main result, \cref{standardmain}. To do this, we fix a measure-theoretic production economy 
\[\label{fixstandard}
\mathcal{E}=\{(X, \succ_{\omega}, P_{\omega}, e_{\omega}, \theta)_{\omega\in \Omega}, (Y_j)_{j\in J}, (\Omega, \BorelSets \Omega, \mu)\}
\]
as in \cref{defmeasureeco}, where the consumer space $\Omega$ is equipped with the Borel $\sigma$-algebra $\BorelSets {\Omega}$ and a probability measure $\mu$.
The proof of \cref{standardmain} makes use of nonstandard analysis and is broken into the following steps:
\begin{enumerate}[leftmargin=*]
    \item Construct a suitable hyperfinite partition $\mathcal{T}_{\Omega}$ of $\NSE{\Omega}$. We then construct an associated hyperfinite production economy $\mathscr{E}$ on $\mathcal{T}_{\Omega}$;
    \item The existence of equilibrium for $\mathscr{E}$ follows from transferring \cref{finitemainexpd};
    \item We further construct a Loeb production economy $\Loeb{\mathscr{E}}$ from $\mathscr{E}$. We then prove the existence of equilibrium in $\Loeb{\mathscr{E}}$ under moderate regularity assumptions;
    \item Under the assumptions of \cref{standardmain}, we construct a standard allocation from an equilibrium allocation of $\Loeb{\mathscr{E}}$ and show that the standard allocation is an equilibrium of the original measure-theoretic production economy $\mathcal{E}$. 
\end{enumerate}
Since all four steps above involve nonstandard analysis, we first give a gentle introduction to nonstandard analysis. 

\subsection{Notation from Nonstandard Analysis}\label{secnspre}
In this section, we give a gentle introduction to nonstandard analysis. 
%For those who are not familiar with nonstandard analysis, \citet{adku2, adku1} provide reviews tailored to economists. \citet{NDV, NSAA97, NAW} provide thorough introductions. 
We use $\NSE{}$ to denote the nonstandard extension map taking elements, sets, functions, relations, etc., to their nonstandard counterparts.
In particular, $\HReals$ and $\NSE{\Nats}$ denote the nonstandard extensions of the reals and natural numbers, respectively.
An element $r\in \HReals$ is \emph{infinite} if $|r|>n$ for every $n\in \Nats$ and is \emph{finite} otherwise. An element $r \in \HReals$ with $r > 0$ is \textit{infinitesimal} if $r^{-1}$ is infinite. For $r,s \in \HReals$, we use the notation $r \approx s$ as shorthand for the statement ``$|r-s|$ is infinitesimal,'' and use $r \gtrapprox s$ as shorthand for the statement ``either $r \geq s$ or $r \approx s$.''

Given a topological space $(X,\topology)$,
the monad of a point $x\in X$ is the set $\bigcap_{ U\in \topology \, : \, x \in U}\NSE{U}$.
An element $x\in \NSE{X}$ is \emph{near-standard} if it is in the monad of some $y\in X$.
We say $y$ is the standard part of $x$ and write $y=\ST(x)$. 
Note that such $y$ is unique provided that $X$ is a Hausdorff space.
The \emph{near-standard part} $\NS{\NSE{X}}$ of $\NSE{X}$ is the collection of all near-standard elements of $\NSE{X}$.
The standard part map $\ST$ is a function from $\NS{\NSE{X}}$ to $X$, taking near-standard elements to their standard parts.
In both cases, the notation elides the underlying space $Y$ and the topology $\topology$,
because the space and topology will always be clear from context.
For a metric space $(X,d)$, two elements $x,y\in \NSE{X}$ are \emph{infinitely close} if $\NSE{d}(x,y)\approx 0$.
An element $x\in \NSE{X}$ is near-standard if and only if it is infinitely close to some $y\in X$.
An element $x\in \NSE{X}$ is finite if there exists $y\in X$ such that $\NSE{d}(x,y)<\infty$ and is infinite otherwise.

Let $X$ be a topological space endowed with Borel $\sigma$-algebra $\BorelSets X$ and
let $\FM{X}$ denote the collection of all finitely additive probability measures on $(X,\BorelSets X)$.
An internal probability measure $\mu$ on $(\NSE{X},\NSE{\BorelSets X})$ is an element of $\NSE{\FM{X}}$.
The Loeb space of the internal probability space $(\NSE{X},\NSE{\BorelSets X}, \mu)$ is a countably additive probability space $(\NSE{X},\Loeb{\NSE{\BorelSets X}}, \Loeb{\mu})$ such that
$
\Loeb{\NSE{\BorelSets X}}=\{A\subset \NSE{X}|(\forall \epsilon>0)(\exists A_i,A_o\in \NSE{\BorelSets X})(A_i\subset A\subset A_o\wedge \mu(A_o\setminus A_i)<\epsilon)\}
$ and
$
\Loeb{\mu}(A)=\sup\{\ST(\mu(A_i))|A_i\subset A,A_i\in \NSE{\BorelSets X}\}=\inf\{\ST(\mu(A_o))|A_o\supset A,A_o\in \NSE{\BorelSets X}\}.
$

Every standard model is connected to its nonstandard extension via the \emph{transfer principle}, which asserts that a first order statement is true in the standard model if and only if it is true in the nonstandard model.
Given a cardinal number $\kappa$, a nonstandard model is $\kappa$-saturated if the following condition holds:
Let $\cF$ be a family of internal sets with cardinality less than $\kappa$. If $\cF$ has the finite intersection property, then the total intersection of $\cF$ is non-empty. In this paper, we assume our nonstandard model is as saturated as we need.\footnote{see \textit{e.g.} {\citet[][Thm.~1.7.3]{NSAA97}} for the existence of $\kappa$-saturated nonstandard models for any uncountable cardinal $\kappa$.}

\subsubsection{Loeb Probability Space}
In this section, we provide a brief introduction of Loeb spaces introduced by \citet{Loeb75}. 
We focus on hyperfinite probability spaces and their corresponding Loeb spaces. 

A hyperfinite set $S$ is equipped with the internal algebra $I(S)$, consisting of all internal subsets of $S$. 
Let $P$ be an internal probability measure on $S$. We use $(S, \Loeb{I(S)}, \Loeb{P})$ to denote the Loeb probability space generated from $(S, I(S), P)$.

\begin{definition}\label{deflifting}
Let $(S, I(S), P)$ be a hyperfinite probability space, and $(S, \Loeb{I(S)}, \Loeb{P})$ be the Loeb space. 
Let $X$ be a Hausdorff topological space, and $f$ be a Loeb measurable function from $S$ to $X$. 
An internal function $F: S\to \NSE{X}$ is a lifting of $f$ provided that $f(s)=\ST(F(s))$ for $\Loeb{P}$-almost all $s\in S$. 
\end{definition}

\begin{lemma}[{\citep[][Section.~4, Corollary.~5.1]{NSAA97}}]
Every Loeb measurable function into a second countable topological space has a lifting. 
\end{lemma}

We now introduce the S-integrability notion, which guarantees that the Loeb integral of a Loeb integrable function almost agrees with the internal integral of its lifting. 

\begin{definition}\label{defsint}
Let $(S, I(S), P)$ be a hyperfinite probability space, and $(S, \Loeb{I(S)}, \Loeb{P})$ be the corresponding Loeb space. 
Let $F: S\to \NSE{\Reals}$ be an internally integrable function such that $\ST(F)$ exists $\Loeb{P}$-almost surely. 
$F$ is said to be S-integrable if $\ST(F)$ is $\Loeb{P}$-integrable and $\int |F|(s) P(\dee s)\approx \int \ST(|F|)(s) \Loeb{P}(\dee s)$. 
\end{definition}

\begin{theorem}[{\citep[][Section.~4, Theorem.~6.2]{NSAA97}}]\label{sintmain}
Let $(S, I(S), P)$ be a hyperfinite probability space, and $(S, \Loeb{I(S)}, \Loeb{P})$ be the Loeb space. 
Let $F: S\to \NSE{\Reals}$ be an internally integrable function such that $\ST(F)$ exists $\Loeb{P}$-almost surely. 
The following are equivalent: 
\begin{enumerate}[(i), leftmargin=*]
    \item $F$ is S-integrable;
    \item $\ST(\int |F(s)| P(\dee s))$ exists and equals to $\lim_{n\to \infty} \ST(\int |F_{n}(s)| P(\dee s))$ (where for $n\in \Nats$, $F_n=\min\{F,n\}$ when $F\geq 0$ and $F_n=\max\{F,-n\}$ when $F<0$);
    \item For every infinite $K>0$, $\int_{|F|>K}|F(s)|P(\dee s)\approx 0$;
    \item $\ST(\int |F(s)| P(\dee s))$ exists, and $\int_{B}|F(s)|P(\dee s)\approx 0$ for all $B$ with $P(B)\approx 0$.
\end{enumerate}
\end{theorem}

We conclude this section with the following theorem which guarantees the existence of an S-integrable lifting for every real-valued Loeb integrable function. 

\begin{theorem}[{\citep[][Section.~4, Theorem.~6.4]{NSAA97}}]\label{extsintlft}
Let $(S, I(S), P)$ be a hyperfinite probability space, and $(S, \Loeb{I(S)}, \Loeb{P})$ be the Loeb space. 
Let $f: S\to \Reals$ be Loeb measurable. Then $f$ is integrable if and only if it has an S-integrable lifting. 
\end{theorem}
We now follow the four steps laid out before \cref{secnspre} to prove our main result, \cref{standardmain}. 

\subsection{Construction of Hyperfinite Partition}\label{secconshp}

This section is devoted to establishing the existence of a desired hyperfinite partition of the nonstandard extension $\NSE{\Omega}$ of the consumer space $\Omega$. 
In particular, we wish to construct a hyperfinite partition $\mathcal{T}_{\Omega}$ of $\NSE{\Omega}$ such that consumers within the same element of the partition have similar consumption sets, preferences, endowments and shareholdings of firms. 

Recall that, in \cref{assumptionmseco}, we assume the consumer space $\Omega$ is a Polish space endowed with Borel $\sigma$-algebra $\BorelSets \Omega$ and $\mu$ is a probability measure on $\Omega$. The concept of Lusin measurable function is of cruicial importance.

\begin{definition}\label{deflusin}
Let $(X, \BorelSets X, \mu)$ be a Radon probability space and $Y$ be a topological space endowed with the Borel $\sigma$-algebras. 
A function $f: X\to Y$ is Lusin measurable if, for every $\epsilon>0$. there is a compact set $K_{\epsilon}\subset X$ such that $f$ is continuous on $K_{\epsilon}$.
\end{definition}

% An immediate consequence of \cref{deflusin} is:

% \begin{lemma}[{\citep[][Chapter.~4, Corollary.~5.3]{NSAA97}}]\label{lusinscts}
% Let $(X, \BorelSets X, \mu)$ be a Radon probability space, $Y$ be a Hausdorff space endowed with the Borel $\sigma$-algebra, and 
% $f: X\to Y$ be strongly Lusin measurable. 
% Then, there is a set $Y\subset \NSE{X}$ of full Loeb measure such that for all $y_1, y_2\in Y$
% \[
% y_1\approx y_2\implies \NSE{f}(y_1)\approx \NSE{f}(y_2). 
% \]
% \end{lemma}

For second countable range space, measurability is equivalent to Lusin measurability. In particular, we have the following result from the nonstandard measure theory: 
\begin{theorem}[{\citep[][heorem.~5.3, page.~167]{NDV}}]\label{measurescts}
Let $(X, \BorelSets X, \mu)$ be a Radon probability space, $Y$ be a second countable Hausdorff space endowed with the Borel $\sigma$-algebra, and 
$f: X\to Y$ be measurable. 
Then, there is a set $Z\subset \NS{\NSE{X}}$ of full Loeb measure such that $\NSE{f}(z)\approx f(\ST(z))$ for all $z\in Z$. 
Consequently, for all $z_1, z_2\in Z$, we have $z_1\approx z_2\implies \NSE{f}(z_1)\approx \NSE{f}(z_2)$.
\end{theorem}

% We now state our assumption on the consumer space $T$ of $\mathcal{E}$. 

% \begin{assumption}\label{assumptionmseco}
% The consumer space $T$ of the measure theoretic production economy (as in \cref{defmeasureeco})
% \[
% \mathcal E=\{(X, R_t, P_t, e_t, \theta)_{t\in T}, (Y_j)_{j\in J}, (T, \cB, \mu), \mathcal{Z}\}
% \]
% is a Polish space endowed with Borel $\sigma$-algebra $\BorelSets T$ and $\mu$ is an atomless probability measure on $T$. As a result, $(T, \BorelSets T, \mu)$ is an atomless Radon probability space. 
% \end{assumption}

Recall from \cref{defmeasureeco} that the preference map $P_{\omega}: \mathcal{L}^{1}(\Omega, \NNReals^{\ell})\times Y\times \Delta\to \mathcal{P}$ is continuous in the norm topology on $\mathcal{L}^{1}(T, \NNReals^{l})\times Y\times\Delta$ for every $\omega\in \Omega$ and measurable in $\Omega$.
Let $\cC[\mathcal{L}^{1}(\Omega,\NNReals^{\ell})\times Y\times \Delta, \mathcal{P}]$ denote the collection of all continuous functions from $\mathcal{L}^{1}(\Omega,\NNReals^{\ell})\times Y\times \Delta$ to $\mathcal{P}$, equipped with the sup-norm topology. 
$\cC[\mathcal{L}^{1}(\Omega,\NNReals^{\ell})\times Y\times \Delta, \mathcal{P}]$ is a complete metric space, but it is not separable, and hence not a Polish space. 
Let $\chi: \Omega\to \cC[\mathcal{L}^{1}(\Omega,\NNReals^{\ell})\times Y\times \Delta, \mathcal{P}]$ be the map $\chi(\omega)=P_{\omega}$.
To deduce the tightness of the induced measure\footnote{Let $X$ be a Hausdorff space equipped with Borel $\sigma$-algebra $\BorelSets X$. A probability measure $P$ on $(X,\BorelSets X)$ is tight if, for any $\epsilon>0$, there is a compact set $K_{\epsilon}\subset X$ such that $P(K_{\epsilon})>1-\epsilon$.} $\mu_{\chi}=\mu\circ \chi^{-1}$ on $\cC[\mathcal{L}^{1}(\Omega,\NNReals^{\ell})\times Y\times \Delta, \mathcal{P}]$, we need the following lemma:
\begin{lemma}[{\citep[][page.~235]{bill1968}}]\label{billmcar}
Let $X$ be a complete metric space endowed with the Borel $\sigma$-algebra $\BorelSets X$. 
Let $P$ be a probability measure on $(X, \BorelSets X)$ such that the support of $P$ is separable. 
Then $P$ is tight. 
\end{lemma}
The support of every probability measure is separable unless the measurable cardinal exist.\footnote{A necessary and sufficient condition that each probability measure's support be separable is that each discrete subset of the sample space has non-measurable cardinality, see Theorem 2 on page 235 of \citet{bill1968}. \citet{bill1968} points out that, if measurable cardinals exist, they must be so large as never to arise in a natural way in mathematics.} 
In any case, the support of a probability measure is separable for any reasonable metric space. 
By \cref{billmcar}, $\mu_{\chi}$ is tight on $\cC[\mathcal{L}^{1}(\Omega,\NNReals^{\ell})\times Y\times \Delta, \mathcal{P}]$.

\begin{theorem}\label{hyptsexist}
Let $\mathcal{E}$ be the measure-theoretic production economy which we fix in \cref{fixstandard}.
Suppose that $\mathcal{E}$ satisfies \cref{assumptionmseco}. 
Then there exists a hyperfinite partition $\mathcal{T}=\{B_i\in \NSE{\BorelSets \Omega}: i\leq K\}$ of $\NSE{\Omega}$ with $\mathcal{T}'\subset \cT_{\Omega}$ such that, for $y\in \bigcup \cT'$:
\begin{enumerate}[(i), leftmargin=*]
    \item $\bigcup \cT'$ is $\Loeb{\NSE{\mu}}$-measurable and $\Loeb{\NSE{\mu}}(\bigcup \cT')=1$;
    \item $\bigcup \cT'\subset \NS{\NSE{\Omega}}$ and the diameter of each element of $\cT'$ is infinitesimal;
    \item $\NSE{e}(y)\approx e(\ST(y))$, $\NSE{\chi}(y)\approx \chi(\ST(y))$, $\NSE{\theta}(y)\approx \theta(\ST(y))$ and $\NSE{X}(y)\approx X(\ST(y))$.
\end{enumerate}
\end{theorem}
\begin{proof}
By \cref{measurescts}, there exists a $Y_{1}\subset \NSE{\Omega}$ with $\Loeb{\NSE{\mu}}(Y_{1})=1$ such that for all $y\in Y_{1}$
\begin{enumerate}[leftmargin=*]
    \item $\NSE{e}(y)\approx e(\ST(y))$;
    \item $\NSE{\theta}(y)\approx \theta(\ST(y))$;
    \item $\NSE{X}(y)\approx X(\ST(y))$.
\end{enumerate}
For every $\epsilon>0$, there exists a compact set $C_{\epsilon}\subset \cC[\mathcal{L}^{1}(\Omega,\NNReals^{\ell})\times Y\times \Delta, \mathcal{P}]$ such that $\mu\big(\chi^{-1}(C_{\epsilon})\big)>1-\epsilon$. 
As every compact metric space is second countable, by \cref{measurescts}, there exists $\Omega_{\epsilon}\subset \NSE{\big(\chi^{-1}(C_{\epsilon})\big)}$ such that
\begin{enumerate}[leftmargin=*]
    \item $\Loeb{\NSE{\mu}}(\Omega_{\epsilon})>1-\epsilon$;
    \item For $x\in \Omega_{\epsilon}$, $\NSE{\chi}(x)\approx \chi(\ST(x))$.
\end{enumerate}
Construct such $\Omega_{\frac{1}{n}}$ for every $n\in \Nats$ and consider $\bigcup_{n\in \Nats}\Omega_{\frac{1}{n}}$, the set is Loeb measurable with respect to $\Loeb{\NSE{\mu}}$ and we have $\Loeb{\NSE{\mu}}(\bigcup_{n\in \Nats}\Omega_{\frac{1}{n}})=1$. Moreover, for every $a\in \bigcup_{n\in \Nats}\Omega_{\frac{1}{n}}$, we have $\NSE{\chi}(a)\approx \chi(\ST(a))$. We use $Y_2$ to denote the set $\bigcup_{n\in \Nats}\Omega_{\frac{1}{n}}$. Then the set $Y=Y_1\cap Y_2$ is a $\Loeb{\NSE{\mu}}$-measurable set with $\Loeb{\NSE{\mu}}(Y)=1$ such that for all $y\in Y$:
\begin{enumerate}[leftmargin=*]
    \item $\NSE{e}(y)\approx e(\ST(y))$, $\NSE{X}(y)\approx X(\ST(y))$, $\NSE{\theta}(y)\approx \theta(\ST(y))$, and $\NSE{\chi}(y)\approx \chi(\ST(y))$.
\end{enumerate}

Let $\delta\in \NSE{\Reals}$ be a positive infinitesimal. 
Let $d$ denote the metric on $\cK(\NNReals^{\ell})$ and $d_{\mathrm{sup}}$ denote the metric generated from the sup-norm on 
$\cC[\mathcal{L}^{1}(\Omega,\NNReals^{\ell})\times Y\times \Delta, \mathcal{P}]$.
For each $n\in \Nats$, let $\phi_n(\cT_{n}, \cT'_{n})$ be the conjunction of the following formulas:
\begin{enumerate}[leftmargin=*]
    \item $\cT_{n}\subset \NSE{\BorelSets {\Omega}}$ is a hyperfinite partition of $\NSE{\Omega}$;
    \item $\cT'_{n}\subset \cT_{n}$ is internal and the diameter of every element in $\cT'_{n}$ is no greater than $\delta$;
    \item $\NSE{\mu}(\bigcup \cT'_{n})>1-\frac{1}{n}$; 
    \item For every element $V\in \cT'_{n}$, we have $|\NSE{e}(a)-\NSE{e}(b)|<\frac{1}{n}$, $|\NSE{\theta}(a)-\NSE{\theta}(b)|<\frac{1}{n}$, $\NSE{d}(\NSE{X}(a), \NSE{X}(b))<\frac{1}{n}$ and $\NSE{d}_{\mathrm{sup}}(\NSE{\chi}(a), \NSE{\chi}(b))<\frac{1}{n}$ for all $a, b\in V$.
\end{enumerate}
To show that $\{\phi_n(\cT_{n}, \cT'_{n}): n\in \Nats\}$ is finitely satisfiable, it is sufficient to show that each $\phi_n(\cT_{n}, \cT'_{n})$ is satisfiable. 
As $(\Omega, \BorelSets \Omega, \mu)$ is a Radon probability space, there exists a compact set $K_{n}\subset \Omega$ such that $\mu(K_{n})>1-\frac{1}{2n}$. 
Pick $Y_n\in \NSE{\BorelSets \Omega}$ such that $Y_n\subset Y$ and $\NSE{\mu}(Y_n)>1-\frac{1}{2n}$.
Let $K'_n=\NSE{K}_n\cap Y_n$. Then $\NSE{\mu}(K'_n)>1-\frac{1}{n}$. 
So there is a mutually disjoint hyperfinite collection $\{V_i: i\leq M\}\subset \NSE{\BorelSets \Omega}$ such that:
\begin{enumerate}[leftmargin=*]
    \item the diameter of each $V_i$ is no greater than $\delta$;
    \item $K'_n=\bigcup_{i\leq M}V_i$.
\end{enumerate}
Let $\cT_{n}=\{V_i: i\leq M\}\cup \{\NSE{\Omega}\setminus K'_n\}$ and let $\cT'_{n}=\{V_i: i\leq M\}$. 
Clearly, $\phi_{n}(\cT_{n}, \cT'_{n})$ is satisfied. 

By saturation, there exist $\cT$ and $\hat{\cT}$ such that $\phi_{n}(\cT_{\Omega}, \cT')$ holds simultaneously for all $n\in \Nats$:
\begin{enumerate}[leftmargin=*]
    \item $\cT\subset \NSE{\BorelSets {\Omega}}$ is a hyperfinite partition of $\NSE{\Omega}$;
    \item $\hat{\cT}\subset \cT$ is internal and the diameter of every element in $\hat{\cT}$ is no greater than $\delta$;
    \item $\NSE{\mu}(\bigcup \hat{\cT})\approx 1$;
    \item For $V\in \hat{\cT}$, $\NSE{e}(a)\approx \NSE{e}(b)$, $\NSE{\theta}(a)\approx \NSE{\theta}(b)$, $\NSE{X}(a)\approx \NSE{X}(b)$ and $\NSE{\chi}(a)\approx \NSE{\chi}(b)$ for $a, b\in V$. 
\end{enumerate}
Let us consider the set $\cT'=\{V\in \hat{\cT}: V\cap \NS{\NSE{\Omega}}\cap Y\neq \emptyset\}$. 
Clearly, $\cT'$ is Loeb measurable and $\Loeb{\NSE{\mu}}(\bigcup \cT')=1$.
Moreover, the diameter of every element of $\cT'$ is infinitesimal. 
Pick some element $V_0\in \cT'$. 
As $V_0\cap \NS{\NSE{\Omega}}\neq \emptyset$, we have $V_0\subset \NS{\NSE{\Omega}}$. 
As $V_0\cap Y\neq \emptyset$, then there exists an element $a_0\in V_0$ such that $\NSE{e}(a_0)\approx e(\ST(a_0))$, $\NSE{\theta}(a_0)\approx \theta(\ST(a_0))$, $\NSE{X}(a_0)\approx X(\ST(a_0))$ and 
$\NSE{\chi}(a_0)\approx \chi(\ST(a_0))$. 
For every $b\in V_0$, we have 
\begin{enumerate}[leftmargin=*]
    \item $\NSE{e}(b)\approx \NSE{e}(a_0)\approx e(\ST(a_0))=e(\ST(b))$;
    \item $\NSE{\theta}(b)\approx \NSE{\theta}(a_0)\approx \theta(\ST(a_0))=\theta(\ST(b))$;
    \item $\NSE{X}(b)\approx \NSE{X}(a_0)\approx X(\ST(a_0))=X(\ST(b))$;
    \item $\NSE{\chi}(b)\approx \NSE{\chi}(a_0)\approx \chi(\ST(a_0))=\chi(\ST(b))$.
\end{enumerate}
Thus, $\cT$ and $\cT'$ satisfy all conditions of the theorem, hence completing the proof. 
\end{proof}

In the next section, we will construct an associated hyperfinite production economy $\mathscr{E}$ of $\mathcal{E}$ via the hyperfinite partition $\cT_{\Omega}$ and  establish the existence of equilibrium of $\mathscr{E}$. The proof of which follows from applying the transfer principle to \cref{finitemainexpd}.

%Let $\mathcal{A}_{\mathscr{P}}=\{f\in \mathcal{L}^{1}(\mathscr{P}_{\Omega}, \NNReals^{l}): f(t)\in X_t \mbox{ almost surely}\}$.

%We now show that, if $\phi$ is homogeneous of degree $0$ in prices, then so is $\phi^{\mathscr{P}}$.

%\begin{lemma}\label{homopreserve}
%Let $\mathscr{P}$ be a countable partition of $\Omega$. 
%If $\phi$ is homogeneous of degree $0$ in prices, then $\phi^{\mathscr{P}}$ is homogeneous of degree $0$ in prices. 
%\end{lemma}
%\begin{proof}
%Pick $p\in \Delta'$ and $\lambda>0$ with $\lambda p\in \Delta'$. 
%For any $f\in \mathcal{L}^{1}(\mathscr{P}_{\Omega}, \NNReals^{l})$ and $y\in Y$, we have 
%\[
%\phi^{\mathscr{P}}(f,y,p)=\phi(E(f),y,p)=\phi(E(f),y,\lambda p)=\phi^{\mathscr{P}}(f, y,\lambda p).
%\]
%Hence, we have the desired result. 
%\end{proof}

\subsection{Existence of Equilibrium in Hyperfinite Production Economy}\label{seceqesthyper}
A hyperfinite production economy is a nonstandard economy but satisfies all the first-order logic properties of a weighted production economy. Hence, the existence of equilibrium for a hyperfinite weighted economy follows from \cref{finitemainexpd} via the transfer principle. 

We first establish some basic properties on standard partitions of the measure-theoretic production economy $\mathcal{E}$. 
Let $\mathscr{H}=\{H_i: i\in \Nats\}\subset \BorelSets \Omega$ be a countable partition of $\Omega$. Let $\mathscr{H}_{\Omega}=\{h_i: i\in \Nats\}$ be a countable subset of $\Omega$ such that $h_i\in H_i$ for each $i\in \Nats$. 
Define $\mu^{\mathscr{H}}$ to be the probability measure on $\mathscr{H}_{\Omega}$ such that $\mu^{\mathscr{H}}(\{h_i\})=\mu(H_i)$ for all $i\in \Nats$.
For a function $f: \mathscr{H}_{\Omega}\to \NNReals^{\ell}$, define $E(f): \Omega\to \NNReals^{\ell}$ by letting $E(f)(x)=f(h_{x})$, where $h_{x}$ is the unique point in $\mathscr{H}_{\Omega}$ such that $x$ is in the element of $\mathscr{H}$ that associates with $h_{x}$. Let $\mathcal{L}^{1}(\mathscr{H}_{\Omega}, \NNReals^{\ell})$ denote the set of integrable functions on $\mathscr{H}_{\Omega}$ with respect to $\mu^{\mathscr{H}}$. It is easy to see that, for every $f\in \mathcal{L}^{1}(\mathscr{H}_{\Omega}, \NNReals^{\ell})$, $E(f)$ is an element of $\mathcal{L}^{1}(\Omega, \NNReals^{\ell})$.

\begin{definition}
Let $\phi: \mathcal{L}^{1}(\Omega, \NNReals^{\ell})\times Y\times \Delta\to \mathcal{P}$ and $\mathscr{P}$ be a countable partition of $\Omega$.
The restriction $\phi^{\mathscr{H}}: \mathcal{L}^{1}(\mathscr{H}_{\Omega}, \NNReals^{\ell})\times Y\times \Delta\to \mathcal{P}$ of $\phi$ is  $\phi^{\mathscr{H}}(f, y,p)=\phi(E(f),y,p)$.
\end{definition}

%Let $\mathcal{A}_{\mathscr{P}}=\{f\in \mathcal{L}^{1}(\mathscr{P}_{\Omega}, \NNReals^{l}): f(t)\in X_t \mbox{ almost surely}\}$.
Recall that the set $\mathcal{P}_{H}$ of convex preferences is a closed subset of $\mathcal{P}$, hence is also compact. 

\begin{lemma}\label{restpreserve}
Suppose $\phi: \mathcal{L}^{1}(\Omega, \NNReals^{\ell})\times Y\times \Delta\to \mathcal{P}_{H}$. 
Let $\mathscr{H}$ be a countable partition of $\Omega$.
Then $\phi^{\mathscr{H}}$ also maps to $\mathcal{P}_{H}$. 
Moreover, if $\phi$ is continuous then so is $\phi^{\mathscr{H}}$.
\end{lemma}
\begin{proof}
We view $\mathcal{L}^{1}(\mathscr{H}_{\Omega}, \NNReals^{\ell})$ as a subset of $\mathcal{L}^{1}(\Omega, \NNReals^{\ell})$ by associating $f\in \mathcal{L}^{1}(\mathscr{H}_{\Omega}, \NNReals^{\ell})$ with $E(f)\in \mathcal{L}^{1}(\Omega, \NNReals^{\ell})$. Thus, $\phi^{\mathscr{H}}$ is a restriction of $\phi$, completing the proof. 
\end{proof}

Under \cref{assumptionmseco}, let $\mathscr{T}=\{T_1, T_2, \dotsc, T_{K}\}\subset \NSE{\BorelSets \Omega}$ be a hyperfinite partition of $\NSE{\Omega}$ as in \cref{hyptsexist}.
Let $\mathscr{T}_{\Omega}=\{t_i: i\leq K\}\subset \NSE{\Omega}$ be a hyperfinite set such that: 
\begin{enumerate}[leftmargin=*]
    \item $\Omega\subset \mathscr{T}_{\Omega}$ and $t_i\in T_i$ for every $i\leq K$;
    \item If $T_i\cap \NSE{\Omega_0}\neq \emptyset$, then $t_i\in \NSE{\Omega_0}$.
\end{enumerate}
Our hyperfinite production economy 
\[\label{fixhyper}
\mathscr{E}=\{(\NSE{X}, \NSE{\succ}^{\mathscr{T}}_t, \NSE{P}^{\mathscr{T}}_t, \hat{e}_t, \hat{\theta})_{t\in \mathscr{T}_{\Omega}}, (\NSE{Y}_j)_{j\in J}, \NSE{\mu}^{\mathscr{T}}\}
\]
is defined to be:
\begin{enumerate}[(i), leftmargin=*]
\item $\mathscr{T}_{\Omega}$ is the hyperfinite consumer space and $\NSE{\mu}^{\mathscr{T}}(\{t_i\})=\NSE{\mu}(T_i)$; 

\item $J$ is the same finite set of firms;

\item For every $t\in \mathscr{T}_{\Omega}$, $\NSE{X}(t): \mathscr{T}_{\Omega}\to \NSE{\cK}(\NSE{\NNReals^{\ell}})$ is the $\NSE{}$consumption set of consumer $t$. By the transfer principle, $\NSE{X}(t)\neq\emptyset$ for all $t\in \mathscr{T}_{\Omega}$. We sometimes write $\NSE{X}_t$ for $\NSE{X}(t)$;

\item $\NSE{Y}_j\subset \NSE{\Reals^{\ell}}$ is the nonstandard extension of $Y_j$, denoting the $\NSE{}$production set of producer $j\in J$. 
Note that $\NSE{Y}=\prod_{j\in J}\NSE{Y}_j$;

\item the set of $\NSE{}$allocations is $\mathscr{A}=\{x\in \NSE{\mathcal{L}^{1}}(\mathscr{T}_{\Omega}, \NSE{\NNReals^{\ell}}): x(t)\in \NSE{X}_t\ \NSE{\mu}^{\mathscr{T}}\mbox{-almost surely}\}$, which is equipped with the $\NSE{\mathcal{L}^{1}}$ strong topology;
 
\item \label{item-hyperprefer} Let $\NSE{M}^{\mathscr{T}}_{t}=\NSE{\mathcal{L}^{1}}(\mathscr{T}_{\Omega}, \NSE{\NNReals^{\ell}})\times \NSE{Y}\times \NSE{\Delta}\times \NSE{X}_{t}$ and $\NSE{\succ}^{\mathscr{T}}_{t}= (\NSE{M}^{\mathscr{T}}_{t}\times \NSE{M}^{\mathscr{T}}_{t})\cap \NSE{\succ}_{t}$ for $t\in \mathscr{T}_{\Omega}$. Let $\NSE{P}^{\mathscr{T}}_{t}: \NSE{\mathcal{L}^{1}}(\mathscr{T}_{\Omega}, \NSE{\NNReals^{\ell}})\times \NSE{Y}\times \NSE{\Delta}\to \NSE{\mathcal{P}}$ be the preference map induced from $\NSE{\succ}^{\mathscr{T}}_{t}$. $\NSE{P}^{\mathscr{T}}_{t}$
is the restriction of $\NSE{P}_{t}$ to $\mathscr{T}_{\Omega}$ for each $t\in \mathscr{T}_{\Omega}$\footnote{That is, $\NSE{P}^{\mathscr{T}}_{t}$ is an internal mapping such that $\NSE{P}^{\mathscr{T}}_{t}(x,y,p)=\NSE{P}_{t}(\NSE{E}(x),y,p)$, where $E(x)$ is the extension of $x$ defined at the beginning of \cref{secexisteq}.}; 

\item As $\mathscr{T}$ satisfies \cref{hyptsexist}, for all $j\in J$, we have
\[
\sum_{i\leq K}\NSE{\theta}(t_i)(j)\NSE{\mu}^{\mathscr{T}}(\{t_i\})\approx \int_{\Omega}\theta(\omega)(j)\mu(\dee \omega)=1.\nonumber
\]
Let $\alpha_{j}=\sum_{i\leq K}\NSE{\theta}(t_i)(j)\NSE{\mu}^{\mathscr{T}}(\{t_i\})$ for all $j\in J$. For each $t\in \mathscr{T}_{\Omega}$ and $j\in J$, define $\hat{\theta}(t)(j)=\frac{1}{\alpha_{j}}\NSE{\theta}(t)(j)$, which is the consumer $t$'s shareholding on firm $j$. Note that $\hat{\theta}(t)\approx \NSE{\theta}(t)$ for all $t\in \mathscr{T}_{\Omega}$. We sometimes write $\hat{\theta}_{tj}$ for $\hat{\theta}(t)(j)$;

\item For each $t\in \mathscr{T}_{\Omega}$, $\hat{e}(t)\approx \NSE{e}(t)$ is to be determined later in this section, and it represents the initial endowment of consumer $t$. 
\end{enumerate}

For every $t\in \mathscr{T}_{\Omega}$, $p\in \NSE{\Delta}$ and $y\in \NSE{Y}$, the $\NSE{}$budget set $\mathscr{B}_{t}(y,p)$ is defined to be: 
\[
\mathscr{B}_t(y,p)=\left\{z\in \NSE{X}_t: p\cdot z \leq p\cdot \hat{e}(t) + \sum_{j\in J}\hat{\theta}_{tj}p\cdot y(j)\right\}.\nonumber
\]

For each $t\in \mathscr{T}_{\Omega}$ and $(x,y,p)\in \NSE{\mathcal{L}^{1}}(\mathscr{T}_{\Omega}, \NSE{\NNReals^{\ell}})\times \NSE{Y}\times \NSE{\Delta}$, let $\mathscr{D}_{t}(x,y,p)$ and $\bar{\mathscr{D}}_t(x,y,p)$ denote the $\NSE{}$demand set and $\NSE{}$quasi-demand set, respectively. That is:
\[
\mathscr{D}_{t}(x,y,p)=\{z \in \mathscr{B}_t(y,p): (w,z)\in \NSE{P}^{\mathscr{T}}_{t}(x,y,p)\implies p\cdot w>p\cdot \hat{e}(t) + \sum_{j\in J}\hat{\theta}_{tj}p\cdot y(j)\}\nonumber
\]
\[
\bar{\mathscr{D}}_t(x,y,p)=\{z \in \mathscr{B}_t(y,p): (w,z)\in \NSE{P}^{\mathscr{T}}_{t}(x,y,p)\implies p\cdot w\geq p\cdot \hat{e}(t) + \sum_{j\in J}\hat{\theta}_{tj}p\cdot y(j)\}.\nonumber
\]
For each $j\in J$, let $\mathcal{S}_j(p)=\argmax_{z \in \NSE{Y}_j} p\cdot z $ denote the (possibly empty) $\NSE{}$supply set at $p\in \NSE{\Delta}$. We now give the definition of hyperfinite (quasi)-equilibrium for $\mathscr{E}$. 

\begin{definition}\label{def_mameqprodhyper}
A hyperfinite quasi-equilibrium for $\mathscr{E}$ is $(\bar{x}, \bar{y}, \bar{p})\in \mathscr{A}\times \NSE{Y}\times \NSE{\Delta}$ such that the following conditions are satisfied:
\begin{enumerate}[(i), leftmargin=*]
\item $\bar{x}(t)\in \bar{\mathscr{D}}_{t}(\bar{x},\bar{y}, \bar{p})$ for all $t\in \mathscr{T}_{\Omega}$ such that $\NSE{\mu}^{\mathscr{T}}(\{t\})>0$;
\item $\bar{y}(j)\in \mathcal{S}_{j}(\bar{p})$ for all $j\in J$;
\item $\sum_{t\in \mathscr{T}_{\Omega}}\bar{x}(t)\NSE{\mu}^{\mathscr{T}}(\{t\})-\sum_{t\in \mathscr{T}_{\Omega}}\hat{e}(t)\NSE{\mu}^{\mathscr{T}}(\{t\})-\sum_{j\in J}\bar{y}(j)=0$.
\end{enumerate}
A hyperfinite equilibrium $(\bar{x}, \bar{y}, \bar{p})\in \mathscr{A}\times \NSE{Y}\times \NSE{\Delta}$ is a hyperfinite quasi-equilibrium with $\bar{x}(t)\in \mathscr{D}_{t}(\bar{x},\bar{y}, \bar{p})$ for all $t\in \mathscr{T}_{\Omega}$ such that $\NSE{\mu}^{\mathscr{T}}(\{t\})>0$.
\end{definition}

We now specify $\hat{e}$ for the hyperfinite production economy $\mathscr{E}$. 
\begin{lemma}\label{defhate}
Suppose $\mathcal{E}$ satisfies \cref{assumptionsurvival}.
Then, there exists an internal function $\hat{e}: \mathscr{T}_{\Omega}\to \NSE{\NNReals^{\ell}}$ such that:
\begin{enumerate}[(i), leftmargin=*]
    \item $\hat{e}(t)\approx \NSE{e}(t)$ for almost all $t\in \mathscr{T}_{\Omega}$;
     \item\label{hyesto0} Let $\mathscr{T}_{\Omega_0}=\bigcup\{T_i: T_i\cap \NSE{\Omega_0}\neq \emptyset\}\cap \mathscr{T}_{\Omega}$. Then $\NSE{\mu}^{\mathscr{T}}(\mathscr{T}_{\Omega_0})>0$ and, for every $t\in \mathscr{T}_{\Omega_0}$, the set $\NSE{X}_{t} - \sum_{j\in J}\hat{\theta}_{tj}\NSE{Y}_j$ has non-empty $\NSE{}$interior $\mathscr{U}_{t}\subset \NSE{\Reals^{\ell}}$ and $\hat{e}(t)\in \mathscr{U}_{t}$;
     \item \label{hycdexplain} there exists 
a commodity $s\in \{1,2,\dotsc,\ell\}$ such that:
\begin{itemize}[leftmargin=*]
    \item for every $t\in \mathscr{T}_{\Omega_0}$, the $\NSE{}$projection $\NSE{\pi}_{s}(\NSE{X}_{t})$ is unbounded, and the consumer $t$ has a strongly monotone preference on the commodity $s$;
    \item for almost all $t\in \mathscr{T}_{\Omega}$, there is an $\NSE{}$open set $\mathscr{V}_{t}$ containing the $s$-th coordinate $\hat{e}(t)_{s}$ of $\hat{e}(t)$ such that $(\hat{e}(t)_{-s},v)\in \NSE{X}_{t} - \sum_{j\in J}\hat{\theta}_{t j}\NSE{Y}_j$ for all $v\in \mathscr{V}_{t}$.
\end{itemize}
\end{enumerate}
\end{lemma}
\begin{proof}
By the second bullet of \cref{cdexplain} in \cref{assumptionsurvival}, we have $e(\omega)=X_{\omega}-\sum_{j\in J}\theta_{\omega j}Y_j$ for almost all $\omega\in \Omega$. By the transfer principle, for almost all $t\in \mathscr{T}_{\Omega}$, we have $\NSE{e}(t)=x_t-\sum_{j\in J}\NSE{\theta}_{tj}y_{j}^{t}$ for some $x_t\in \NSE{X}_{t}$ and $y_{j}^{t}\in \NSE{Y}_{j}$. Let $\hat{e}(t)=x_t-\sum_{j\in J}\hat{\theta}_{tj}y_{j}^{t}$. As $\hat{\theta}(t)\approx \NSE{\theta}(t)$ for all $t\in \mathscr{T}_{\Omega}$, $\hat{e}(t)\approx \NSE{e}(t)$ for almost all $t\in \mathscr{T}_{\Omega}$. 

As $\mu(\Omega_0)>0$, we have $\NSE{\mu}^{\mathscr{T}}(\mathscr{T}_{\Omega_0})>0$. By the construction of $\mathscr{T}_{\Omega}$, every $t\in \mathscr{T}_{\Omega_0}$ is also an element of $\NSE{\Omega}_{0}$. Thus, by the transfer principle, the set $\NSE{X}_{t}-\sum_{j\in J}\NSE{\theta}_{tj}\NSE{Y}_{j}$ has non-empty $\NSE{}$interior $\mathcal{U}_{t}$. Note that, for every $u\in \mathcal{U}_{t}$, we have $u=x^{u}_{t}-\sum_{j\in J}\NSE{\theta}_{tj}y_{j}^{(u,t)}$ for $x^{u}_t\in \NSE{X}_{t}$ and $y_{j}^{(u,t)}\in \NSE{Y}_{j}$. For $u\in \mathcal{U}_{t}$, define $\hat{u}=x^{u}_{t}-\sum_{j\in J}\hat{\theta}_{tj}y_{j}^{(u,t)}$, and let $\mathscr{U}_{t}$ be the collection of all such points $\hat{u}$. It is clear that $\mathscr{U}_{t}$ is $\NSE{}$open subset of $\NSE{X}_{t}-\sum_{j\in J}\hat{\theta}_{tj}\NSE{Y}_{j}$, and $\hat{e}(t)\in \mathscr{U}_{t}$.

As every $t\in \mathscr{T}_{\Omega_0}$ is an element of $\NSE{\Omega_0}$, by the transfer principle, the $\NSE{}$projection $\NSE{\pi}_{s}(\NSE{X}_{t})$ is unbounded, and the consumer $t$ has a strongly monotone preference on the commodity $s$. By the transfer principle, there is an $\NSE{}$open set $\mathcal{V}_{t}$ containing the $s$-th coordinate $\NSE{e}(t)_{s}$ of $\NSE{e}(t)$ such that $(\NSE{e}(t)_{-s},v)\in \NSE{X}_{t} - \sum_{j\in J}\NSE{\theta}_{t j}\NSE{Y}_j$ for all $v\in \mathcal{V}_{t}$. Thus, for each $v\in \mathcal{V}_{t}$, there exist $x^{v}(t)\in \NSE{X}_{t}$ and $y_{j}^{(v,t)}\in \NSE{Y}_{j}$ such that $(\NSE{e}(t)_{-s},v)=x^{v}(t)-\sum_{j\in J}\NSE{\theta}_{t j}y_{j}^{(v,t)}$, which further implies that $v=\NSE{\pi}_{s}\big(x^{v}(t)\big)-\sum_{j\in J}\NSE{\theta}_{t j}\NSE{\pi}_{s}\big(y_{j}^{(v,t)}\big)$.\footnote{$\pi_{s}$ is the projection onto the $s$-th coordinate, and $\NSE{\pi}_{s}$ is the nonstandard extension of $\pi_{s}$.} For $v\in \mathcal{V}_{t}$, define $\hat{v}=\NSE{\pi}_{s}\big(x^{v}(t)\big)-\sum_{j\in J}\hat{\theta}_{t j}\NSE{\pi}_{s}\big(y_{j}^{(v,t)}\big)$, and let $\mathscr{V}_{t}$ be the collection of all such points $\hat{v}$. It is clear that $\mathscr{V}_{t}$ is an $\NSE{}$open set containing the $s$-th coordinate $\hat{e}(t)_{s}$ of $\hat{e}(t)$, and $(\hat{e}(t)_{-s},v)\in \NSE{X}_{t} - \sum_{j\in J}\hat{\theta}_{t j}\NSE{Y}_j$ for all $v\in \mathscr{V}_{t}$.
\end{proof}

We fix $\hat{e}$ as the initial endowment for the hyperfinite production economy $\mathscr{E}$. 
The set $\mathscr{O}$ of hyperfinite attainable consumption-production pairs for $\mathscr{E}$ is: 
\[
\left\{(x',y')\in \NSE{\mathcal{L}^{1}}(\mathscr{T}_{\Omega}, \NSE{\NNReals^{\ell}})\times \NSE{Y}: \sum_{t\in \mathscr{T}_{\Omega}} x'(t) \NSE{\mu}^{\mathscr{T}}(\{t\})-\sum_{t\in \mathscr{T}_{\Omega}} \hat{e}(t)\NSE{\mu}^{\mathscr{T}}(\{t\})-\sum_{j\in J} {y'}(j)=0\right\}.\nonumber
\] 

\begin{lemma}\label{attainpreserve}
Suppose for some $\epsilon>0$, almost all $\omega\in \Omega$ and all $(x,y)\in \cO_{\epsilon}$ with $x(\omega)\in X_\omega$, there exists $u\in X_\omega$ such that $(u, x(\omega))\in \bigcap_{p\in \Delta}P_{\omega}(x,y,p)$. 
Then, for almost all $t\in \mathscr{T}_{\Omega}$, all $(f,y)\in \mathscr{O}$ with $f(t)\in \NSE{X}_{t}$, there exists $z\in \NSE{X}_t$ such that $(z, f(t))\in \bigcap_{p\in \NSE{\Delta}}\NSE{P}^{\mathscr{T}}_{t}(f,y,p)$. 
\end{lemma}
\begin{proof}
Pick $t\in \mathscr{T}_{\Omega}$ with $\NSE{\mu}^{\mathscr{T}}(\{t\})>0$ and $(f,y)\in \mathscr{O}$ with $f(t)\in \NSE{X}_{t}$.
Note that $\NSE{E}(f): \NSE{\Omega}\to \NSE{\NNReals^{\ell}}$ is an internal function such that $\NSE{E}(f)(x)=f(t_i)$ for every $x\in T_i$.
We have 
\[
\sum_{s\in \mathscr{T}_{\Omega}} f(s)\NSE{\mu}^{\mathscr{T}}(\{s\})=\int_{\NSE{\Omega}}\NSE{E}(f)(\omega)\NSE{\mu}(\dee \omega).\nonumber
\]
We also know that $\sum_{s\in \mathscr{T}_{\Omega}} \hat{e}(s)\NSE{\mu}^{\mathscr{T}}(\{s\})\approx \int_{\NSE{\Omega}} \NSE{e}(\omega)\NSE{\mu}(\dee \omega)$. 
So we can conclude that $(\NSE{E}(f),y)\in \NSE{\cO_{\epsilon}}$. 
As $\NSE{E}(f)(t)=f(t)$, by the transfer principle, there exists $z\in \NSE{X}_{t}$ such that $(z, f(t))\in \bigcap_{p\in \NSE{\Delta}}\NSE{P}_{t}(\NSE{E}(f),y,p)$. 
As $\NSE{P}^{\mathscr{T}}_{t}(f,y,p)=\NSE{P}_{t}(\NSE{E}(f),y,p)$ for all $p\in \NSE{\Delta}$, we have the desired result. 
\end{proof}

We now present our main result in this section:

\begin{theorem}\label{hypereqexist}
Suppose that the measure-theoretic production economy $\mathcal{E}$ satisfies \cref{assumptionsurvival}, \cref{assumptionmseco} and:
\begin{enumerate}[(i), leftmargin=*]
    \item for almost all $\omega\in \Omega$, $P_{\omega}$ takes value in $\mathcal{P}_{H}$; 
    \item $\bar{Y}$ is closed, convex, and  $\bar Y \cap (-\bar Y)=\{0\}=\bar Y \cap \NNReals^{\ell}$, where $\bar Y =\left\{\sum_{j\in J} {y}(j): y\in Y\right\}$;
    \item for some $\epsilon>0$, for almost all $\omega\in \Omega$ and all $(x,y)\in \cO_{\epsilon}$ with $x(\omega)\in X_\omega$, there exists $u\in X_\omega$ such that $(u, x(\omega))\in \bigcap_{p\in \Delta}P_{\omega}(x,y,p)$.
\end{enumerate}
The hyperfinite production economy $\mathscr{E}$ has a hyperfinite equilibrium. 
\end{theorem}

\begin{proof}
By the transfer principle, $\NSE{\bar{Y}}$ is $\NSE{}$closed, $\NSE{}$convex, and $\NSE{\bar{Y}}\cap -\NSE{\bar{Y}}=\{0\}=\NSE{\bar{Y}}\cap \NSE{\NNReals^{\ell}}$.
By the transfer of \cref{restpreserve}, \cref{defhate}, \cref{attainpreserve} and the transfer of \cref{finitemainexpd}, there exists a hyperfinite equilibrium. 
\end{proof}

\subsection{Loeb Production Economy}\label{exstlbeqsec}

In this section, we construct a special type of measure theoretic production economy $\Loeb{\mathscr{E}}$, called the Loeb production economy, from the hyperfinite production economy $\mathscr{E}$ defined in \cref{fixhyper}. A Loeb production economy is a measure-theoretic production economy where the consumer space is a hyperfinite Loeb probability space. Under suitable regularity conditions, we establish the existence of a quasi-equilibrium for the Loeb production economy $\Loeb{\mathscr{E}}$. 

\subsubsection{Standard Parts of (Quasi)-Demand Set}
We present two general results on pushing down nonstandard (quasi)-demand set. In particular, we show that, under moderate regularity conditions, if a near-standard point is an element of a nonstandard (quasi)-demand set, then its standard part is an element of the standard part of the nonstandard (quasi)-demand set. 

 Recall that $\mathcal{P}$ is compact with respect to the closed convergence topology. Thus, every $(S, \succ)\in \NSE{\mathcal{P}}$ is near-standard. 
In particular, we have $\ST\big((S, \succ)\big)=(\ST(S), \ST(\succ))$, where $(a, b)\in (\ST(S), \ST(\succ))$ if $a,b\in \ST(S)$ and $u\succ w$ for all $u, w\in S$ such that $u\approx a$ and $w\approx b$. 

\begin{lemma}\label{qsdemand_standard_part}
Suppose that $S\in \NSE{\cK}(\NSE{\NNReals^{\ell}})$, $e\in \NS{\NSE{\NNReals^{\ell}}}$, $\theta\in \NS{\NSE{\NNReals^{|J|}}}$ and $y(j)\in \NS{\NSE{\NNReals^{\ell}}}$ for all $j\in J$. 
Suppose $p \in \NSE \Delta$ such that $p \not \approx 0$, and $(S, \succ)\in \NSE{\mathcal{P}}$. 
Let $\bar{D}(p,e,\theta, y, (S, \succ))$ be 
\[
\{z\in S: p \cdot z\leq p \cdot e+\sum_{j\in J}\theta(j)p\cdot y(j) \wedge
(u, z)\in (S,\succ) \implies p \cdot u \geq p \cdot e+\sum_{j\in J}\theta(j)p\cdot y(j)\}.\nonumber
\]
If $s\in \bar{D}(p,e,\theta, y, (S,\succ))\cap \NS{\NSE{\NNReals^{\ell}}}$, then 
$\ST(s)\in \bar{D}\big(\ST(p), \ST(e), \ST(\theta), \ST(y), (\ST(S), \ST(\succ))\big)$.
\end{lemma}

\begin{proof}
Clearly, we have
$\ST(p)\cdot \ST(s)\leq \ST(p)\cdot \ST(e)+\sum_{j\in J}\ST(\theta)(j)\ST(p)\cdot \ST(y)(j)$.
Suppose that there exists $u\in \ST(S)$ such that $(u, \ST(s))\in (\ST(S), \ST(\succ))$, but 
\[
\ST(p)\cdot u<\ST(p)\cdot \ST(e)+\sum_{j\in J}\ST(\theta)(j)\ST(p)\cdot \ST(y)(j).\nonumber
\]
There is $v\in S$ with $v\approx u$ such that  $(v, s)\in (S, \succ)$. 
Note that $\ST(p)\cdot u\approx p\cdot v$ and
\[
\ST(p)\cdot \ST(e)+\sum_{j\in J}\ST(\theta)(j)\ST(p)\cdot \ST(y)(j)\approx p\cdot e+\sum_{j\in J}\theta(j)p\cdot y(j)\nonumber
\]
As $\ST(p)\cdot u\approx p\cdot v$, we have $p\cdot v<p\cdot e+\sum_{j\in J}\theta(j)p\cdot y(j)$.
This is a contradiction, so $\ST(s)\in \bar{D}\big(\ST(p), \ST(e), \ST(\theta), \ST(y), (\ST(S), \ST(\succ))\big)$.
\end{proof}

The following result is a slight modification of \cref{qsdemand_standard_part}, simply replacing nonstandard quasi-demand set by nonstandard demand set. 

\begin{lemma}\label{demand_standard_part}
Suppose that $S\in \NSE{\cK}(\NSE{\NNReals^{\ell}})$, $e\in \NS{\NSE{\NNReals^{\ell}}}$, $\theta\in \NS{\NSE{\NNReals^{|J|}}}$ and $y(j)\in \NS{\NSE{\NNReals^{\ell}}}$ for all $j\in J$.
Suppose $p\in\NSE \Delta$ such that $p \not \approx 0$. 
Moreover, suppose $(S, \succ)\in \NSE{\mathcal{P}}$, and $x\in S$ for all $x\in \NSE{\Reals^{\ell}}$ such that $x\approx e+\sum_{j\in J}\theta(j)y(j)$.
Let $D(p,e,\theta, y, (S,\succ))$ be 
\[
\{z\in S: p \cdot z\leq p \cdot e+\sum_{j\in J}\theta(j)p\cdot y(j) \wedge
(u, z)\in (S,\succ) \implies p \cdot u>p \cdot e+\sum_{j\in J}\theta(j)p\cdot y(j)\}.\nonumber
\]
If $s\in D(p,e,\theta, y, (S,\succ))\cap \NS{\NSE{\NNReals^{\ell}}}$, then $\ST(s)\in D\big(\ST(p), \ST(e), \ST(\theta), \ST(y), (\ST(S), \ST(\succ))\big)$.
\end{lemma}

\begin{proof}
Clearly, we have
$\ST(p)\cdot \ST(s)\leq \ST(p)\cdot \ST(e)+\sum_{j\in J}\ST(\theta)(j)\ST(p)\cdot \ST(y)(j).$
Suppose that there exists $u\in \ST(S)$ such that $(u, \ST(s))\in (\ST(S), \ST(\succ))$, but 
\[
\ST(p)\cdot u\leq \ST(p)\cdot \ST(e)+\sum_{j\in J}\ST(\theta)(j)\ST(p)\cdot \ST(y)(j).\nonumber
\]
Since $u\in \ST(X)$, we can choose $v\in S$ with $v\approx u$.
Since $\|p\|=1$, we have $e+\sum_{j\in J}\theta(j)y(j)-\lambda p \in S$ for all $\lambda \approx 0$. As $S$ is convex, we have $v_{\lambda}=(1-\lambda)v+\lambda(e+\sum_{j\in J}\theta(j)y(j)-\lambda p)\in S$. Note that $v_{\lambda}\approx v\approx u$ so we have $(v_{\lambda}, s)\in (S,\succ)$. 
\[
p \cdot v_\lambda&=(1-\lambda) p \cdot v + \lambda p \cdot \big(e+\sum_{j\in J}\theta(j)y(j)\big) - \lambda^2 \|p\|\big)\nonumber \\
&\leq \max\{p\cdot v, p\cdot e+\sum_{j\in J}\theta(j)p\cdot y(j)\} -\lambda^2 \|p\|\nonumber\\
&\approx p\cdot e+\sum_{j\in J}\theta(j)p\cdot y(j) -\lambda^2\|p\|.\nonumber
\]
Since $\|p\|=1$, so for $\lambda$ a sufficiently large infinitesimal, $p \cdot v_\lambda \leq p\cdot e+\sum_{j\in J}\theta(j)p\cdot y(j)$, which contradicts with
$s\in D(p,e,\theta, y, (S,\succ))$.
\end{proof}

\subsubsection{Existence of Quasi-Equilibrium in the Loeb Production Economy}

The endowment $e$ and the shareholdings $\theta$ are integrable. As $\hat{e}(t)\approx \NSE{e}(t)$ and $\hat{\theta}(t)\approx \NSE{\theta}(t)$ for all $t\in \mathscr{T}_{\Omega}$, both $\hat{e}$ and $\hat{\theta}$ are S-integrable. Recall that the set $\cK(\NNReals^{\ell})$ of closed and convex subsets of $\NNReals^{\ell}$ is a compact metric space under the closed convergence topology. 
For each $t\in T$, we have $\NSE{X}(t)\in \NSE{\cK}(\NSE{\NNReals^{\ell}})$. Let $\ST(\NSE{X}(t))$ be the standard part of $\NSE{X}(t)$ under the closed convergence topology. The Loeb production economy
\[\label{fixloebprod}
\Loeb{\mathscr{E}}=\{(\ST(\NSE{X}_{t}), \Loeb{\NSE{\succ}^{\mathscr{T}}_{t}}, \ST(\NSE{P}^{\mathscr{T}}_{t}), \ST(\hat{e}_{t}), \ST(\hat{\theta}_{t}))_{t\in \mathscr{T}_{\Omega}}, (Y_j)_{j\in J}, (\mathscr{T}_{\Omega}, \Loeb{I(\mathscr{T}_{\Omega})}, \Loeb{\NSE{\mu}^{\mathscr{T}}})\}
\]
is defined as: 
\begin{enumerate}[(i), leftmargin=*]
    \item $(\mathscr{T}_{\Omega}, \Loeb{I(\mathscr{T}_{\Omega})}, \Loeb{\NSE{\mu}^{\mathscr{T}}})$ is the Loeb probability space generated from 
    $(\mathscr{T}_{\Omega}, I(\mathscr{T}_{\Omega}), \NSE{\mu}^{\mathscr{T}})$, where $I(\mathscr{T}_{\Omega})$ is the collection of all internal subsets of $\mathscr{T}_{\Omega}$;
    \item $J$ is the same finite set of firms;
    \item A Loeb measurable mapping $\ST(\NSE{X}): \mathscr{T}_{\Omega}\to \cK(\NNReals^{\ell})$ given by $\ST(\NSE{X})(t)=\ST(\NSE{X}(t))$. We sometimes use $\ST(\NSE{X})_{t}$ to denote $\ST(\NSE{X})(t)$; 
    \item $Y_{j}\in \NNReals^{\ell}$ is non-empty, denoting the production set of $j$. Note that $Y=\prod_{j\in J}Y_{j}$;
    \item The set of Loeb allocations $\Loeb{\mathscr{A}}$ is:
    \[
    \{f\in \mathcal{L}^{1}\big((\mathscr{T}_{\Omega}, \Loeb{I(\mathscr{T}_{\Omega})}, \Loeb{\NSE{\mu}^{\mathscr{T}}}), \NNReals^{\ell}\big): f(t)\in \ST(\NSE{X})(t) \mbox{ almost surely}\};\nonumber
    \]
    \item\label{choicesint} For each $f\in \mathcal{L}^{1}\big((\mathscr{T}_{\Omega}, \Loeb{I(\mathscr{T}_{\Omega})}, \Loeb{\NSE{\mu}^{\mathscr{T}}}),\NNReals^{\ell}\big)$, pick and fix $F\in \NSE{\mathcal{L}}^{1}(\mathscr{T}_{\Omega}, \NSE{\NNReals^{\ell}})$ such that $F$ is an S-integrable lifting of $f$\footnote{For $f\in \mathcal{L}^{1}\big((\mathscr{T}_{\Omega}, \Loeb{I(\mathscr{T}_{\Omega})}, \Loeb{\NSE{\mu}^{\mathscr{T}}}),\NNReals^{\ell}\big)$, it may have more than one S-integrable lifting. We simply fix one S-integrable lifting for every Loeb integrable function.}. For $t\in \mathscr{T}_{\Omega}$, let $\Loeb{\NSE{M}^{\mathscr{T}}_{t}}=\mathcal{L}^{1}\big((\mathscr{T}_{\Omega}, \Loeb{I(\mathscr{T}_{\Omega})}, \Loeb{\NSE{\mu}^{\mathscr{T}}}),\NNReals^{\ell}\big)\times Y\times \Delta\times X_t$. Let  $\Loeb{\NSE{\succ}^{\mathscr{T}}_{t}}\subset \Loeb{\NSE{M}^{\mathscr{T}}_{t}}\times \Loeb{\NSE{M}^{\mathscr{T}}_{t}}$ be: For $(f_1,y_1,p_1,x_1), (f_2, y_2,p_2, x_2)\in \Loeb{\NSE{M}^{\mathscr{T}}_{t}}$, let $F_1, F_2$ denote the S-integrable liftings associated with $f_1, f_2$, respectively. Then $(f_1,y_1,p_1,x_1)\Loeb{\NSE{\succ}^{\mathscr{T}}_{t}}(f_2, y_2,p_2, x_2)$ if $(F_1,y_1,p_1,a_1)\NSE{\succ}^{\mathscr{T}}_{t}(F_2, y_2,p_2, a_2)$ for all $a_1\approx x_1$ and $a_2\approx x_2$. Let 
    \[
    \ST(\NSE{P}_{t}^{\mathscr{T}}): \mathcal{L}^{1}\big((\mathscr{T}_{\Omega}, \Loeb{I(\mathscr{T}_{\Omega})}, \Loeb{\NSE{\mu}^{\mathscr{T}}}),\NNReals^{\ell}\big)\times Y\times \Delta\to \mathcal{P}\nonumber
    \]
    be its induced preference map. Note that $\ST(\NSE{P}_{t}^{\mathscr{T}})(f,y,p)=\ST\big(\NSE{P}_{t}^{\mathscr{T}}(F, \NSE{y}, \NSE{p})\big)$;
    \item For each $t\in \mathscr{T}_{\Omega}$, $\ST(\hat{\theta})(t)$ represents consumer $t$'s shareholdings. As $\hat{\theta}$ is S-integrable, $\ST(\hat{\theta})(t)$ exists $\Loeb{\NSE{\mu}^{\mathscr{T}}}$-almost surely and $\int_{\mathscr{T}_{\Omega}}\ST(\hat{\theta})(t)(j)\Loeb{\NSE{\mu}^{\mathscr{T}}}(\dee t)=1$ for all $j\in J$. We sometimes write $\ST(\hat{\theta})_{tj}$ for $\ST(\hat{\theta})(t)(j)$;
    \item For each $t\in \mathscr{T}_{\Omega}$, $\ST(\hat{e})(t)$ represents consumer $t$'s endowment. As $\hat{e}$ is S-integrable, $\ST(\hat{e})$ is an element of $\mathcal{L}^{1}\big((\mathscr{T}_{\Omega}, \Loeb{I(\mathscr{T}_{\Omega})}, \Loeb{\NSE{\mu}^{\mathscr{T}}}), \NNReals^{\ell}\big)$. 
\end{enumerate}

For every $t\in \mathscr{T}_{\Omega}$, $p\in \Delta$ and $y\in Y$, the Loeb budget set $\mathbb{B}_{t}(y,p)$ is defined to be: 
\[
\mathbb{B}_t(y,p)=\left\{z\in \ST(\NSE{X}_t): p\cdot z \leq p\cdot \ST(\hat{e})(t) + \sum_{j\in J}\ST(\hat{\theta})_{tj}p\cdot y(j)\right\}.\nonumber
\]

For each $t\in \mathscr{T}_{\Omega}$, let $\mathbb{D}_{t}(x,y,p)$ and $\bar{\mathbb{D}}_{t}(x, y,p)$ denote the (possibly empty) Loeb demand and Loeb quasi-demand set, respectively. That is
\[
\mathbb{D}_{t}(x,y,p)= \{z \in \mathbb{B}_t(y,p): (w, z)\in \ST(\NSE{P}_{t}^{\mathscr{T}})(x,y,p)\nonumber 
\implies p\cdot w> p\cdot \ST(\hat{e})(t) + \sum_{j\in J}\ST(\hat{\theta})_{tj}p\cdot y(j)\}\nonumber
\]
\[
\bar{\mathbb{D}}_t(x,y,p) = \{z \in \mathbb{B}_t(y,p): (w, z)\in \ST(\NSE{P}_{t}^{\mathscr{T}})(x,y,p)\nonumber
\implies p\cdot w\geq p\cdot \ST(\hat{e})(t) + \sum_{j\in J}\ST(\hat{\theta})_{tj}p\cdot y(j)\}\nonumber
\]
at $(x,y,p)\in \mathcal{L}^{1}\big((\mathscr{T}_{\Omega}, \Loeb{I(\mathscr{T}_{\Omega})}, \Loeb{\NSE{\mu}^{\mathscr{T}}}),\NNReals^{\ell}\big)\times Y\times \Delta$.
For $j\in J$, let $\mathbb{S}_j(p)=\argmax_{z\in Y_j}p\cdot z$ denote the (possibly empty) Loeb supply set at $p\in \Delta$. We now give the definition of a Loeb (quasi)-equilibrium for the Loeb production economy $\Loeb{\mathscr{E}}$. 

\begin{definition}
A Loeb quasi-equilibrium for $\Loeb{\mathscr{E}}$ is a tuple $(\bar{x}, \bar{y}, \bar{p})\in \Loeb{\mathscr{A}}\times Y\times \Delta$ such that the following conditions are satisfied:

\begin{enumerate}[(i), leftmargin=*]
\item $\bar{x}(t)\in \bar{\mathbb{D}}_{t}(\bar{x},\bar{y}, \bar{p})$ for $\Loeb{\NSE{\mu}^{\mathscr{T}}}$-almost all $t\in \mathscr{T}_{\Omega}$;
\item $\bar{y}(j)\in \mathbb{S}_{j}(\bar{p})$ for all $j\in J$;
\item $\int_{\mathscr{T}_{\Omega}}\bar{x}(t)\Loeb{\NSE{\mu}^{\mathscr{T}}}(\dee t)-\int_{\mathscr{T}_{\Omega}}\ST(\hat{e})(t)\Loeb{\NSE{\mu}^{\mathscr{T}}}(\dee t)-\sum_{j\in J}\bar{y}(j)=0$.
\end{enumerate}

A Loeb  equilibrium $(\bar{x}, \bar{y}, \bar{p})\in \Loeb{\mathscr{A}}\times Y\times \Delta$ for $\Loeb{\mathscr{E}}$ is a Loeb quasi-equilibrium with $\bar{x}(t)\in \mathbb{D}_{t}(\bar{x},\bar{y}, \bar{p})$ for $\Loeb{\NSE{\mu}^{\mathscr{T}}}$-almost all $t\in \mathscr{T}_{\Omega}$. 
\end{definition}

To establish the existence of quasi-equilibrium in $\Loeb{\mathscr{E}}$, we assume:

\begin{assumption}\label{assumptionuc}
For each $\omega\in \Omega$, the preference map
\[
P_{\omega}: \mathcal{L}^{1}(\Omega, \NNReals^{\ell})\times Y\times \Delta\to \mathcal{P}\nonumber
\]
is uniformly continuous in the norm topology on $\mathcal{L}^{1}(\Omega,\NNReals^{\ell})\times Y\times \Delta$.\footnote{Uniform continuity depends on the underlying metric. However, as $\mathcal{P}$ is a compact metric space, if $P_{\omega}$ is uniformly continuous with respect to one metric on $\mathcal{P}$, then $P_{\omega}$ is uniformly continuous with respect to any metric that generates the same topology on $\mathcal{P}$.} 
\end{assumption}
\begin{remark}\label{suffsintcond}
Let $\mathcal{V}$ be the collection of all functions $v: \PosReals\to \PosReals$ such that $\lim_{x\to 0}v(x)=0$. 
For each $v\in \mathcal{V}$, let:
\[
\mathcal{L}^{1}_{v}=\{f\in \mathcal{L}^{1}(\Omega, \NNReals^{\ell}): (\forall \epsilon>0)(\forall E\in \BorelSets{\Omega})(\mu(E)<v(\epsilon)\implies \int_{E}f(\omega)\mu(\dee \omega)<\epsilon)\}.\nonumber
\]
In fact, to obtain the main result of this section, we only need $\NSE{P}_{\omega}$ to be S-continuous at S-integrable allocations. That is, 
we only need to assume that: For each $\omega\in \Omega$ and $v\in \mathcal{V}$, the preference map
$P_{\omega}: \mathcal{L}^{1}(\Omega, \NNReals^{\ell})\times Y\times \Delta\to \mathcal{P}$
is uniformly continuous in the norm topology on $\mathcal{L}^{1}_{v}\times Y\times \Delta$. 
\end{remark}

\begin{lemma}\label{prefscts}
Suppose $\mathcal{E}$ satisfies \cref{assumptionmseco} and \cref{assumptionuc}. 
Let $F_1, F_2\in \NSE{\mathcal{L}}^{1}(\NSE{\Omega}, \NSE{\NNReals^{\ell}})$ be such that $F_1\approx F_2$, $y_1, y_2\in \NSE{Y}$ be such that $y_1\approx y_2$ and $p_1, p_2\in \NSE{\Delta}$ be such that $p_1\approx p_2$. 
Then, for $\Loeb{\NSE{\mu}^{\mathscr{T}}}$ almost all $\omega\in \NSE{\Omega}$, $\NSE{P}_{\omega}(F_1, y_1, p_1)\approx \NSE{P}_{\omega}(F_2, y_2, p_2)$. 
\end{lemma}
\begin{proof}
Pick $F_1, F_2\in \NSE{\mathcal{L}}^{1}(\NSE{\Omega}, \NSE{\NNReals^{\ell}})$, $y_1, y_2\in \NSE{Y}$ and $p_1, p_2\in \NSE{\Delta}$ such that $F_1\approx F_2$, $y_1\approx y_2$ and $p_1\approx p_2$. 
Recall that $\chi: \Omega\to \cC[\mathcal{L}^{1}(\Omega,\NNReals^{\ell})\times Y\times \Delta, \mathcal{P}]$ is a measurable function. 
By \cref{hyptsexist},  there exists a $\Loeb{\NSE{\mu}^{\mathscr{T}}}$-measurable set $U$ with $\Loeb{\NSE{\mu}^{\mathscr{T}}}(U)=1$ such that $\NSE{\chi}(u)\approx \chi(\ST(u))$ for all $u\in U$. 
Pick $\omega\in U$. By \cref{assumptionuc}, we have 
\[
\NSE{P}_{\omega}(F_1, y_1,p_1)\approx \NSE{P}_{\ST(\omega)}(F_1, y_1, p_1)\approx \NSE{P}_{\ST(\omega)}(F_2, y_2,p_2)\approx \NSE{P}_{\omega}(F_2,y_2, p_2), \nonumber 
\]
completing the proof. \end{proof}

\begin{theorem}\label{estfreenonfree}
Suppose that the measure-theoretic production economy $\mathcal{E}$ satisfies \cref{assumptionmseco} and \cref{assumptionuc}. 
If the hyperfinite production economy $\mathscr{E}$ has a hyperfinite quasi-equilibrium $(\bar{x}, \bar{y}, \bar{p})$ such that
\begin{enumerate}[(i),leftmargin=*]
    \item the quasi-equilibrium allocation $\bar{x}$ is S-integrable;
    \item the quasi-equilibrium production $\bar{y}$ is near-standard, and $\ST(\bar{y})\in Y$.
\end{enumerate}
Then $(\ST(\bar{x}), \ST(\bar{y}), \ST(\bar{p}))$ is a Loeb quasi-equilibrium for $\Loeb{\mathscr{E}}$.
\end{theorem}
\begin{proof}
Let $(\bar{x}, \bar{y}, \bar{p})$ be a hyperfinite quasi-equilibrium for $\mathscr{E}$ such that $\bar{x}$ is S-integrable and $\ST(\bar{y})\in Y$. As $\bar{p}\in \NSE{\Delta}$, we have $\ST(\bar{p})\in \Delta$. 
Note that we have $\bar{x}(t)\in \NSE{X}_{t}$ for $\NSE{\mu}^{\mathscr{S}}$-almost all $t\in \mathscr{T}_{\Omega}$. 
As $X_{\omega}\in \cK(\NNReals^{\ell})$ for all $\omega\in \Omega$ and $\bar{x}$ is S-integrable, we have $\ST(\bar{x})\in \Loeb{\mathscr{A}}$.

\begin{claim}\label{inqdclaim}
$\ST(\bar{x})(t)\in \bar{\mathbb{D}}_{t}(\ST(\bar{x}),\ST(\bar{y}),\ST(\bar{p}))$ for $\Loeb{\NSE{\mu}^{\mathscr{T}}}$-almost all $t\in \mathscr{T}_{\Omega}$.
\end{claim}
\begin{proof}
By \cref{hyptsexist}, there exists a set $Z\subset \mathscr{T}_{\Omega}$ with $\Loeb{\NSE{\mu}^{\mathscr{S}}}(Z)=1$ such that 
$\NSE{e}(z)\approx e(\ST(z))$, $\NSE{\chi}(z)\approx \chi(\ST(z))$, $\NSE{\theta}(z)\approx \theta(\ST(z))$ and $\NSE{X}(z)\approx X(\ST(z))$ for all $z\in Z$. 
In particular, we know that $\hat{e}(z)\approx \NSE{e}(z)\in \NS{\NSE{\NNReals^{\ell}}}$, $\hat{\theta}(z)\approx \NSE{\theta}(z)\in \NS{\NSE{\NNReals^{|J|}}}$ and $\ST(\NSE{X})(z)\in \cK(\NNReals^{\ell})$ is non-empty for all $z\in Z$. By moving to a subset of $Z$ with $\Loeb{\NSE{\mu}^{\mathscr{S}}}$-measure $1$ if necessary, we can assume that $\bar{x}(z)$ is near-standard for all $z\in Z$.

We first show that $\ST(\bar{x})(z)\in \mathbb{B}_{z}(\ST(\bar{y}), \ST(\bar{p}))$ for all $z\in Z$. We have
\[
\ST(\bar{p})\cdot \ST(\bar{x})(z)&\approx \bar{p}\cdot \bar{x}(z)\leq \bar{p}\cdot \hat{e}(z)+\sum_{j\in J}\hat{\theta}_{zj}\bar{p}\cdot \bar{y}(j)\nonumber \\
&\approx \ST(\bar{p})\cdot \ST(\hat{e})(z)+\sum_{j\in J}\ST(\hat{\theta})_{zj}\ST(\bar{p})\cdot \ST(\bar{y})(j).\nonumber
\]
Hence, we conclude that $\ST(\bar{x})(z)\in \mathbb{B}_{z}(\ST(\bar{y}), \ST(\bar{p}))$ for all $z\in Z$.

Let $\bar{F}\in \NSE{\mathcal{L}}^{1}(\mathscr{T}_{\Omega}, \NSE{\NNReals^{\ell}})$ be the S-integrable lifting associated with $\ST(\bar{x})$ as specified in \cref{choicesint} in the construction of $\Loeb{\mathscr{E}}$. Hence, we have: 
\[
\ST(\NSE{P}_{t}^{\mathscr{T}})(\ST(\bar{x}), \ST(\bar{y}), \ST(\bar{p}))=\ST\big(\NSE{P}_{t}^{\mathscr{T}}(\bar{F}, \ST(\bar{y}),\ST(\bar{p}))\big)\nonumber
\]
for all $t\in \mathscr{T}_{\Omega}$. For every $z\in Z$, by \cref{prefscts}, we have
\[
\ST(\NSE{P}_{z}^{\mathscr{T}}(\bar{F}, \ST(\bar{y}),\ST(\bar{p})))&=\ST(\NSE{P}_{z}(\NSE{E}(\bar{F}), \ST(\bar{y}),\ST(\bar{p})))
=\ST(\NSE{P}_{z}(\NSE{E}(\bar{x}), \ST(\bar{y}), \ST(\bar{p})))\nonumber \\
&=\ST(\NSE{P}_{z}^{\mathscr{T}}(\bar{x}, \bar{y}, \bar{p})).\nonumber
\]
By \cref{qsdemand_standard_part}, we have $\ST(\bar{x})(z)=\ST(\bar{x}(z))\in \bar{\mathbb{D}}_{z}(\ST(\bar{x}), \ST(\bar{y}),\ST(\bar{p}))$ for all $z\in Z$. 
\end{proof}

\begin{claim}\label{insupclaim}
$\ST(\bar{y})(j)\in \mathbb{S}_{j}(\ST(\bar{p}))$ for all $j\in J$. 
\end{claim}
\begin{proof}
Pick $j\in J$. 
By assumption, $\ST(\bar{y})(j)$ is an element of $Y_j$. 
As $\bar{y}_j\in \mathcal{S}_j(\bar{p})$, we have $\bar{y}_j\in \argmax_{z\in \NSE{Y}_j}\bar{p}\cdot z$. 
Thus, we conclude that $\ST(\bar{y})(j)\in \argmax_{z\in Y_j}\ST(\bar{p})\cdot z$. 
\end{proof}

Note that $\int_{t\in \mathscr{T}_{\Omega}}\ST(\bar{x})(t)\Loeb{\NSE{\mu}^{\mathscr{T}}}(\dee t)\approx \sum_{t\in \mathscr{T}_{\Omega}}\bar{x}(t)\NSE{\mu}^{\mathscr{T}}(\{t\})$ and 
\[
\sum_{t\in \mathscr{T}_{\Omega}}\hat{e}(t)\NSE{\mu}^{\mathscr{T}}(\{t\})+\sum_{j\in J}\bar{y}(j)
\approx \int_{t\in \mathscr{T}_{\Omega}}\ST(\hat{e})(t)\Loeb{\NSE{\mu}^{\mathscr{T}}}(\dee t)+\sum_{j\in J}\ST(\bar{y})(j).\nonumber
\]
We have $\sum_{t\in \mathscr{T}_{\Omega}}\bar{x}(t)\NSE{\mu}^{\mathscr{T}}(\{t\})-\sum_{t\in \mathscr{T}_{\Omega}}\hat{e}(t)\NSE{\mu}^{\mathscr{T}}(\{t\})-\sum_{j\in J}\bar{y}(j)=0$ since $(\bar{x}, \bar{y}, \bar{p})$ is a hyperfinite quasi-equilibrium.
Hence, we conclude that 
\[
\int_{t\in \mathscr{T}_{\Omega}}\ST(\bar{x})(t)\Loeb{\NSE{\mu}^{\mathscr{T}}}(\dee t)-\int_{t\in \mathscr{T}_{\Omega}}\ST(\hat{e})(t)\Loeb{\NSE{\mu}^{\mathscr{T}}}(\dee t)-\sum_{j\in J}\ST(\bar{y})(j)=0.\nonumber
\]
Combining Claims \ref{inqdclaim} and \ref{insupclaim}, $(\ST(\bar{x}), \ST(\bar{y}), \ST(\bar{p}))$ is a Loeb quasi-equilibrium for $\Loeb{\mathscr{E}}$. 
\end{proof}

We now show that \cref{assumptionalocbound} implies the assumptions of \cref{estfreenonfree}.

\begin{lemma}\label{monotonesint}
Suppose $\mathcal{E}$ satisfies \cref{assumptionmseco}, \cref{assumptionalocbound} and \cref{assumptionuc}. Let $(\bar{x}, \bar{y}, \bar{p})$ be a hyperfinite quasi-equilibrium for the hyperfinite production economy $\mathscr{E}$. If \cref{esto0} in \cref{assumptionsurvival} 
is satisfied, and $\bar{y}$ is near-standard, then $\bar{x}$ is S-integrable. 
\end{lemma}
\begin{proof}
By \cref{badintegrable} of \cref{assumptionalocbound}, $\proj{k}\circ \bar{x}$ is S-integrable. 
We now show that $\bar{p}_{j}$ is positive and non-infinitesimal for all $j>k$.  
Suppose not. Without loss of generality, we assume that $\bar{p}_{k+1}$ is infinitesimal or negative.
As $(\bar{x},\bar{y},\bar{p})$ is a hyperfinite quasi-equilibrium, by the same argument in \cref{attainpreserve}, $(\NSE{E}(\bar{x}),\bar{y})\in \NSE{\cO_{\epsilon_{k+1}}}$ for the same $\epsilon_{k+1}$ in \cref{goodmonotone} of \cref{assumptionalocbound}. 
Thus, there exists $t_0\in \mathscr{T}_{\Omega}\cap \NSE{\Omega_0}$ such that 
\begin{enumerate}[leftmargin=*]
    \item $\NSE{\mu}^{\mathscr{T}}(\{t_0\})>0$ and $\ST(\bar{x}(t_0))$ exists;
    \item $\NSE{P}^{\mathscr{T}}_{t_0}(\bar{x},\bar{y},\bar{p})\in \NSE{M_{k+1}}$.
\end{enumerate}
Note that $\bar{x}(t_0)\in \bar{\mathscr{D}}_{t_0}(\bar{x},\bar{y}, \bar{p})$. By \cref{defhate}, the fact that $t_0\in \NSE{\Omega_0}$ and the same proof as in \cref{omega0equil}, we have $\bar{x}(t_0)\in \mathscr{D}_{t_0}(\bar{x},\bar{y}, \bar{p})$. 
Hence, by \cref{esto0} in \cref{assumptionsurvival}, \cref{defhate} and \cref{demand_standard_part}, there is no $w$ with $\ST(\bar{p})\cdot w\leq \ST(\bar{p})\cdot \ST(\hat{e}(t_0))+\sum_{j\in J}\hat{\theta}_{t_0j}\ST(\bar{p})\cdot \ST(\bar{y}(j))$ such that $\big(w,\ST(\bar{x}(t_0))\big)\in \ST(\NSE{P}^{\mathscr{T}}_{t_0}(\bar{x},\bar{y},\bar{p}))$.
By Proposition 2 of \citet{grodal74}, $M_{k+1}$ is compact. Hence, the preference $\ST(\NSE{P}^{\mathscr{T}}_{t_0}(\bar{x},\bar{y},\bar{p}))$ is strongly monotonic on the commodity $k+1$. 
As the price of commodity $k+1$ is infinitesimal or negative, we can pick $w'$ to be $\ST(\bar{x}(t_0))$ plus one extra unit of good $k+1$. We then have $\big(w',\ST(\bar{x}(t_0))\big)\in \ST(\NSE{P}^{\mathscr{T}}_{t_0}(\bar{x},\bar{y},\bar{p}))$ and $\ST(\bar{p})\cdot w'\leq \ST(\bar{p})\cdot \ST(\hat{e}(t_0))+\sum_{j\in J}\hat{\theta}_{t_0j}\ST(\bar{p})\cdot \ST(\bar{y}(j))$.
This is a contradiction, hence $\bar{p}_{j}$ is strictly positive and non-infinitesimal for all $j>k$. 

Let $\proj{(\ell-k)}$ be the projection onto the coordinates $k+1,\dotsc,\ell$.
Recall that $\psi$ is the integrable function in \cref{badintegrable} of \cref{assumptionalocbound}.
As $\bar{y}$ is near-standard, there are $r\in \PosReals$ and $n\leq k$ such that $\|\proj{(\ell-k)}\left(\bar{x}(t)\right)\|\leq r\|\hat{e}(t)\|+\|\NSE{\psi}_{n}(t)\|$\footnote{As usual, $\NSE{\psi}_{n}(t)$ is the $n$-th coordinate of $\NSE{\psi}(t)$.} for all $t$ with $\NSE{\mu}^{\mathscr{T}}(\{t\})>0$. 
As $\hat{e}$ and $\NSE{\psi}$ are S-integrable, so is $\bar{x}$.
\end{proof}

We now present the main result of this section:

\begin{theorem}\label{alocboundlemma}
Suppose $\mathcal{E}$ satisfies \cref{assumptionsurvival}, \cref{assumptionmseco}, \cref{assumptionalocbound}, \cref{assumptionuc}, and the following conditions:
\begin{enumerate}[(i), leftmargin=*]
\item for almost all $\omega\in \Omega$, $P_\omega$ takes value in  $\mathcal{P}_{H}$;
\item\label{loebsatiate} for some $\epsilon>0$, for almost all $\omega\in \Omega$ and all $(x,y)\in \cO_{\epsilon}$ such that $x(\omega)\in X_{\omega}$, there exists $u\in X_{\omega}$ such that $(u, x(\omega))\in \bigcap_{p\in \Delta}P_{\omega}(x,y,p)$;

\item $\bar{Y}$ is closed and convex, $\bar Y \cap (-\bar Y)=\bar Y \cap \NNReals^{\ell}=\{0\}$;

\item $Y_j$ is closed for all $j\in J$.
 \end{enumerate}
Then, $\Loeb{\mathscr{E}}$ has a Loeb quasi-equilibrium.\footnote{Using a similar argument as in \cref{interioritylm}, we can in fact establish the existence of a Loeb equilibrium under the same set of assumptions. On the other hand, we do not need the full strength of \cref{cdexplain} in \cref{assumptionsurvival} to establish the existence of a Loeb quasi-equilibrium. In fact, if we instead assume $e(\omega)\in X_{\omega}-\sum_{j\in J}\theta_{\omega j}Y_j$ for almost all $\omega\in \Omega$, then we can establish the existence of a hyperfinite quasi-equilibrium in the hyperfinite production economy $\mathscr{E}$, which, by \cref{monotonesint}, implies $\Loeb{\mathscr{E}}$ has a Loeb quasi-equilibrium.}
\end{theorem}
\begin{proof}
By \cref{hypereqexist}, the hyperfinite weighted production economy $\mathscr{E}$ has a hyperfinite equilibrium $(\bar{f},\bar{y},\bar{p})$. Hence,we have:  
\[
\sum_{t\in \mathscr{T}_{\Omega}}\bar{f}(t)\NSE{\mu}^{\mathscr{T}}(\{t\})-\sum_{t\in \mathscr{T}_{\Omega}}\hat{e}(t)\NSE{\mu}^{\mathscr{T}}(\{t\})-\sum_{j\in J}\bar{y}(j)=0.\nonumber
\]
Since $\sum_{t\in \mathscr{T}_{\Omega}}\bar{f}(t)\NSE{\mu}^{\mathscr{T}}(\{t\})$ and $\sum_{t\in \mathscr{T}_{\Omega}}\hat{e}(t)\NSE{\mu}^{\mathscr{T}}(\{t\})$ are near-standard,
by Theorem 2 in page 77 of \citet{de59}, $\bar{y}$ is near-standard. 
By \cref{monotonesint}, $\bar{f}$ is S-integrable. 
As $Y_{j}$'s are closed, we have $\ST(\bar{y})\in Y$.
By \cref{estfreenonfree}, $(\ST(\bar{f}), \ST(\bar{y}), \ST(\bar{p}))$ is a Loeb quasi-equilibrium.   
\end{proof}

\subsection{Measure-theoretic Production Economy}\label{seceqestorig}

In this section, we establish equilibrium existence for the measure-theoretic production economy $\mathcal{E}$, by  constructing an equilibrium from a Loeb quasi-equilibrium of the Loeb production economy $\Loeb{\mathscr{E}}$. 

\subsubsection{Convexity of the Quasi-Demand Set}\label{cqdappendix}
In this section, we provide sufficient conditions on the preference map under which the quasi-demand set is convex. The convexity of the quasi-demand is needed for the push down of Loeb quasi-equilibrium allocation to be in the quasi-demand set of the measure theoretic economy. 
For $(C,\succ)\in \mathcal P$, let $\succsim$ be the derived weak preference on $C$.\footnote{For $a, b\in C$, we say $a$ is weakly preferred to $b$ and write $a\succsim b$ if $b\not\succ a$. It is easy to verify that $\succsim$ is complete and reflexive. $\succsim$ is in addition transitive if $\succ$ is negatively transitive.} The following result is stated in \citet{de59} without a proof. 

\begin{lemma}[{\citep[][page.~59]{de59}}]\label{thm_convex}
Let $(C,\succ)\in \mathcal P_H^{-}$ be a preference. 
Then the derived weak preference $\succsim$ is convex. 
\end{lemma}
 
\begin{proof}
Assume that $\succsim$ is not convex. Then there exist $x,y,z\in C$ with $y\neq z$ and $\lambda \in (0,1)$ such that $y,z\succsim x$ but $\lambda y +(1-\lambda)z\not\succsim x$. By the definition of $\succsim$, $x\succ \lambda y +(1-\lambda)z$. 
\begin{claim}\label{succclaim}
$y,z \succ \lambda y +(1-\lambda)z$.
\end{claim}
\begin{proof}
It is sufficient to show that $y\succ \lambda y+(1-\lambda) z$. 
Suppose not. Then we have $\lambda y+(1-\lambda) z \succsim y$. 
By the negative transitivity of $\succ$, we have $x\succsim y$. 
If $\lambda y+(1-\lambda) z\succ y$, by the transitivity of $\succ$, we have $x\succ y$, which is a contradiction. If $\lambda y+(1-\lambda) z\sim y$, by the transitivity of $\sim$, we have $x\sim \lambda y+(1-\lambda) z$, which is also a contradiction. 
\end{proof}
By \cref{succclaim} and the convexity of $\succ$,   $\lambda y +(1-\lambda)z\succ \lambda y +(1-\lambda)z$, which yields a contradiction. Hence, $\succsim$ is convex. 
\end{proof}

\begin{theorem}\label{thm_convexdemandsupply} 
Suppose the preference map $P_{\omega}$ takes value in $\mathcal{P}_{H}^{-}$. Then, the quasi-demand set $\bar{D}_{\omega}(x,y,p)$ is convex for every $(x,y,p)\in \mathcal{L}^{1}(\Omega, \NNReals^{\ell})\times Y\times \Delta$.
\end{theorem}
\begin{proof}
Fix $(x,y,p)\in \mathcal{L}^{1}(\Omega, \NNReals^{\ell})\times Y\times \Delta$. 
Suppose $z_1, z_2$ are two elements of $\bar{D}_{\omega}(x,y,p)$. 
Pick $\lambda\in (0, 1)$. 
For every $w\in X_{\omega}$ such that $p\cdot w<p\cdot e(\omega)+\sum_{j\in J}\theta_{\omega j}p\cdot y(j)$, we have $z_1, z_2\succsim_{x,y,\omega,p} w$.
By \cref{thm_convex}, we have $\lambda z_1+(1-\lambda)z_2\succsim_{x,y,\omega,p} w$. Hence, we have $\lambda z_1+(1-\lambda)z_2 \in \bar{D}_{\omega}(x,y,p)$, completing the proof. 
\end{proof}

\subsubsection{Extension of the Strong Lusin Theorem}
The strong Lusin theorem is equivalent to the Lusin theorem if the Tietze extension theorem holds.
The classical Tietze extension theorem assumes that the range space is a Euclidean space.\footnote{The Tietze extension theorem can fail if the domain is connected while the range is disconnected.} In this section, we present an extension of the strong Lusin Theorem when the range is a space of subsets of $\NNReals^{\ell}$. 

\citet{Dugun51} provides the following generalization of the Tietze extension theorem:
\begin{theorem}[{\citep[][Theorem.~4.1]{Dugun51}}]\label{exttieze}
Let $X$ be an arbitrary metric space, $X'$ a closed subset of $X$, $\mathcal{L}$ a locally convex topological vector space, and $f: X'\to \mathcal{L}$ a continuous map. Then there exists a continuous extension $F: X\to \mathcal{L}$ of $f$. Further more, the range of $F$ is a subset of the convex hull of the range of $f$.
\end{theorem}
We are particularly interested in the case where the range space is the set of bounded, closed and convex subsets of $\NNReals^{\ell}$, which we denote by $\cK_{\mathrm{bd}}(\NNReals^{\ell})$. However, $\cK_{\mathrm{bd}}(\NNReals^{\ell})$ equipped with the Minkowski sum and the scalar multiplication is not a vector space since there does not exist an additive inverse for a generic element of $\cK_{\mathrm{bd}}(\NNReals^{\ell})$. On the other hand, it is easy to verify that $\cK_{\mathrm{bd}}(\NNReals^{\ell})$ satisfies the following conditions:
\begin{itemize}[leftmargin=*]
    \item $\cK_{\mathrm{bd}}(\NNReals^{\ell})$ is closed under the Minkowski sum and non-negative scalar multiplication;
    \item If $A\in \cK_{\mathrm{bd}}(\NNReals^{\ell})$ and $S$ is the unit sphere of $\NNReals^{\ell}$, then $A+S$ is closed;
    \item $\cK_{\mathrm{bd}}(\NNReals^{\ell})$ is metrized by the Hausdorff metric. 
\end{itemize}
Theorem 2 in \citet{rd51} implies that $\cK_{\mathrm{bd}}(\NNReals^{\ell})$ can be embedded as a convex cone in a real normed vector space $\mathcal{N}$ such that:
\begin{itemize}[leftmargin=*]
    \item the embedding is isometric;
    \item addition in $\mathcal{N}$ induces addition in $\cK_{\mathrm{bd}}(\NNReals^{\ell})$;
    \item multiplication by non-negative scalars induces the corresponding operation in $\cK_{\mathrm{bd}}(\NNReals^{\ell})$. 
\end{itemize}
\begin{theorem}\label{stlusinset}
Suppose $(M,\BorelSets M, P)$ is a Borel probability space where $M$ is Polish. Let $\cK_{\mathrm{bd}}(\NNReals^{\ell})$ be endowed with the closed convergence topology, and $\Phi: M\to \cK_{\mathrm{bd}}(\NNReals^{\ell})$ be a measurable mapping. Then, for every $\epsilon>0$, there is a compact set $K\subset M$ and a continuous function $\Phi': M\to \cK_{\mathrm{bd}}(\NNReals^{\ell})$ such that $P(K)>1-\epsilon$ and $\Phi=\Phi'$ on $K$. 
\end{theorem}
\begin{proof}
Pick $\epsilon>0$. By the Lusin theorem, there is a compact set $K\subset M$ such that $\Phi$ is continuous on $K$ and $P(K)>1-\epsilon$. Let $\kappa$ be the isometric embedding in Theorem 2 of \citet{rd51}. Then $\kappa\circ\Phi: M\to \kappa\big(\cK_{\mathrm{bd}}(\NNReals^{\ell})\big)$ is continuous on $K$. By \cref{exttieze}, there is a continuous function $\Xi: M\to \kappa\big(\cK_{\mathrm{bd}}(\NNReals^{\ell})\big)$\footnote{By Theorem 2 of \citet{rd51}, the set $\kappa\big(\cK_{\mathrm{bd}}(\NNReals^{\ell})\big)$ is convex.} such that $\Xi=\kappa\circ\Phi$ on $K$. Let $\Phi'=\kappa^{-1}\circ \Xi$. Then $\Phi'$ is continuous from $M$ to $\cK_{\mathrm{bd}}(\NNReals^{\ell})$ such that $\Phi'=\Phi$ on $K$.    
\end{proof}

\subsubsection{Existence of Equilibrium in the Original Measure-theoretic Economy}
We start with the following technical result on S-integrable functions and the weak topology. 

\begin{lemma}\label{sintns}
Let $F\in \NSE{\mathcal{L}^{1}}(\NSE{\Omega},\NSE{\NNReals^{\ell}})$ be S-integrable. 
Then $F$ is a near-standard element in $\NSE{\mathcal{L}^{1}}(\NSE{\Omega}, \NSE{\NNReals^{\ell}})$ under the weak topology. 
\end{lemma}
\begin{proof}
Let $F\in \NSE{\mathcal{L}^{1}}(\NSE{\Omega},\NSE{\NNReals^{\ell}})$ be S-integrable. 
Then $\ST(F): \NSE{\Omega}\to \NNReals^{\ell}$ is Loeb measurable. 
Let $\mathcal{G}$ denote the $\sigma$-algebra generated by $\{\ST^{-1}(B): B\in \BorelSets {\Omega}\}$.
Let $\bar{F}=\mathbb{E}(\ST(F)|\mathcal{G}): \NSE{\Omega}\to \NNReals^{\ell}$ be the conditional expectation of $\ST(F)$ with respect to the $\sigma$-algebra $\mathcal{G}$. 
Note that $\bar{F}$ is constants over monads. 
Define $f: \Omega\to \NNReals^{\ell}$ to be $f(\omega)=\bar{F}(\omega)$.  
Since we have
\[
\int_{\Omega}f(\omega)\mu(\dee \omega)=\int_{\NS{\NSE{\Omega}}}\bar{F}(\omega)\Loeb{\NSE{\mu}}(\dee \omega)=\int_{\NS{\NSE{\Omega}}}\ST(F)(\omega)\Loeb{\NSE{\mu}}(\dee \omega),\nonumber
\]
we conclude that $f\in \mathcal{L}^{1}(\Omega, \NNReals^{\ell})$.

Pick any $g\in \mathcal{L}^{\infty}(\Omega, \NNReals^{\ell})$. 
As $g$ is essentially uniformly bounded, by \cref{measurescts}, we have $\ST(\NSE{g}(\omega))=g(\ST(\omega))$ for $\Loeb{\NSE{\mu}}$-almost all $\omega\in \NSE{\Omega}$. Then we have
\[
\int_{\NSE{\Omega}}F(\omega)\NSE{g}(\omega)\NSE{\mu}(\dee \omega)&\approx \int_{\NS{\NSE{\Omega}}} \ST(F)(\omega) \ST(\NSE{g})(\omega)\Loeb{\NSE{\mu}}(\dee \omega)
=\int_{\NS{\NSE{\Omega}}}\mathbb{E}(\ST(F))\ST(\NSE{g})|\mathcal{G})(\omega)\Loeb{\NSE{\mu}}(\dee \omega)\nonumber \\
&=\int_{\NS{\NSE{\Omega}}}\bar{F}(\omega)\ST(\NSE{g})(\omega)\Loeb{\NSE{\mu}}(\dee \omega)
=\int_{\Omega}f(\omega)g(\omega)\mu(\dee \omega).\nonumber
\]
Thus, $F$ is in the monad of $f$ with respect to the weak topology on $\mathcal{L}^{1}(\Omega, \NNReals^{\ell})$. 
\end{proof}

We use $\ST_{w}$ to denote the standard part map from $\NSE{\mathcal{L}^{1}}(\NSE{\Omega}, \NSE{\NNReals^{\ell}})$ to $\mathcal{L}^{1}(\Omega, \NNReals^{\ell})$ with respect to the weak topology. 
In particular, for an S-integrable function $F$, $\ST_{w}(F)$ is the standard function $f\in \mathcal{L}^{1}(\Omega, \NNReals^{\ell})$ such that $f(\omega)=\mathbb{E}(\ST(F)|\mathcal{G})(\omega)$ for all $\omega\in \Omega$, where $\mathcal{G}$ is the $\sigma$-algebra generated by $\{\ST^{-1}(B): B\in \BorelSets {\Omega}\}$.

\begin{lemma}\label{inconsume}
Suppose that the measure-theoretic production economy $\mathcal{E}$ satisfies \cref{assumptionmseco} and \cref{assumptionalocbound}.
Let $f$ be an element in the Loeb allocation set $\Loeb{\mathscr{A}}$ and $F\in \NSE{\mathcal{L}}^{1}(\mathscr{T}_{\Omega}, \NSE{\NNReals^{\ell}})$ be the S-integrable lifting associated with $f$ specified in \cref{choicesint} in the construction of $\Loeb{\mathscr{E}}$. 
Let $\mathcal{G}=\{\ST^{-1}(B): B\in \BorelSets {\Omega}\}$. 
Then $\bar{f}$ is an element of the allocation set $\cA$, where 
$\bar{f}(\omega)=\ST_{w}(\NSE{E}(F))=\mathbb{E}(\ST(\NSE{E}(F))|\mathcal{G})(\omega)$ for every $\omega\in \Omega$.
\end{lemma}
\begin{proof}
Let $g=\proj{k}\circ f$ be the projection of $f$ to the first $k$ coordinates. Let $\bar{g}=\ST_{w}\big(\NSE{E}(\proj{k}\circ F)\big)=\mathbb{E}\big(\ST(\NSE{E}(\proj{k} \circ F))|\mathcal{G}\big)(\omega)$. Note that $\bar{g}=\proj{k}\circ \bar{f}$. By \cref{projform} of \cref{assumptionalocbound}, it is sufficient to show that $\bar{g}(\omega)\in \proj{k}\big(X(\omega)\big)$ for almost all $\omega\in \Omega$. 

By \cref{hyptsexist}, there exists a $\Loeb{\NSE{\mu}^{\mathscr{T}}}$-measurable set $V\subset \mathscr{T}_{\Omega}$ with $\Loeb{\NSE{\mu}^{\mathscr{T}}}(V)=1$ such that $\NSE{X}(v)\approx X(\ST(v))$ for all $v\in V$.
Hence, we have $\NSE{(\proj{k}\circ X)}(v)\approx \proj{k}\circ X(\ST(v))$ for all $v\in V$. By \cref{badintegrable} of \cref{assumptionalocbound}, the map $\proj{k}\circ X$ maps from $\Omega$ to $\mathcal{K}_{\mathrm{bd}}(\NNReals^{k})$, and is integrably bounded. 
Pick $\epsilon>0$.
By \cref{stlusinset},
there exists a compact set $K_{\epsilon}\subset \Omega$ with $\mu(K_{\epsilon})>1-\epsilon$ and a map $X^{\epsilon}: \Omega\to \cK_{\mathrm{bd}}(\NNReals^{k})$ such that:
\begin{enumerate}[leftmargin=*]
    \item $X^{\epsilon}$ is continuous, hence is upper hemicontinuous as a correspondence;\footnote{For every $\omega\in \Omega$ and every $\omega'\in \NSE{\Omega}$ with $\omega'\approx \omega$, we have $\ST\big(\NSE{X^{\epsilon}}(\omega)\big)=\ST\big(\NSE{X^{\epsilon}}(\omega')\big)$. The result follows from the nonstandard characterization of upper hemicontinuity from \citet{adku1}.}
     \item $X^{\epsilon}(\omega)=\proj{k}\circ X(\omega)$ for all $\omega\in K_{\epsilon}$.
    \item By \cref{exttieze}, the range of $X^{\epsilon}$ is a subset of the convex hull of $\proj{k}\circ X(\omega)$, hence $X^{\epsilon}$ is also integrably bounded.
\end{enumerate}
Define the map $\mathbb{X}^{\epsilon}: \mathscr{T}_{\Omega}\to \cK_{\mathrm{bd}}(\NNReals^{k})$ by letting $\mathbb{X}^{\epsilon}(t)=X^{\epsilon}(\ST(t))$ for $t\in \NS{\mathscr{T}_{\Omega}}$ and $\mathbb{X}^{\epsilon}(t)=\{0\}$ otherwise. Let $V_{\epsilon}=\{v\in V: T(v)\cap \ST^{-1}(K_{\epsilon})\neq \emptyset\}$, where $T(v)$ is the unique element in the hyperfinite partition $\mathscr{T}$ that contains $v$. 
For every $v\in V_{\epsilon}$, we have $X^{\epsilon}(\ST(v))=\proj{k} \circ X(\ST(v))=
\ST\big(\NSE{(\proj{k} \circ X)}(v)\big)=\ST\big(\NSE{(\proj{k}\circ X)}\big)(v)$.
Hence, as $\Loeb{\NSE{\mu}^{\mathscr{T}}}(V_{\epsilon})=\mu(K_{\epsilon})>1-\epsilon$, we have 
$\Loeb{\NSE{\mu}^{\mathscr{T}}}(\{t\in \mathscr{T}_{\Omega}: \mathbb{X}^{\epsilon}(t)=\ST(\NSE{(\proj{k}\circ X)})(t)\})>1-\epsilon$. 
As $X^{\epsilon}$ is integrably bounded,
we can find a Loeb integrable function $f': \mathscr{T}_{\Omega}\to \NNReals^{k}$: 
\begin{enumerate}[leftmargin=*]
    \item $f'(t)\in \mathbb{X}^{\epsilon}(t)$ for all $t\in \mathscr{T}_{\Omega}$;
    \item $f'=g$ on $V_{\epsilon}$, hence    $\Loeb{\NSE{\mu}^{\mathscr{T}}}(\{t\in \mathscr{T}_{\Omega}: f'(t)=g(t)\})>1-\epsilon$.
\end{enumerate}
By \cref{assumptionmseco}, the consumer space $\Omega$ is second countable.
For each $n\in \Nats$, we can construct a countable partition $\mathcal{B}_n$ of $\Omega$ such that:
\begin{enumerate}[leftmargin=*]
    \item $\mathcal{B}_n\subset \BorelSets {\Omega}$, and the diameter of each element in $\mathcal{B}_{n}$ is no greater than $\frac{1}{n}$;
    \item $\mathcal{B}_{n+1}$ is a refinement of $\mathcal{B}_n$;
    \item the $\sigma$-algebra generated by $\bigcup_{n\in \Nats}\mathcal{B}_{n}$ equals $\BorelSets {\Omega}$.
\end{enumerate}
For each $n\in \Nats$, let $\cF_{n}$ be the $\sigma$-algebra generated by $\mathcal{B}_{n}$.
Let $\mathcal{G}_{n}$ be the $\sigma$-algebra generated by $\{\ST^{-1}(A): A\in \cF_{n}\}$, note that $\mathcal{G}_n$ is the same as the $\sigma$-algebra generated by $\{\ST^{-1}(A): A\in \mathcal{B}_{n}\}$. 
Let $F'$ be an S-integrable lifting of $f'$. As $f'(t)\in \mathbb{X}^{\epsilon}(t)$ for all $t\in \mathscr{T}_{\Omega}$, $F'(t)$ is in the monad of $X^{\epsilon}(\ST(t))$ for almost all $t\in \mathscr{T}_{\Omega}$.
Let $\bar{f'}_n(\omega)=\mathbb{E}(\ST(\NSE{E}(F'))|\mathcal{G}_n)(\omega)$ and $\bar{f'}(\omega)=\ST_{w}(\NSE{E}(F'))=\mathbb{E}(\ST(\NSE{E}(F'))|\mathcal{G})(\omega)$ for $\omega\in \Omega$.  
Note that $\NSE{E}(F')\approx \NSE{E}(\proj{k}\circ F)$ on a Loeb measure $1$ subset of $\ST^{-1}(K_{\epsilon})$.
Since $\mu(K_{\epsilon})>1-\epsilon$, we conclude that $\mu(\{\omega: \bar{f'}(\omega)=\bar{g}(\omega)\})>1-2\epsilon$. 
\begin{claim}\label{fnconverge}
$\lim_{n\to \infty}\bar{f'}_n(\omega)=\bar{f'}(\omega)$ for almost all $\omega\in \Omega$.
\end{claim}
\begin{proof}
By the Martingale convergence theorem, $\bar{f'}_n$ converges pointwise to some function $h$ almost surely. 
For every element $A\in \bigcup_{n\in \Nats}\cF_{n}$, we have $\int_{A}h(\omega)\mu(\dee \omega)=\int_{A}\bar{f'}(\omega)\mu(\dee \omega)$. 
As $\bigcup_{n\in \Nats}\cF_{n}$ is a $\pi$-system that generates $\BorelSets {\Omega}$, we have $\int_{B}h(\omega)\mu(\dee \omega)=\int_{B}\bar{f'}(\omega)\mu(\dee \omega)$ for all $B\in \BorelSets {\Omega}$, hence $\lim_{n\to \infty}\bar{f'}_n(\omega)=\bar{f'}(\omega)$ for almost all $\omega\in \Omega$.
\end{proof}
We now show that $\bar{g}(\omega)\in \proj{k}\big(X(\omega)\big)$ for almost all $\omega\in \Omega$. 
Pick $\omega_0\in \Omega$ such that $\lim_{n\to \infty}\bar{f'}_n(\omega_0)=\bar{f'}(\omega_0)$ and let $O$ be an open set that contains $X^{\epsilon}(\omega_0)$ as a subset. 
Note that $X^{\epsilon}(\omega)$ is convex for every $\omega\in \Omega$. 
By the upper hemicontinuity of $X^{\epsilon}(\omega_0)$, 
there exists some $n_0\in \Nats$ such that the closed convex hull of $\bigcup\{X^{\epsilon}(\omega): |\omega-\omega_0|<\frac{1}{n_0}\}$ is contained in $O$. 
By the construction of $\bar{f'}_{n_0}$, we know that $\bar{f'}_{n_0}(\omega_0)$ is in the closed convex hull of $\bigcup\{X^{\epsilon}(\omega): |\omega-\omega_0|<\frac{1}{n_0}\}$, hence is in $O$. 
Thus, $\bar{f'}(\omega_0)$ is in $X^{\epsilon}(\omega_0)$. 
As our choice of $\omega_0$ is arbitrary, we have $\bar{f'}(\omega)\in X^{\epsilon}(\omega)$ for almost all $\omega\in \Omega$. 
As our choice of $\epsilon$ is arbitrary, we have $\bar{g}(\omega)\in \proj{k}\big(X(\omega)\big)$ for almost all $\omega\in \Omega$. Hence, $\bar{f}$ is an element of $\cA$. 
\end{proof}

\begin{lemma}\label{prefctswk}
Suppose the measure-theoretic production economy $\mathcal{E}$ satisfies \cref{assumptionmseco} and \cref{assumptionwkcts}. 
Then, for every $g\in \mathcal{L}^{1}\big((\mathscr{T}_{\Omega}, \Loeb{I(\mathscr{T}_{\Omega})}, \Loeb{\NSE{\mu}^{\mathscr{T}}}),\NNReals^{\ell}\big)$ and every $(y, p)\in Y\times \Delta$, we have $\ST(\NSE{P}_{t}^{\mathscr{T}})(g,y,p)=P_{\ST(t)}(\ST_{w}(\NSE{E}(G)),y,p)$
for $\Loeb{\NSE{\mu}^{\mathscr{T}}}$-almost all $t\in \mathscr{T}_{\Omega}$,
where $G$ is the S-integrable lifting associated with $g$ specified in \cref{choicesint} in the construction of $\Loeb{\mathscr{E}}$.
\end{lemma}
\begin{proof}
Pick $g\in \mathcal{L}^{1}\big((\mathscr{T}_{\Omega}, \Loeb{I(\mathscr{T}_{\Omega})}, \Loeb{\NSE{\mu}^{\mathscr{T}}}),\NNReals^{\ell}\big)$, $y\in Y$ and $p\in \Delta$.
By \cref{hyptsexist},  there exists a $\Loeb{\NSE{\mu}^{\mathscr{T}}}$-measurable set $U\subset \mathscr{T}_{\Omega}$ with $\Loeb{\NSE{\mu}^{\mathscr{T}}}(U)=1$ such that $\NSE{\chi}(u)\approx \chi(\ST(u))$ for all $u\in U$.
For every $t\in U$, we have $\ST(\NSE{P}_{t}^{\mathscr{T}})(g,y,p)=\ST\big(\NSE{P}_{t}^{\mathscr{T}}(G,y,p)\big)
=\ST\big(\NSE{P}_{t}(\NSE{E}(G),y,p)\big)
=\ST\big(\NSE{P}_{\ST(t)}(\NSE{E}(G),y,p)\big)
=P_{\ST(t)}(\ST_{w}(\NSE{E}(G)),y, p)$.
\end{proof}

\begin{theorem}\label{steqexist}
Suppose that the measure theoretic production economy $\mathcal{E}$ satisfies \cref{assumptionmseco}, \cref{assumptionalocbound} and \cref{assumptionwkcts}.
Suppose $P_{\omega}$ takes value in $\mathcal{P}_{H}^{-}$ for almost all $\omega\in \Omega$.
Then, if the Loeb production economy $\Loeb{\mathscr{E}}$ has a Loeb quasi-equilibrium such that the equilibrium prices of the commodities $k+1,\dotsc, \ell$ are positive, then $\mathcal{E}$ has a quasi-equilibrium. 
\end{theorem}
\begin{proof}
Let $(f, \bar{y}, \bar{p})$ be a Loeb quasi-equilibrium for $\Loeb{\mathscr{E}}$ such that every coordinate of $\proj{\ell-k}(\bar{p})$ is positive. 
Let $F$ be the S-integrable lifting associated with $f$ specified in \cref{choicesint} in the construction of $\Loeb{\mathscr{E}}$. Hence, we have $\ST(\NSE{P}_{t}^{\mathscr{T}})(f,y,p)=\ST(\NSE{P}_{t}^{\mathscr{T}}(F,y,p))$
for all $t\in \mathscr{T}_{\Omega}$ and all $(y, p)\in Y\times \Delta$. 
Let $\bar{f}=\ST_{w}(\NSE{E}(F))$ and $\mathcal{G}$ be the $\sigma$-algebra generated by $\{\ST^{-1}(B): B\in \BorelSets {\Omega}\}$.
By construction in \cref{sintns}, we have $\bar{f}(\omega)=\mathbb{E}(\ST(\NSE{E}(F))|\mathcal{G})(\omega)$ for every $\omega\in \Omega$. 
By \cref{inconsume}, $\bar{f}$ is in the standard allocation set $\cA$. 
We shall show that $(\bar{f}, \bar{y}, \bar{p})$ is a quasi-equilibrium for $\mathcal{E}$.  
\begin{claim}\label{instqd}
For almost all $\omega\in \Omega$, $\bar{f}(\omega)\in \bar{D}_{\omega}(\bar{f},\bar{y},\bar{p})$.
\end{claim}
\begin{proof}
For almost all $\omega\in \Omega$, $P_{\omega}$ is a  function from $\mathcal{L}^{1}(\Omega, \NNReals^{\ell})\times Y\times \Delta$ to $\mathcal{P}_{H}^{-}$. As every coordinate of $\proj{\ell-k}(\bar{p})$ is positive, 
by \cref{assumptionalocbound} and \cref{thm_convexdemandsupply}, the quasi-demand set $\bar{D}_{\omega}(\bar{f}, \bar{y}, \bar{p})$ is a measurable map on $\Omega$, taking values almost surely in $\cK_{\mathrm{bd}}(\NNReals^{\ell})$.\footnote{\label{intselectorfoot} By the same argument in \cref{monotonesint}, the projection $\proj{(\ell-k)}\big(\bar{D}_{\omega}(\bar{f}, \bar{y}, \bar{p})\big)$ of $\bar{D}_{\omega}(\bar{f}, \bar{y}, \bar{p})$ onto coordinates $\{k+1,\dotsc,\ell\}$ is bounded for all $\omega$.}
Pick $\epsilon>0$. By the same proof as in \cref{inconsume}, there exists a compact set $K_{\epsilon}\subset \Omega$ with $\mu(K_{\epsilon})>1-\epsilon$ and a map
$\bar{D}_{\omega}^{\epsilon}(\bar{f}, \bar{y}, \bar{p}): \Omega\to \cK_{\mathrm{bd}}(\NNReals^{\ell})$ such that:
\begin{enumerate}[leftmargin=*]
    \item $\bar{D}_{\omega}^{\epsilon}(\bar{f}, \bar{y}, \bar{p})$ is continuous, hence is upper hemicontinuous as a correspondence;
    \item $\bar{D}_{\omega}^{\epsilon}(\bar{f}, \bar{y}, \bar{p})$ is integrably bounded;
    \item $\bar{D}_{\omega}^{\epsilon}(\bar{f}, \bar{y}, \bar{p})=\bar{D}_{\omega}(\bar{f}, \bar{y}, \bar{p})$ for all $\omega\in K_{\epsilon}$.
\end{enumerate}
Define the correspondence $\bar{\mathbb{D}}_{t}^{\epsilon}(f,\bar{y},\bar{p})$ by letting $\bar{\mathbb{D}}_{t}^{\epsilon}(f,\bar{y},\bar{p})=\bar{D}_{\ST(t)}^{\epsilon}(\bar{f}, \bar{y}, \bar{p})$ for $t\in \NS{\mathscr{T}_{\Omega}}$ and 
$\bar{\mathbb{D}}_{t}^{\epsilon}(f,\bar{y},\bar{p})=\{0\}$ otherwise. 
Note that $f(t)\in \bar{\mathbb{D}}_{t}(f,\bar{y},\bar{p})$\footnote{Recall that $\bar{\mathbb{D}}_{t}(f,\bar{y},\bar{p})$ is the Loeb quasi-demand set at $(f,\bar{y},\bar{p})$.} for $\Loeb{\NSE{\mu}^{\mathscr{T}}}$-almost all $t\in \mathscr{T}_{\Omega}$. 
By \cref{hyptsexist}, \cref{prefctswk} and \cref{qsdemand_standard_part}, there exists a $\Loeb{\NSE{\mu}^{\mathscr{T}}}$-measurable set $U\subset \mathscr{T}_{\Omega}$ with $\Loeb{\NSE{\mu}^{\mathscr{T}}}(U)=1$ such that $\bar{\mathbb{D}}_{t}(f,\bar{y},\bar{p})=\bar{D}_{\ST(t)}(\bar{f},\bar{y},\bar{p})$ for all $t\in U$. 
Using a similar argument as in \cref{inconsume}, we can find a Loeb integrable function
$f': \mathscr{T}_{\Omega}\to \NNReals^{\ell}$ such that: 
\begin{enumerate}[leftmargin=*]
    \item $f'(t)\in \bar{\mathbb{D}}_{t}^{\epsilon}(f, \bar{y}, \bar{p})$ for all $t\in \mathscr{T}_{\Omega}$;
    \item $\Loeb{\NSE{\mu}^{\mathscr{T}}}(\{t\in T: f'(t)=f(t)\})>1-\epsilon$. 
\end{enumerate}

For each $n\in \Nats$, we can construct a countable partition $\mathcal{B}_n$ of $\Omega$ such that:
\begin{enumerate}[leftmargin=*]
    \item $\mathcal{B}_n\subset \BorelSets {\Omega}$, and the diameter of each element in $\mathcal{B}_{n}$ is no greater than $\frac{1}{n}$;
    \item $\mathcal{B}_{n+1}$ is a refinement of $\mathcal{B}_n$, and the $\sigma$-algebra generated by $\bigcup_{n\in \Nats}\mathcal{B}_{n}$ is $\BorelSets {\Omega}$.
\end{enumerate}
Let $\mathcal{G}_{n}$ be the $\sigma$-algebra generated by $\{\ST^{-1}(A): A\in \mathcal{B}_{n}\}$. 
Let $F'$ be the S-integrable lifting associated with $f'$ specified in \cref{choicesint} in the construction of $\Loeb{\mathscr{E}}$. 
Then $F'(t)$ is in the monad of $\bar{D}^{\epsilon}_{\ST(t)}(\bar{f},\bar{y},\bar{p})$ for almost all $t \in \mathscr{T}_{\Omega}$.
Note that we also have
$\ST(\NSE{P}_{t}^{\mathscr{T}})(f',y, p)=\ST(\NSE{P}_{t}^{\mathscr{T}}(F',y, p))$
for all $t\in \mathscr{T}_{\Omega}$ and all $(y,p)\in Y\times \Delta$. 
Let $\bar{f}'_n(\omega)=\mathbb{E}(\ST(\NSE{E}(F'))|\mathcal{G}_n)(\omega)$ and $\bar{f}'(\omega)=\mathbb{E}(\ST(\NSE{E}(F'))|\mathcal{G})(\omega)$ for all $\omega\in \Omega$.  
By the same argument as in \cref{inconsume},
we have $\mu(\{\omega: \bar{f}(\omega)=\bar{f}'(\omega)\})>1-2\epsilon$. 
By the same argument as in \cref{fnconverge}, we have $\lim_{n\to \infty}\bar{f}'_n(\omega)=\bar{f}'(\omega)$ for almost all $\omega\in \Omega$.

Pick $\omega_0\in \Omega$ such that $\lim_{n\to \infty}\bar{f}'_n(\omega_0)=\bar{f}'(\omega_0)$ and let $O$ be an open set such that $\bar{D}_{\omega_0}^{\epsilon}(\bar{f}, \bar{y}, \bar{p})\subset O$. 
As $\bar{D}_{\omega_0}^{\epsilon}(\bar{f}, \bar{y}, \bar{p})$ is convex, by the upper hemicontinuity of $\bar{D}_{\omega}^{\epsilon}(\bar{f}, \bar{y}, \bar{p})$, 
there is some $n_0\in \Nats$ such that the closed convex hull of $\bigcup\{\bar{D}_{\omega}^{\epsilon}(\bar{f}, \bar{y}, \bar{p}): |\omega-\omega_0|<\frac{1}{n_0}\}$ is contained in $O$. 
By construction, $\bar{f}'_{n_0}(\omega_0)$ is in the closed convex hull of 
$\bigcup\{\bar{D}_{\omega}^{\epsilon}(\bar{f}, \bar{y}, \bar{p}): |\omega-\omega_0|<\frac{1}{n_0}\}$, hence is in $O$.
Thus, $\bar{f}'(\omega_0)$ is in $\bar{D}_{\omega_0}^{\epsilon}(\bar{f},\bar{y}, \bar{p})$. 
As our choice of $\omega_0$ is arbitrary, we have $\bar{f}'(\omega)\in \bar{D}_{\omega}^{\epsilon}(\bar{f},\bar{y},\bar{p})$ for almost all $\omega\in \Omega$. 
As our choice of $\epsilon$ is arbitrary, we have $\bar{f}(\omega)\in \bar{D}_{\omega}(\bar{f}, \bar{y}, \bar{p})$ for almost all $\omega\in \Omega$, completing the proof.  
\end{proof}
As the Loeb supply set is the same as the supply set, we have $\bar{y}(j)\in S_{j}(\bar{p})$. We now show that market clears at the candidate quasi-equilibrium. Note that:
\[
\int_{\Omega}\bar{f}(\omega)\mu(\dee \omega)&=\int_{\NS{\NSE{\Omega}}}\mathbb{E}(\ST(\NSE{E}(F))|\mathcal{G})(\omega)\Loeb{\NSE{\mu}}(\dee \omega)
=\int_{\NS{\NSE{\Omega}}}\ST(\NSE{E}(F))(\omega)\Loeb{\NSE{\mu}}(\dee \omega)\nonumber \\
&=\int_{\NSE{\Omega}}\ST(\NSE{E}(F))(\omega)\Loeb{\NSE{\mu}}(\dee \omega)
%&\approx \int_{\NSE{\Omega}}\NSE{E}(F)(\omega)\NSE{\mu}(\dee \omega)\nonumber \\
\approx \int_{\mathscr{T}_{\Omega}}F(t)\NSE{\mu}^{\mathscr{T}}(\dee t)\approx \int_{\mathscr{T}_{\Omega}}f(t)\Loeb{\NSE{\mu}^{\mathscr{T}}}(\dee t).\nonumber
\]
We also have $\int_{\Omega}e(\omega)\mu(\dee \omega)=\int_{\mathscr{T}_{\Omega}}\ST(\NSE{e})(t)\Loeb{\NSE{\mu}^{\mathscr{T}}}(\dee t)=\int_{\mathscr{T}_{\Omega}}\ST(\hat{e})(t)\Loeb{\NSE{\mu}^{\mathscr{T}}}(\dee t)$.
As $(f, \bar{y}, \bar{p})$ is a Loeb quasi-equilibrium, we have 
\[
&\int_{\Omega}\bar{f}(\omega)\mu(\dee \omega)-\int_{\Omega}e(\omega)\mu(\dee \omega)-\sum_{j\in J}\bar{y}(j)\nonumber \\
&=\int_{\mathscr{T}_{\Omega}}f(t)\Loeb{\NSE{\mu}^{\mathscr{T}}}(\dee t)-\int_{\mathscr{T}_{\Omega}}\ST(\hat{e})(t)\Loeb{\NSE{\mu}^{\mathscr{T}}}(\dee t)-\sum_{j\in J}\bar{y}(j)=0.\nonumber
\]
Combining with \cref{instqd} and the fact that $\bar{y}(j)\in S_{j}(\bar{p})$ for all $j\in J$, $(\bar{f},\bar{y},\bar{p})$ is a quasi-equilibrium for the original measure-theoretic production economy $\mathcal{E}$. 
\end{proof}

Finally, we show that, under \cref{assumptionsurvival}, every quasi-equilibrium is an equilibrium. 

\begin{lemma}\label{interioritylm}
Let $\mathcal{E}=\{(X, \succ_{\omega},P_{\omega}, e_{\omega}, \theta_{\omega})_{\omega\in \Omega}, (Y_j)_{j\in J}, (\Omega, \cB,\mu)\}$ be a measure-theoretic production economy satisfying \cref{assumptionsurvival}, and $(x,y,p)$ be a quasi-equilibrium. Then $(x,y,p)$ is an equilibrium.
\end{lemma} 

\begin{proof}
Let $(\bar{x},\bar{y},\bar{p})$ be a quasi-equilibrium.
For each consumer $\omega$, define a correspondence  $\delta_\omega: \Delta\cto X_{\omega}$ as
\[
\delta_{\omega}(p)=\{x_{\omega}\in X_{\omega}: p\cdot x_{\omega} < p\cdot e(\omega) + \sum_{j\in J}\theta_{\omega j}\sup\{p\cdot y: y\in Y_j\}\}. \nonumber
\]
We start by establishing the following claim:
\begin{claim}\label{delta0claim}
For every $\omega\in \Omega_0$, $\delta_{\omega}(\bar{p})\neq \emptyset$.
\end{claim}
\begin{proof}
Note that, for every $\omega\in \Omega_0$, the set $X_{\omega} - \sum_{j\in J}\theta_{\omega j}Y_j$ has non-empty interior $U_{\omega}$ and $e(\omega)\in U_{\omega}$. Hence, we can pick $u_{\omega}\in \Reals^{\ell}$ such that $\bar{p}\cdot u_{\omega}<0$ and that $(e(\omega) + u_{\omega}) \in (X_{\omega}-\sum_{j\in J}\theta_{\omega j}Y_j)$. As $(\bar{x},\bar{y},\bar{p})$ is a quasi-equilibrium, we have
$\bar{p}\cdot \tilde x_{\omega} < \bar{p}\cdot e(\omega)+\sum_{j\in J}\theta_{\omega j}\bar{p}\cdot \bar{y}(j)$
for some $\tilde x_{\omega}\in X_{\omega}$.
So we have $\delta_{\omega}(\bar{p})\neq\emptyset$.
\end{proof}
\cref{delta0claim} leads to the following result:
\begin{claim}\label{omega0equil}
For almost all $\omega\in \Omega_0$,
if $\hat{x}\in X_{\omega}$ with $(\hat{x}, \bar{x}(\omega))\in P_{\omega}(\bar{x},\bar{y},\bar{p})$, then $\bar{p}\cdot \hat{x}>\bar{p}\cdot e(\omega)+\sum_{j\in J}\theta_{\omega j}\bar{p}\cdot \bar{y}(j)$. 
\end{claim}
\begin{proof}
Let $\Omega'_{0}\subset \Omega_0$ be the set of consumers such that their quasi-equilibrium consumption is in their quasi-demand set. Note that $\mu(\Omega'_{0})=\mu(\Omega_0)$. 
Fix some $\omega\in \Omega'_0$.
Let $\hat{x}\in X_{\omega}$ be such that 
$(\hat{x}, \bar{x}(\omega))\in P_{\omega}(\bar{x},\bar{y},\bar{p})$. By \cref{delta0claim}, pick $z_{\omega}\in \delta_{\omega}(\bar{p})$. 
Thus, we have $\bar{p}\cdot \hat{x}\geq \bar{p}\cdot e(\omega)+\sum_{j\in J}\theta_{\omega j}\bar{p}\cdot \bar{y}(j)$ since $(\bar{x},\bar{y},\bar{p})$ is a quasi-equilibrium. As $P_{\omega}(\bar{x},\bar{y},\bar{p})$ is continuous, there exists $\lambda\in (0,1)$ such that 
$(\lambda z_{\omega}+(1-\lambda)\hat{x}, \bar{x}(\omega)) \in P_{\omega}(\bar{x},\bar{y},\bar{p})$.

Assume that $\bar{p}\cdot \hat{x}=\bar{p}\cdot e(\omega)+\sum_{j\in J}\theta_{\omega j}\bar{p}\cdot \bar{y}(j)$. 
Then we have $(\lambda z_{\omega}+(1-\lambda)\hat{x}, \bar{x}(\omega)) \in P_{\omega}(\bar{x},\bar{y},\bar{p})$ and $\lambda z_{\omega}+(1-\lambda)\hat{x}\in \delta_{\omega}(\bar{p})$. 
This furnishes us a contradiction since $(\bar{x},\bar{y},\bar{p})$ is a quasi-equilibrium. So we have $\bar{p}\cdot \hat{x}>\bar{p}\cdot e(\omega)+\sum_{j\in J}\theta_{\omega j}\bar{p}\cdot \bar{y}(j)$.  
\end{proof}
As each consumer $\omega\in \Omega_0$ has a strongly monotone preference on the commodity $s$ and the projection $\pi_{s}(X_{\omega})$ is unbounded, by \cref{omega0equil}, we conclude that $\bar{p}_{s}>0$. 

\begin{claim}\label{deltaclaim}
For almost all $\omega\in \Omega$, $\delta_{\omega}(\bar{p})\neq \emptyset$. 
\end{claim}
\begin{proof}
Note that, for almost all $\omega\in \Omega$, there is an open set $V_{\omega}$ containing the $s$-th coordinate $e(\omega)_{s}$ of $e(\omega)$ such that $(e(\omega)_{-s},v)\in X_{\omega} - \sum_{j\in J}\theta_{\omega j}Y_j$ for all $v\in V_{\omega}$.
As $\bar{p}_{s}>0$, for almost all $\omega\in \Omega$, we can pick $u_{\omega}\in \Reals^{\ell}$ such that $\bar{p}\cdot u_{\omega}<0$ and that $(e(\omega) + u_{\omega}) \in (X_{\omega}-\sum_{j\in J}\theta_{\omega j}Y_j)$.
Thus, for almost all $\omega\in \Omega$, we have
\[
\bar{p}\cdot \tilde x_{\omega} < \bar{p}\cdot e(\omega)+\sum_{j\in J}\theta_{\omega j}\bar{p}\cdot \bar{y}(j) \nonumber
\]
for some $\tilde x_\omega\in X_{\omega}$.
So we have $\delta_{\omega}(\bar{p})\neq\emptyset$ for almost all $\omega\in \Omega$.  
\end{proof}

We now show that $(\bar{x},\bar{y},\bar{p})$ is an equilibrium. The proof is similar to the proof of \cref{omega0equil}. 
For almost all $\omega\in \Omega$, by \cref{deltaclaim},
pick $z_{\omega}\in \delta_{\omega}(\bar{p})$ and $\hat{x}_{\omega}\in X_{\omega}$ such that $(\hat{x}_{\omega}, \bar{x}(\omega))\in P_{\omega}(\bar{x},\bar{y},\bar{p})$. 
Hence, 
we have $\bar{p}\cdot \hat{x}_{\omega}\geq \bar{p}\cdot e(\omega)+\sum_{j\in J}\theta_{\omega j}\bar{p}\cdot \bar{y}(j)$ since $(\bar{x},\bar{y},\bar{p})$ is a quasi-equilibrium. As $P_{\omega}(\bar{x},\bar{y},\bar{p})$ is continuous, there exists $\lambda\in (0,1)$ such that 
$(\lambda z_\omega+(1-\lambda)\hat x_\omega, \bar{x}(\omega)) \in P_{\omega}(\bar{x},\bar{y},\bar{p})$.
Assume that $\bar{p}\cdot \hat{x}_{\omega}=\bar{p}\cdot e(\omega)+\sum_{j\in J}\theta_{\omega j}\bar{p}\cdot \bar{y}(j)$. 
Then we have $(\lambda z_\omega+(1-\lambda)\hat x_\omega, \bar{x}(\omega)) \in P_{\omega}(\bar{x},\bar{y},\bar{p})$ and $\lambda z_\omega+(1-\lambda)\hat x_\omega\in \delta_{\omega}(\bar{p})$. 
This furnishes us a contradiction since $(\bar{x},\bar{y},\bar{p})$ is a quasi-equilibrium. Therefore, we have $\bar{p}\cdot \hat x_\omega>\bar{p}\cdot e(\omega)+\sum_{j\in J}\theta_{\omega j}\bar{p}\cdot \bar{y}(j)$. 
Hence, $(\bar{x},\bar{y},\bar{p})$ is an equilibrium. 
\end{proof}

We are now at the place to prove our main result, \cref{standardmain}.

\begin{proof}[Proof of \cref{standardmain}]
As we have pointed out in \cref{suffsintcond}, \cref{alocboundlemma} follows from the condition in \cref{suffsintcond}, which is implied by \cref{assumptionwkcts}. 
As stated in the proof of \cref{monotonesint}, all equilibrium prices of the commodities $k+1,\dotsc,\ell$ are positive. 
\cref{standardmain} follows from \cref{alocboundlemma}, \cref{steqexist} and \cref{interioritylm}. 
\end{proof}

{\printbibliography}

\pagebreak

\section{Supplementary Material - Existence of Equilibrium in Measure-theoretic Quota Economies}\label{sectionquota}

Both \cref{finitemainexpd} and \cref{standardmain} consider non-free-disposal equilibrium, which requires that demand exactly equals supply for each commodity. In this Supplement, we incorporate the quota regulatory scheme, developed in \citet{ad25}, into measure-theoretic production economies. Doing so allows one to limit the total amount of bads disposed to a prespecified positive level.

\begin{definition}\label{defmeasureqteco}
A measure-theoretic quota economy  
\[
\mathcal E\equiv \{(X, \succ_{\omega}, P_{\omega}, e_{\omega}, \theta)_{\omega\in \Omega}, (Y_j)_{j\in J}, (\Omega, \cB, \mu), (m^{(j)})_{j\in J}, \mathcal{Z}(m)\}\nonumber
\] 
is a list such that: 
\begin{enumerate}[(i), leftmargin=*]
\item $(X, \succ_{\omega}, P_{\omega}, e_{\omega}, \theta)_{\omega\in \Omega}$ and $(\Omega, \cB, \mu)$ are defined the same as in \cref{defmeasureeco};
\item As in \cref{defmeasureeco}, $J$ is a finite set of firms. However, firms are categorized into two types: private firms and a single government firm. The government firm, denoted as firm $0$, has the production set $\{0\}$. For each private firm $j\in J$, its production set $Y_{j}\subset \Reals^{\ell}$ is a non-empty subset. We write $Y=\prod_{j\in J}Y_{j}$; 
\item The government chooses to regulate the first $t\leq \ell$ commodities and assigns quotas on regulated commodities to the firms. For each $j\in J$, define $m^{(j)}\in \NegReals^{t}$ to be the negative of the quota for the firm $j$. Let $m=\sum_{j\in J}m^{(j)}$. 
The \emph{quota-compliance region}
$\mathcal{Z}(m)=\{m\}\times \{0\}^{\ell-t}$ is a convex subset of $\NegReals^{\ell}$. 
\end{enumerate}
\end{definition}

\cref{defmeasureqteco} is the measure-theoretic version of the quota equilibrium model in \citet{ad25}. We note that the set of regulated commodities need not be the same as the set of bads in \cref{assumptionalocbound}, since the society may choose to tolerate certain bads.

For every $\omega\in \Omega$, $p\in \Delta$ and $y\in Y$, the \emph{quota budget set} $B_{\omega}^{m}(y, p)$ is defined as
\[
\{z\in X_{\omega}: p\cdot z \leq p\cdot e(\omega) + \sum_{j\in J}\theta_{\omega j}\big(p\cdot y(j)+\proj{t}(p)\cdot m^{(j)}\big)\}.\nonumber
\]
For each private firm $j\in J$, since the firm can emit the first $t$ commodities freely up to its quota $m^{(j)}$, the firm's profit at a given price $p$ is $p\cdot y(j)+\proj{t}(p)\cdot m^{(j)}$.\footnote{If a private firm emits less than its quota, then the firm generates additional revenue by selling its remaining quota to other firms. If a private firm emits more than its quota, then the firm needs to purchase quota from other firms.} The government firm's profit comes solely from selling its quota. In particular, the government firm's profit at a given price $p$ is $p\cdot y(0)+\proj{t}(p)\cdot m^{(0)}=\proj{t}(p)\cdot m^{(0)}$. Hence, the consumer's budget consists of the value of her endowment and dividend from firms. 
For $\omega\in \Omega$ and $(x,y,p)\in \mathcal{L}^{1}(\Omega, \NNReals^{\ell})\times Y\times \Delta$, the \emph{quota demand set} $D^{m}_{\omega}(x,y,p)$ is
\[
\{z \in B^{m}_{\omega}(y,p): w \succ_{x,y,\omega,p} z\implies w\not\in B^{m}_{\omega}(y,p)\}.\nonumber
\]
Given a price $p$, the firm $j$'s supply set $S_{j}^{m}(p)$ is $\argmax_{z\in Y_{j}}\big(p\cdot z+\proj{t}(p)\cdot m^{(j)}\big)$. As $\proj{t}(p)\cdot m^{(j)}$ does not depend on the firm's production plan, $S_{j}^{m}(p)=\argmax_{z\in Y_{j}}p\cdot z$. 
All firms' profits depend only on prices and their own production.

\begin{definition}\label{def_mameqprodqt}
Let $\mathcal{E}=\{(X, \succ_{\omega}, P_\omega, e_\omega, \theta)_{\omega\in \Omega}, (Y_j)_{j\in J}, (\Omega, \cB, \mu), (m^{(j)})_{j\in J}, \mathcal{Z}(m)\}$ be a measure-theoretic quota economy. 
A $\mathcal{Z}(m)$-compliant quota equilibrium is $(\bar{x}, \bar{y}, \bar{p})\in \mathcal{A}\times Y\times \Delta$ such that the following conditions are satisfied:
\begin{enumerate}[(i), leftmargin=*]
\item $\bar{x}(\omega)\in D^{m}_{\omega}(\bar{x},\bar{y},\bar{p})$ for almost all $\omega\in \Omega$;
\item $\bar{y}(j)\in S^{m}_{j}(\bar{p})$ for all $j\in J$. Every firm is profit maximizing given the price $\bar{p}$;
\item $\int_{\Omega}\bar{x}(\omega)\mu(\dee \omega)-\int_{\Omega}e(\omega)\mu(\dee \omega)-\sum_{j\in J}\bar{y}(j)\in \mathcal{Z}(m)$.
\end{enumerate} 
\end{definition}
The quota-compliance region $\mathcal{Z}(m)$ and the feasibility constraint $\sum_{\omega\in \Omega}\bar{x}(\omega)-\sum_{\omega\in \Omega}e(\omega)-\sum_{j\in J} \bar{y}(j)\in \mathcal{Z}(m)$ jointly imply that, at equilibrium, the total net emission of the regulated commodities equals the pre-specified total quota, which is the aggregation of the government firm's quota and private firms' quota. 
The set of quota-compliant consumption-production pair of $\mathcal{E}$ is
\[
\cO^{m}=\left\{(x,y)\in \mathcal{L}^{1}(\Omega, \NNReals^{\ell})\times Y: \int_{\Omega}x(\omega) \mu(\dee \omega)-\int_{\Omega}e(\omega)\mu(\dee \omega)-  \sum_{j\in J}y(j)\in \mathcal{Z}(m)\right\}.\nonumber
\]
For $\epsilon>0$, let $\mathcal{Z}(m)_{\epsilon}$ be the $\epsilon$-neighborhood of $\{m\}\times \{0\}^{\ell-t}$.    
The set of \emph{$\epsilon$-quota-compliant consumption-production pair} for the measure-theoretic quota economy $\mathcal{E}$ is
\[
\cO^{m}_{\epsilon}=\left\{(x,y)\in \mathcal{L}^{1}(\Omega, \NNReals^{\ell})\times Y: \int_{\Omega} x(\omega)\mu(\dee \omega)-\int_{\Omega}e(\omega)\mu(\dee \omega)-\sum_{j\in J}{y}(j)\in \mathcal{Z}(m)_{\epsilon}\right\}.\nonumber
\]
Our main result of this section establishes the existence of a quota equilibrium for measure-theoretic quota economies: 

\begin{customthm}{3}\label{standardmainqt}
Let $\mathcal{E}$ be a measure-theoretic quota economy as in \cref{defmeasureqteco}. Suppose $\mathcal{E}$
satisfies \cref{assumptionmseco}, \cref{assumptionalocbound},\footnote{The set $\cO_{\epsilon_{s}}$ in \cref{goodmonotone} of \cref{assumptionalocbound} needs to be replaced by $\cO^{m}_{\epsilon_{s}}$.} \cref{assumptionwkcts}, and the following conditions:
\begin{enumerate}[(i), leftmargin=*]
\item for almost all $\omega\in \Omega$, $P_\omega$ takes value in  $\mathcal{P}_{H}^{-}$;
\item\label{qtmeresto0} there exists $\Omega_0\subset \Omega$ of positive measure such that, for every $\omega\in \Omega_0$, the set $X_{\omega} - \sum_{j\in J}\theta_{\omega j}(Y_j+\{E(m^{(j)})\})$\footnote{For all $j\in J$, $E(m^{(j)})\in \NegReals^{\ell}$ is the vector such that its projection to the first $t$-th coordinates is $m^{(j)}$ and its other coordinates are $0$.} has non-empty interior $U_{\omega}\subset \Reals^{\ell}$ and $e(\omega)\in U_{\omega}$;
\item \label{qtmercdexplain} there exists 
a commodity $s\in \{1,2,\dotsc,\ell\}$ such that:
\begin{itemize}[leftmargin=*]
    \item for every $\omega\in \Omega_0$, the projection $\pi_{s}(X_{\omega})$ is unbounded, and the consumer $\omega$ has a strongly monotone preference on the commodity $s$;
    \item for almost all $\omega\in \Omega$, there is an open set $V_{\omega}$ containing the $s$-th coordinate $e(\omega)_{s}$ of $e(\omega)$ such that $(e(\omega)_{-s},v)\in X_{\omega} - \sum_{j\in J}\theta_{\omega j}(Y_j+\{E(m^{(j)})\})$ for all $v\in V_{\omega}$.
\end{itemize}
\item\label{qtmersatiate} for some $\epsilon>0$, for almost all $\omega\in \Omega$ and all $(x,y)\in \cO^{m}_{\epsilon}$ such that $x(\omega)\in X_{\omega}$, there exists $u\in X_{\omega}$ such that $(u, x(\omega))\in \bigcap_{p\in \Delta}P_{\omega}(x,y,p)$;

\item The aggregate production set $\bar{Y}$ is closed and convex, $\bar Y \cap (-\bar Y)=\bar Y \cap \NNReals^{\ell}=\{0\}$, and $Y_j$ is closed for all $j\in J$.
\end{enumerate}
Then, $\cO^{m}$ is non-empty, i.e., it is feasible to achieve the quota, and $\mathcal{E}$ has a quota equilibrium. 
\end{customthm}

Since firms obtain profit from the property right of pre-assigned quota, the relevant production set for firm $j$ is $Y_{j}+E(m^{(j)})$. Thus, \cref{qtmeresto0} and \cref{qtmercdexplain} of \cref{standardmainqt} are similar, and play the same role as our survival assumption \cref{assumptionsurvival}. 
The proof of \cref{standardmainqt} follows from \cref{standardmain} and shifting the production set of each firm by its pre-assigned quota. The assumptions on production sets in \cref{standardmain} are translation-invariant. 

\begin{proof}[Proof of \cref{standardmainqt}]
By \cref{qtmeresto0} and the second bullet of \cref{qtmercdexplain} in the assumptions of \cref{standardmainqt}, we have $e(\omega)\in X_{\omega}-\sum_{j\in J}\theta_{\omega j}(Y_j+\{E(m^{(j)})\})$ for almost all $\omega\in \Omega$. So the set $\cO^{m}$ of quota-compliant consumption-production pairs is non-empty, hence is feasible to achieve the quota. 
Let $\mathcal{E}'=\{(X, \succ'_{\omega}, P'_{\omega}, e_{\omega}, \theta)_{\omega\in \Omega}, (Y'_j)_{j\in J}, (\Omega, \cB, \mu)\}$ be a measure-theoretic production economy with quota: 
\begin{enumerate}[leftmargin=*]
    \item $Y'_{j}=Y_{j}+\{E(m^{(j)})\}$ for all $j\in J$. Let $Y'=\prod_{j\in J'}Y'_{j}$; 
    \item We only provide a rigorous definition of the induced preference map $P'_{\omega}$ \footnote{The consumer's global preference relation $\succ'_{\omega}$ is defined similarly. To establish the existence of a quota equilibrium, one only needs to work with the preference map $P'_{\omega}$.}. For $y\in Y'$, let $y(\mathcal{E})\in Y$ be such that $y(\mathcal{E})_{j}=y_{j}-E(m^{(j)})$ for all $j\in J$. 
    For $\omega\in \Omega$, the preference map $P'_{\omega}: \mathcal{L}^{1}(\Omega, \NNReals^{\ell})\times Y'\times \Delta\to \mathcal{P}$ is given by
    \[
    P'_{\omega}(x,y,p)=(X_{\omega},\{(a, b)\in X_{\omega}\times X_{\omega} |(x,y(\mathcal{E}),p,a)\succ_{\omega}(x,y(\mathcal{E}),p,b)\})=P_{\omega}(x,y(\mathcal{E}),p).\nonumber
    \]  
\end{enumerate}
To show that the derived economy $\mathcal{E}'$ has an equilibrium, we must verify that $\mathcal{E}'$ satisfies the assumptions of \cref{standardmain}. It is easy to see that: 
\begin{enumerate}[leftmargin=*]
    \item  \cref{assumptionmseco}, \cref{assumptionalocbound} and \cref{assumptionwkcts} are satisfied;
    \item By the construction of $P'_{\omega}$, $P'_{\omega}$ takes value in $\mathcal{P}_{H}^{-}$ for almost all $\omega\in \Omega$;
    \item there exists $\Omega_0\subset \Omega$ of positive measure such that, for every $\omega\in \Omega_0$, the set $X_{\omega} - \sum_{j\in J}\theta_{\omega j}Y'_j$ has non-empty interior $U_{\omega}\subset \Reals^{\ell}$ and $e(\omega)\in U_{\omega}$;
\item there exists 
a commodity $s\in \{1,2,\dotsc,\ell\}$ such that:
\begin{itemize}
    \item for every $\omega\in \Omega_0$, the projection $\pi_{s}(X_{\omega})$ is unbounded, and the consumer $\omega$ has a strongly monotone preference on the commodity $s$;
    \item for almost all $\omega\in \Omega$, there is an open set $V_{\omega}$ containing the $s$-th coordinate $e(\omega)_{s}$ of $e(\omega)$ such that $(e(\omega)_{-s},v)\in X_{\omega} - \sum_{j\in J}\theta_{\omega j}Y'_j$ for all $v\in V_{\omega}$;
\end{itemize}  
\item $\bar{Y'}$ is closed and convex, and $Y'_{j}$ is closed for all $j\in J$.
\end{enumerate}

Let $\cO'$ be the set of attainable consumption-production pairs for $\mathcal{E}'$. For $\epsilon>0$, let $\cO'_{\epsilon}$ be the set of $\epsilon$-attainable consumption-production pairs for $\mathcal{E}'$. 
\begin{claim}\label{eattainmod}
For some $\epsilon>0$, almost all $\omega\in \Omega$ and all $(x, y)\in \cO'_{\epsilon}$ such that $x(\omega)\in X_{\omega}$, there exists $u\in X_{\omega}$ such that $(u, x(\omega))\in \bigcap_{p\in \Delta}P'_{\omega}(x,y,p)$.
\end{claim}
\begin{proof}
Pick the same $\epsilon$ as in \cref{qtmersatiate} of \cref{standardmainqt}. For almost all $\omega\in \Omega$ and all $(x, y)\in \cO'_{\epsilon}$ such that $x(\omega)\in X_{\omega}$, we have $(x, y(\mathcal{E}))\in \cO^{m}_{\epsilon}$ and $x(\omega)\in X_{\omega}$, hence there exists $u\in X_{\omega}$ such that $(u, x(\omega))\in \bigcap_{p\in \Delta}P_{\omega}(x,y(\mathcal{E}),p)$. As $P'_{\omega}(x,y,p)=P_{\omega}(x,y(\mathcal{E}),p)$ for all $p\in \Delta$, we have $(u, x(\omega))\in \bigcap_{p\in \Delta}P'_{\omega}(x,y,p)$.
\end{proof}
Recall that, for any set $B\subset \Reals^{\ell}$, $\rc(B)$ denote the recession cone of $B$. By the proof of \cref{thm_attainable}, it is sufficient to show that $\rc(\bar{X})\cap (-\rc(\bar{X}))=\{0\}$, $\rc(\bar{X})\cap \rc(\bar{Y'})=\{0\}$ and $\bar{Y'}\cap (-\bar{Y'})=\{0\}$. As $\rc(\bar{X})\subset \NNReals^{\ell}$, we have $\rc(\bar{X})\cap (-\rc(\bar{X}))=\{0\}$. Note that $\bar{Y'}=\bar{Y}+\{E(m)\}$. As $\rc(\bar{X})\cap \rc(\bar{Y})=\{0\}$ and $\bar{Y}\cap (-\bar{Y'})=\{0\}$, we have $\rc(\bar{X})\cap \rc(\bar{Y'})=\{0\}$ and $\bar{Y'}\cap (-\bar{Y'})=\{0\}$.

By \cref{standardmain},
there is an equilibrium $(\bar{x}, \bar{y}, \bar{p})$ for $\mathcal{E}'$.
We now show that $(\bar{x}, \bar{y}(\mathcal{E}), \bar{p})$ is a quota equilibrium for $\mathcal{E}$:
\begin{enumerate}[leftmargin=*]
    \item Note that we have $\bar{p}\cdot \bar{y}(j)=\bar{p}\cdot \bar{y}(\mathcal{E})(j)+\proj{k}(\bar{p})\cdot m^{(j)}$. For every $j\in J$, we have $\bar{y}(j)\in \argmax_{z\in Y'_{j}}\bar{p}\cdot z$. As $\proj{k}(\bar{p})\cdot m^{(j)}$ is a constant over $Y_{j}$, we have $\bar{y}(\mathcal{E})(j)\in S^{m}_{j}(\bar{p})$ for all $j\in J$;
    \item As $\int_{\Omega}\bar{x}(\omega)\mu(\dee \omega)-\int_{\Omega}e(\omega)\mu(\dee \omega)-\sum_{j\in J} \bar{y}(j)=0$, we have 
    \[
    &\int_{\Omega}\bar{x}(\omega)\mu(\dee \omega)-\int_{\Omega}e(\omega)\mu(\dee \omega)-\sum_{j\in J} \bar{y}(\mathcal{E})(j) \nonumber\\ 
    &=\int_{\Omega}\bar{x}(\omega)\mu(\dee \omega)-\int_{\Omega}e(\omega)\mu(\dee \omega)-\sum_{j\in J} \bar{y}(j)+E(m)\in \mathcal{Z}(m). \nonumber
    \]
\end{enumerate}
\begin{claim}\label{derebatein}
$\bar{x}(\omega)\in D^{m}_{\omega}(\bar{x},\bar{y}(\mathcal{E}),\bar{p})$ for almost all $\omega\in \Omega$. 
\end{claim}
\begin{proof}
Note that $\bar{p}\cdot \bar{y}(j)=\bar{p}\cdot \bar{y}(\mathcal{E})(j)+\proj{k}(\bar{p})\cdot m^{(j)}$ for all $j\in J$. 
Thus, for all $\omega\in \Omega$, the budget set $B'_{\omega}(\bar{y}, \bar{p})$ for consumer $\omega$ of the economy $\mathcal{E}'$ can be written as: 
\[
\left\{z\in X_{\omega}: \bar{p}\cdot z \leq \bar{p}\cdot e(\omega)+\sum_{j\in J}\theta_{\omega j}\big(\bar{p}\cdot \bar{y}(\mathcal{E})(j)+\proj{k}(\bar{p})\cdot m^{(j)}\big)\right\}, \nonumber
\]
which is the same as the quota budget set $B^{m}_{\omega}(\bar{y}(\mathcal{E}),\bar{p})$ of the economy $\mathcal{E}$. As $P_{\omega}(\bar{x},\bar{y}(\mathcal{E}),\bar{p})=P'_{\omega}(\bar{x},\bar{y},\bar{p})$ for all $\omega\in \Omega$, the quota demand set $D'_{\omega}(\bar{x},\bar{y}, \bar{p})$ for consumer $\omega$ of the economy $\mathcal{E}'$ is the same as the quota demand set $D^{m}_{\omega}(\bar{x},\bar{y}(\mathcal{E}),\bar{p})$ of the economy $\mathcal{E}$. We conclude that $\bar{x}(\omega)\in D^{m}_{\omega}(\bar{x},\bar{y}(\mathcal{E}),\bar{p})$ for almost all $\omega\in \Omega$.
\end{proof}
By \cref{derebatein},  $(\bar{x}, \bar{y}(\mathcal{E}), \bar{p})$ is a quota equilibrium for $\mathcal{E}$.
\end{proof}

% Combining \cref{thm_makarov_nc} and \cref{interioritylm}, we have the following result: 

% \begin{theorem}\label{finitemainexpd}
% Let $\mathcal{E}=\{(X_t, R_t, P_t, e_t, \theta_t)_{t\in T}, (Y_j)_{j\in J}, \mu, \mathcal{Z}\}$ be a weighted production economy. Suppose, for each $t\in T$, each $j\in J$:
% \begin{enumerate}
%     \item $X_t$ is  closed and  convex;
%     \item $\bar{Y}$ is closed, convex, and  $\bar Y \cap (-\bar Y)=\{0\}=\bar Y \cap \NNReals^{\ell}$, where $\bar Y =\left\{\sum_{j\in J} {y}(j): y\in Y\right\}$; 
%     \item $\mathcal{Z}$ is closed;
%     \item $P_t: \prod_{t\in T}\NNReals^{\ell}\times Y\times \Delta\to \mathcal{P}_{H}$;
%     \item $e_t\in X_t-\sum_{j\in J}\theta_{tj}Y_j$;
%     %\item For each $(x,y,p)\in \cO\times \Delta'$ with $x_t\in X_t$, $p\neq 0$, there exists $z\in X_t$ such that $(z, x_t)\in P_{t}(x,y,p)$.
%     \item \label{satiate}For each $(x,y)\in \cO$ with $x_t\in X_t$, there exists $u\in X_t$ such that $(u, x_t)\in \bigcap_{p\in \Delta\cap \mathcal{Z}^{0}}P_{t}(x,y,p)$.
% \end{enumerate}
% Then:
% \begin{itemize}
%     \item $\mathcal{E}$ has a $\mathcal{Z}$-disposal quasi-equilibrium $(\bar{x}, \bar{y}, \bar{p})$;
%     \item If $e_t\in \mathrm{int}(X_t - \sum_{j\in J}\theta_{tj}Y_j)$ for all $t\in T$, then $(\bar{x}, \bar{y}, \bar{p})$ is a $\mathcal{Z}$-disposal equilibrium. 
% \end{itemize}
% \end{theorem}

\end{document}